%% file: arXiv-v3.tex
\documentclass[aps,prl,groupedaddress,superscriptaddress,twocolumn,notitlepage,bibnotes]{revtex4-1}

\input{def}
\usepackage{pgfplots}
\pgfplotsset{width=10cm,compat=1.9}
\usetikzlibrary{decorations.markings}

\usepackage{yfonts}
\usepackage[normalem]{ulem}
\usepackage{soul,xcolor}
\setstcolor{red}

\newcommand{\dtext}[2]{\left\|#1 - #2\right\|_{\raisemath{-1pt}{\diamond}}^{\raisemath{-1.5pt}{H,E}}}

\begin{document}

\title{Energy-constrained discrimination of unitaries, quantum speed limits and a Gaussian Solovay--Kitaev theorem}

\author{Simon Becker}
\email{simon.becker@damtp.cam.ac.uk}
\affiliation{Department of Applied Mathematics and Theoretical Physics, Centre for Mathematical Sciences, University of Cambridge, Cambridge CB3 0WA, United Kingdom}

\author{Nilanjana Datta}
\email{n.datta@damtp.cam.ac.uk}
\affiliation{Department of Applied Mathematics and Theoretical Physics, Centre for Mathematical Sciences, University of Cambridge, Cambridge CB3 0WA, United Kingdom}

\author{Ludovico Lami}
\email{ludovico.lami@gmail.com}
\affiliation{Institut f\"{u}r Theoretische Physik und IQST, Universit\"{a}t Ulm, Albert-Einstein-Allee 11, D-89069 Ulm, Germany}

\author{Cambyse Rouz\'{e}}
\email{rouzecambyse@gmail.com}
\affiliation{Zentrum Mathematik, Technische Universit\"{a}t M\"{u}nchen, 85748 Garching, Germany}

\begin{abstract}
We investigate the energy-constrained (EC) diamond norm distance between unitary channels acting on possibly infinite-dimensional quantum systems, and establish a number of results. Firstly, we prove that optimal EC discrimination between two unitary channels does not require the use of any entanglement. Extending a result by Ac\'in, we also show that a finite number of parallel queries suffices to achieve zero error discrimination even in this EC setting. Secondly, we employ EC diamond norms to study a novel type of quantum speed limits, which apply to pairs of quantum dynamical semigroups. We expect these results to be relevant for benchmarking internal dynamics of quantum devices. Thirdly, we establish a version of the Solovay--Kitaev theorem that applies to the group of Gaussian unitaries over a finite number of modes, with the approximation error being measured with respect to the EC diamond norm relative to the photon number Hamiltonian.
\end{abstract}

\maketitle

\textbf{\em Introduction.}--- 
The task of distinguishing unknown objects is arguably a fundamental one in experimental science. Quantum state discrimination, one of the simplest examples of a problem of this sort, has gained a central role in the flourishing field of quantum information science. The optimal measurement for discriminating between two quantum states via quantum hypothesis testing was found by Holevo and Helstrom~\cite{Holevo1972-analogue, Holevo1973-statistical, Holevo1976, HELSTROM}. \mbox{Subsequent} fundamental contributions related to state discrimination include the operational interpretation of quantum relative entropy~\cite{Umegaki1962} and of a related entanglement measure via quantum generalisations of Stein's lemma \cite{Hiai1991, Ogawa2000, Brandao2010}, the identification of a quantum Chernoff bound for symmetric hypothesis testing \cite{Nussbaum2009, qChernoff, Audenaert2008}, and the discovery of quantum data hiding \cite{dh-original-1, dh-original-2, VV-dh, VV-dh-Chernoff, ultimate}.

While quantum states are simpler objects, quantum processes, or channels, are more fundamental~\footnote{For instance, states of a quantum system $A$ can be thought of as channels from the trivial system to $A$.}. The basic primitive in distinguishing them is that of binary channel discrimination: two distant parties, Alice and Bob, are granted access to one query of one of two channels $\cN$ and $\cM$, with a priori probabilities $p$ and $1-p$, and they have to guess which channel was chosen. The best strategy consists of Alice preparing a (possibly entangled) bipartite state $\ket{\Psi}_{AA'}$, sending the system $A$ through the noisy channel, and the auxiliary system (or ancilla) $A'$ through an ideal (noiseless) channel to Bob, who then performs state discrimination on the bipartite system $AA'$ that he receives.
When both $\cN$ and $\cM$ are unitary channels, however, the auxiliary system is \emph{not} needed~\cite{Aharonov1998} (c.f.~\cite[Theorem~3.55]{WATROUS}). Experimentally, this simplification is helpful, as it exempts us from using: (a)~an ancilla and entanglement; and (b)~an ideal side channel, which might be technologically challenging.

More insight into the channel distinguishability problem can be gained by looking at multi-query discrimination~\cite{Hayashi2009, Berta2018, Fang2020}. When the channels are unitary, a seminal result by Ac\'in states that perfect discrimination is possible with only a finite number of queries \cite{Acin2001, DAriano2001}, a phenomenon that has no analogue for states~\cite{Duan2009}. The same result can be achieved by using an adaptive strategy that requires no entanglement~\cite{Duan2007}.

It is common to assume that any arbitrary quantum operation can be employed for the discrimination task at hand. This is, however, often unrealistic, due to technological as well as physical limitations. This is the case e.g.~when the quantum states (respectively, the channels) to be discriminated are distributed among (respectively, connect) two parties who can only employ local operations assisted by classical communication. Such a restriction could severely hinder the discrimination power, both for states~\cite{dh-original-1, dh-original-2, VV-dh, VV-dh-Chernoff, ultimate} and for channels~\cite{Matthews2010, Duan2008}. 

Another example of physical restriction comes about, for instance, when one studies continuous-variable (CV) quantum systems, e.g.~collections of electromagnetic modes travelling along an optical fibre. This setting, which constitutes the basis of practically all proposed protocols for quantum communication, is of outstanding technological and experimental relevance~\cite{KLM, Braunstein-review, CERF, weedbrook12}. Accordingly, the theoretical study of CV quantum channels is a core area of quantum information~\cite{HOLEVO, BUCCO, HOLEVO-CHANNELS-2}. CV channel discrimination can be thought of as a fundamental primitive for benchmarking such channels.

When accessing a CV quantum system governed by a Hamiltonian $H$, one only has access to states $\rho$ with bounded mean energy $\tr[\rho H]\leq E$. This fundamentally unavoidable restriction motivates us to look into \emph{energy-constrained (EC) channel discrimination}~\cite{Shirokov2016, VV-diamond, Berta2018, Sharma2020}. In our setting, we separate the energy cost of manufacturing probes from that of measuring the output states~\cite{Navascues2014}, and only account for the former. This is justified operationally by thinking of the unknown channel (either $\cN$ or $\cM$) as connecting an EC client to a quantum computing server that has access to practically unlimited energy. In the above context, the figure of merit is the so-called EC diamond norm distance $\dtext{\cN}{\cM}$~\cite{PLOB, Shirokov2016, VV-diamond}.

In this paper, we (1)~study the EC diamond norm distance between unitary channels, and employ it to establish (2)~operationally meaningful quantum speed limits~\cite{Deffner2017} for experimentally relevant Hamiltonians, as well as (3)~a Solovay--Kitaev theorem~\cite{Kitaev1997, Dawson2006} for Gaussian (i.e.\ symplectic) unitaries. Our first result states that optimal EC discrimination of two unitary channels does not require any entanglement (Theorem~\ref{unentangled_thm}). This extends the analogous result for unconstrained discrimination~\cite[Theorem~3.55]{WATROUS}. In the same setting, we then generalise Ac\'in's result~\cite{Acin2001}, proving that a finite number of parallel queries suffices to achieve zero error (Theorem~\ref{eventual_discrimination_thm}).

We then employ the EC diamond norm distance to quantify in an operationally meaningful way the speed at which time evolutions under two different Hamiltonians drift apart from each other (Theorem~\ref{thm:favard}). Our result amounts to a quantum speed limit~\cite{Deffner2017} that applies to a more general setting than previously investigated~\cite{Mandelstam1945, Mandelstam1991, Bhattacharyya1983, Pfeifer1993, Margolous1998, V-2003b, Levitin2009, Pires2016, Campaioli2018, Okuyama2018, Okuyama2018comment, Bukov2019, Sun2019, Simon-Nila}, namely, that involving two different unitary groups. As a special case, we study evolutions induced by quadratic Hamiltonians on a collection of harmonic oscillators (Corollary~\ref{propgaussexample}). Analogous estimates are then given for the case in which one of the two channels models an open quantum system (Theorem~\ref{thm:open})~\cite{delCampo2013}.

Our last result is a Solovay--Kitaev theorem~\cite{Kitaev1997, Dawson2006} for Gaussian unitaries (Theorem~\ref{theosollovaykit}). It states that any finite set of gates generating a dense subgroup of the symplectic group can be used to construct short gate sequences that approximate well, in the EC diamond norm corresponding to the photon number Hamiltonian, any desired Gaussian unitary. The significance of our result rests on the compelling operational interpretation of the EC diamond norm in terms of channel discrimination: the action of the constructed gate will be almost indistinguishable from that of the target on all states with a certain maximum average photon number. 

\textbf{\em The setting.}--- Quantum states on a Hilbert space $\cH$ are represented by density operators, i.e.\ positive trace-class operators with trace one, on $\cH$. Quantum channels are modelled by completely positive and trace preserving (CPTP) maps acting on the space of trace-class operators on $\cH$. A \emp{Hamiltonian} on $\cH$ is a densely defined self-adjoint operator $H$ whose spectrum $\spec(H)$ is bounded from below. Up to re-defining the ground state energy, we can assume that $\min \spec(H)=0$, in which case we call $H$ \emp{grounded}. In what follows, for a pure state $\ket{\psi}\in \cH$, we will denote with $\psi\coloneqq \ketbra{\psi}$ the corresponding density matrix.

CV quantum systems, i.e.\ finite collections of harmonic oscillators, or modes, are central for applications~\cite{HOLEVO, BUCCO}. The Hilbert space of an $m$-mode system is formed by all square-integrable functions on $\RR^m$, and is denoted by $\cH_m \coloneqq L^2\left(\RR^m\right)$. The creation and annihilation operators corresponding to the $j^{\text{th}}$ mode ($j=1,\ldots, m$) will be denoted by $a_j^\dag$ and $a_j$, respectively. They satisfy the \emph{canonical commutation relations} (CCRs) $[a_j, a_k^\dag] = \delta_{jk}$. In the (equivalent) real picture, one defines the \emph{position} and \emph{momentum} operators $x_j \coloneqq \frac{a_j+a^\dag_j}{\sqrt2}$ and $p_j \coloneqq \frac{a_j-a^\dag_j}{\sqrt2\, i}$, organised in the vector $R\coloneqq (x_1,p_1,\ldots, x_m, p_m)^\intercal$. The CCRs now read $\left[R,R^\intercal\right] = i\Omega_m$, with $\Omega_m \coloneqq \left( \begin{smallmatrix} 0 & 1 \\ -1 & 0 \end{smallmatrix}\right)^{\oplus m}$. \emph{Gaussian unitaries} are products of exponentials $e^{-\frac{i}{2} R^\intercal Q R}$, where $Q$ is an arbitrary $2m\times 2m$ symmetric matrix, and $\frac12 R^\intercal Q R$ is called a \emph{quadratic Hamiltonian}. Gaussian unitaries are in one-to-one correspondence with symplectic matrices via the relation $U_S \leftrightarrow S$ defined by $U_S^\dag R_j U_S = \sum_k S_{jk} R_k$. The corresponding unitary channel will be denoted with $\cU_S(\cdot)\coloneqq U_S(\cdot) U_S^\dag$. Recall that a $2m\times 2m$ real matrix $S$ is called \emph{symplectic} if $S\Omega_m S^\intercal = \Omega_m$, and that symplectic matrices form a group, hereafter denoted by $\symp_{2m}(\RR)$~\cite{GOSSON}.

The energy cost of a channel discrimination protocol comes from two main sources: first, the preparation of the probe state to be fed into the unknown channel, and, second, the subsequent quantum measurement, which inescapably requires energy to be carried out~\cite{Navascues2014}. In this paper we consider only the first contribution, i.e.\ the energy cost of the probe. Operationally, we can separate the above two contributions by considering the following setting. An unknown channel, either $\cN_{A\to B}$ (with a priori probability $p$) or $\cM_{A\to B}$ (with a priori probability $1-p$) connects two distant parties, Alice (the sender) and Bob (the receiver). We assume that Alice's equipment only allows for the preparation of probe states with an average energy at most $E$, as measured by some positive Hamiltonian $H_A\geq 0$ on the input system. No such restriction is placed on Bob, who can carry out any measurement he desires, and whose task is that of guessing the channel. We can further distinguish two possibilities: (i)~Alice is limited to preparing states $\rho_A$ on the input system $A$, to be sent to Bob via the unknown channel; or (ii)~she can prepare a (possibly entangled) state $\rho_{AA'}$, where $A'$ is an arbitrary ancilla, and send also $A'$ to Bob via an ideal (noiseless) channel. The energy constraint reads $\tr[\rho_A H_A]\leq E$, where in case~(ii) we set $\rho_A\coloneqq \tr_{A'} \rho_{AA'}$. The error probability corresponding to~(ii) takes the form $P_{\raisemath{-1pt}{e}}^{\raisemath{-1.5pt}{H,E}}(\cN,\cM;p)=\frac12 \left( 1 - \dtext{p\cN\!}{\!(1\!-\!p)\cM}\right)$, where for a superoperator $\cL_A$ that preserves self-adjointness the \emph{EC diamond norm} is defined by
\bb
\left\| \cL_A \right\|_{\diamond}^{H,E} = \sup_{\substack{\ket{\Psi}_{AA'}: \\[.2ex] \tr \Psi_{\!A}\! H_{\!A} \leq E}} \left\| \left( \cL_A \otimes \Id_{A'}\right)(\Psi_{AA'}) \right\|_1\, ,
\label{en_con_diamond}
\ee
where $\|\cdot\|_1$ is the trace norm, while the supremum is over all states $\ket{\Psi}_{AA'}$ on $AA'$, with $A'$ being an ancilla, whose reduced state on $A$ has energy bounded by $E$. A similar expression but without $A'$ holds in setting~(i).

\textbf{\em Results.}--- Throughout this section we discuss our main findings. Complete proofs as well as additional technical details can be found in the Supplemental Material~\footnote{See the Supplemental Material, which contains the references~\cite{HALL, Shirokov2018, AuYeung1979, AuYeung1983, Binding1985, HJ2, Toeplitz1918, Hausdorff1919, Si15, A02, HV2, HV, A02a, Arnold2004, Vacchini2002, CGQ03, CBPZ, assisted-Ryuji, pramana, heinosaari2009semigroup, reed1975ii, aubrun2017alice, BHATIA-MATRIX, Kitaev1997b, NC, harrow2001quantum, harrow2002efficient, aharonov2007polynomial, kuperberg2015hard, bouland2017trading} for complete proofs of the results discussed in the main text.}.

\emph{(1) EC discrimination of unitaries.} Our first result states that the above settings~(i) and~(ii) are equivalent in the case of two unitary channels. This generalises the seminal result of Aharonov et al.~\cite{Aharonov1998} (cf.~\cite[Theorem~3.55]{WATROUS}), and implies that optimal EC discrimination of unitaries can be carried out without the use of any entanglement.

\begin{thm}\label{unentangled_thm}
Let $U,V$ be two unitaries acting on a Hilbert space of dimension $\dim \cH\geq 3$, and call $\mathcal{U}(\cdot) \coloneqq U(\cdot)U^\dag$, $\mathcal{V}(\cdot)\coloneqq V(\cdot)V^\dag$ the associated channels. Let $H$ be a grounded Hamiltonian, and fix $E>0$. Then
\bb
\begin{aligned}
\left\| \cU - \cV \right\|_{\diamond}^{H,E} 
&= \sup_{\braket{\psi|H|\psi}\leq E} \left\| \left( \cU - \cV\right) (\psi) \right\|_1 \\
&= 2 \sqrt{1 - \inf_{\braket{\psi|H|\psi}\leq E} \left| \braket{\psi|U^\dag V |\psi}\right|^2} \, .
\end{aligned}
\label{unentangled}
\ee
In other words, in this case the supremum in \eqref{en_con_diamond} can be restricted to unentangled pure states. 
\end{thm}

The above result can be used to estimate the EC diamond norm distance between \emp{displacement channels}. These are defined for $z\in \RR^{2m}$ by $\cD_z(\cdot) \coloneqq \D(z)(\cdot) \D(z)^\dag$, where $\D(z)\coloneqq e^{- i \sum_j (\Omega_m z)_j R_j}$. Letting $N\coloneqq \sum_j a_j^\dag a_j$ be the total photon number Hamiltonian, one has that
\bb
\begin{aligned}
\sqrt{1\!-\!e^{-\|z-w\|^2 f(E)^2}}\! \leq&\ \frac12 \left\|\cD_{z} - \cD_{w}\right\|_{\diamond}^{N\!,E} \\
\leq&\ \sin\left(\! \min\left\{\|z\!-\!w\| f(E),\, \frac{\pi}{2}\right\}\right) , \\
f(E) \coloneqq &\ \frac{1}{\sqrt2} \left( \sqrt{E}+\sqrt{E+1} \right) .
\end{aligned}
\label{displacements_EC_diamond_distance}
\ee
Using the structure of the symplectic group, we also obtain the following upper bound for the difference of two symplectic unitaries: given $S,S'\in\operatorname{Sp}_{2m}(\RR)$,
\bb
\begin{aligned}
\frac12 \|\cU_S-\cU_{S'}\|_\diamond^{N\!,E} \! \le&\, \sqrt{\left({\sqrt{6}\!+\!\sqrt{10}\!+\!5\sqrt{2}m}\right)(E+1)} \\
&\ g\left(\|(S')^{-1}S\|_\infty\right) \sqrt{\|(S')^{-1}S-I\|_2}\, , \\
g(x) \!\coloneqq&\, \sqrt{\frac{\pi}{x+1}} + \sqrt{2x}\, ,
\end{aligned}
\ee
where $\|\cdot\|_\infty$ and $\|\cdot\|_2$ denote the operator norm and the Hilbert--Schmidt norm, respectively. We can also exploit Theorem~\ref{unentangled_thm} to immediately extend a celebrated result by Ac\'in~\cite{Acin2001} (see also~\cite{Duan2007, Duan2009}), and establish that even in the presence of an energy constraint (which is particularly relevant in the case of unitaries acting on CV quantum systems), a finite number of parallel queries achieves zero-error discrimination.

\begin{thm} \label{eventual_discrimination_thm}
In the setting of Theorem~\ref{unentangled_thm}, there exists a positive integer $n$ such that $n$ parallel uses of $\cU$ and $\cV$ can be discriminated perfectly using inputs of finite total energy $E$, i.e.
\bb
\left\| \cU^{\otimes n} - \cV^{\otimes n} \right\|_\diamond^{H_{(n)},\,E}=2\, ,
\ee
where $H_{(n)}\coloneqq \sum_{j=1}^n H_j$ is the $n$-copy Hamiltonian, and $H_j\coloneqq  I \otimes \cdots I \otimes H \otimes I \cdots \otimes I$, with the $H$ in the $j^{\text{th}}$ location.
\end{thm}

\emph{(2) Quantum speed limits.} Our first application deals with the problem of quantifying the relative drift caused by two different unitary dynamics on a quantum system. This may be important, for instance, in benchmarking internal Hamiltonians of quantum devices. 

In what follows, our findings are generally presented in the form of an upper bound on the EC diamond norm distance between time evolution channels. This is an alternative yet completely equivalent reformulation of a quantum speed limit. To recover the standard one~\cite{Deffner2017}, one has to turn the inequality around and recast it as a lower bound on the time taken to reach a certain prescribed distance~\cite{Note2}.
Our first result extends previous findings by Winter~\cite[Theorem~6]{VV-diamond} and some of us~\cite[Proposition~3.2]{Simon-Nila} by tackling the case of two different unitary groups.

\begin{thm}
Let $H,H'$ be self-adjoint operators. Without loss of generality, assume that $0$ is in the spectrum of $H$. Let the `relative boundedness' inequality
\bb
\left\| (H-H')\ket{\psi}\right\| \leq \alpha \left\|H\ket{\psi}\right\| + \beta
\label{relatbound}
\ee
hold for some constants $\alpha,\beta>0$ and for all (normalised) states $\ket{\psi}$. Then the unitary channels 
\bb
\cU_t(\cdot)\coloneqq e^{-iHt}(\cdot) e^{iHt}, \quad \cV_t(\cdot)\coloneqq e^{-iH't}(\cdot) e^{iH't}
\label{unitary_groups}
\ee
satisfy the following: for all $t\geq 0$ and $E>0$,
\bb
\left\| \cU_{t} - \cV_{t} \right\|_{\diamond}^{|H|,E} \le 2 \sqrt{2} \sqrt{\alpha E t} + \sqrt{2}\beta t.
\label{drift}
\ee
\label{thm:favard}
\end{thm}

Let us note that \eqref{drift} admits a simple reformulation in terms of the Loschmidt echo operator $M_t \coloneqq e^{iH't}e^{-iHt}$~\cite{Gorin2006, Note2}. The relative boundedness condition \eqref{relatbound} is not merely an artefact of the proof, and is there to ensure that low energy eigenvectors of $H$ do not have very high energies relative to $H'$, which would trivialise the bound \eqref{drift}. The estimate in~\eqref{drift} can be shown to be optimal up to multiplicative constants: in general, the diffusive term proportional to $\sqrt{t}$ cannot be removed even for very small times~\cite[\S~III.B]{Note2}.

A special case of Theorem~\ref{thm:favard} that is particularly relevant for applications is that of two quadratic Hamiltonians on a collection of $m$ harmonic oscillators, or modes.

\begin{cor}\label{propgaussexample}
On a system of $m$ modes, consider the two Hamiltonians $H = \sum_{j=1}^m d_j a^\dag_j a_j$ and $H' = \sum_{j,k=1}^m \left( X_{jk} a_j^\dag a_k+Y_{jk} a_j a_k+Y_{jk}^{*} a_j^\dag a_k^\dag \right)$, where $d_j>0$ for all $j$, and $X,Y$ are two $m\times m$ matrices, with $X$ Hermitian. Then the corresponding unitary channels~\eqref{unitary_groups} satisfy~\eqref{drift} for all $t\geq 0$ and $E> 0$, with
\begin{equation}
\begin{split}
\alpha &= \Vert D^{-1} \Vert \left(\sqrt{\tfrac{3}{2}} \Vert X-D \Vert_2 + \left(1+\sqrt{\tfrac{3}{2}}\right) \Vert Y \Vert_2\right) , \\
\beta &= \tfrac{m-1}{\sqrt{2}}  \Vert X-D \Vert_2 +\sqrt{\tfrac{(2m+1)^2}{2}+2m^2} \Vert Y \Vert_2\, ,
\end{split}
\end{equation}
where $D_{jk} \coloneqq d_j \delta_{jk}$.
\end{cor}

We now look at the more general scenario where the discrimination is between a closed-system unitary evolution and an open-system quantum dynamics. We expect this task to be critical e.g.\ in benchmarking quantum memories, where the effects of external interactions are detrimental and must be carefully controlled. Open quantum systems are described by \emph{quantum dynamical semigroups} (QDSs)~\cite{ENGEL, ENGEL-SHORT}, i.e.\ families of channels $\left( \Lambda_t\right)_{t\geq 0}$ that (i)~obey the semigroup law, $\Lambda_{t+s} = \Lambda_t \circ \Lambda_s$ for $t,s\geq 0$, and (ii)~are strongly continuous, in the sense that $\lim_{t\to 0^+} \left\|\Lambda_t(\rho) - \rho\right\|_1=0$ for all $\rho$. QDSs take the form $\Lambda_t = e^{t\cL}$, where the generator $\cL$ is assumed to be of Gorini--Kossakowski--Lindblad--Sudarshan (GKLS) type ~\cite{G-K-Lindblad-S, Gorini-Kossakowski-L-Sudarshan, Davies1977} and acts on an appropriate dense subspace of the space of trace class operators as
\bb
\cL(X) = -i\!\left[H,X\right] + \frac12\! \sum_\ell \left( 2 L_\ell X L_\ell^\dag\! -\! L_\ell^\dag L_\ell X\! -\! X L_\ell^\dag L_\ell\right) .
\ee
Here, $H$ is the internal Hamiltonian, while the \emph{Lindblad operators} $L_\ell$ ($\ell=1,2,\ldots$) model dissipative processes. In our approach these can be unbounded, and hence our results significantly generalise previous works on quantum speed limits in open systems~\cite{delCampo2013}.

\begin{thm} \label{thm:open}
Let $H$ be a self-adjoint operator with $0$ in its spectrum, and set $\cU_t(\cdot) \coloneqq e^{-iHt}(\cdot) e^{iHt}$. Let $\left( \Lambda_t\right)_{t\geq 0}$ be a QDS whose generator $\cL$ is of GKLS-type and satisfies the relative boundedness condition
\bb
\frac12 \left\| \sumno_\ell L_\ell^\dag L_\ell \ket{\psi} \right\| \leq \alpha \left\| H \ket{\psi} \right\| + \beta
\label{relatbound_Lindbladian}
\ee
for all (normalised) states $\ket{\psi}$, where $\beta\geq 0$ and $0\leq \alpha<1$ are two constants. Then it holds that
\bb
\left\|\cU_t - \Lambda_t \right\|_\diamond^{|H|,E} \leq 4\left( \sqrt{\sqrt{2}\alpha E t}  +  \beta t\right)
\label{drift_Lindbladian}
\ee
for all $t\geq 0$ and $E> 0$.
\end{thm}

Once again, the role of condition \eqref{relatbound_Lindbladian} is that of ensuring that the Lindblad operators do not make low energy levels decay too rapidly, an effect that we could exploit to design a simple discrimination protocol with a small energy budget. We now demonstrate the applicability of our result by looking at the example of \emph{quantum Brownian motion} \cite{Vacchini2002, Arnold2004}. Consider a single quantum particle in one dimension, subjected to a harmonic potential and to a diffusion process. The Hilbert space is $\cH_1=L^2(\RR)$; we set $H=\frac12 (x^2+p^2)$ and $L_\ell = \gamma_\ell x + i \delta_\ell p$ ($\ell=1,2$), where $p\coloneqq -i \frac{d}{dx}$ is the momentum operator, and $\gamma_\ell, \delta_\ell \in {\mathbb{C}}$. In this case~\eqref{relatbound_Lindbladian} is satisfied e.g.\ with $\alpha = \left(|\gamma_1| + |\delta_1|\right)^2 + \left(|\gamma_2| + |\delta_2|\right)^2$, provided that the right-hand side is smaller than $1$, and $\beta = |\gamma_1| |\delta_1| + |\gamma_2||\delta_2| + \kappa$, where $\kappa=0.2047$ is a constant~\cite{Note2}. Therefore,~\eqref{drift_Lindbladian} yields an upper estimate on the operational distinguishability between closed and open dynamics for given waiting time and input energy.

\emph{(3) A Gaussian Solovay--Kitaev theorem.} The celebrated Solovay--Kitaev theorem~\cite{Kitaev1997, Dawson2006} is a fundamental result in the theory of quantum computing. In layman's terms, it states that any finite set of quantum gates that generates a dense subgroup of the special unitary group is capable of approximating any such desired unitary by means of short sequences of gates. In practice, many 
of the elementary gates that form the toolbox of CV platforms for quantum computing~\cite{Gottesman2001, KLM} are modelled by Gaussian unitaries. Therefore, a Gaussian version of the Solovay--Kitaev theorem is highly desirable. In establishing our result, we measure the approximation error for gates on an $m$-mode quantum system by means of the operationally meaningful EC diamond norm distance relative to the total photon number Hamiltonian $N = \sum_{j=1}^m a^\dag_j a_j$.

\begin{thm}\label{theosollovaykit}
Let $m\in\NN$, $r>0$, $E>0$ and define $\widetilde{\operatorname{Sp}}_{2m}^r(\RR)$ to be the set of all symplectic transformations $S$ such that $\|S\|_\infty\le r$. Then, given a set $\mathcal{G}$ of gates that is closed under inverses and generates a dense subset of $\widetilde{\operatorname{Sp}}_{2m}^r(\RR)$, for any symplectic transformation $S\in \widetilde{\operatorname{Sp}}_{2m}^r(\RR)$ and every $0<\delta$, there exists a finite concatenation $S'$ of $\operatorname{poly}(\log\delta^{-1})$ elements from $\mathcal{G}$, which can be found in time $\operatorname{poly}(\log\delta^{-1})$ and such that 
\bb    
\|\mathcal{U}_S-\mathcal{U}_{S'}\|_{\diamond}^{N,E}\le F(m)G(r)\sqrt{E+1}\sqrt{\delta}\, 
  \,,\label{approximation}
\ee
where $\mathcal{U}_S(\cdot)\coloneqq U_S (\cdot) U_S^\dag$, and 
\begin{align*}
&F(m)\coloneqq 2\sqrt{\,\sqrt{2m}\,({\sqrt{6}+\sqrt{10}+5\sqrt{2}m})}\,,\\
&G(r)\coloneqq \Big(\sqrt{\pi}+
    \sqrt{2}(r+2)\Big)\,\sqrt{(r+2)}\,.
    \end{align*}
\end{thm}

The above result guarantees that any Gaussian unitary can be approximated with a relatively short sequence of gates taken from our base set. Note that the sequence length increases with both the squeezing induced by $S$ (quantified by the parameter $\|S\|_\infty$) and the energy threshold $E$. Theorem~\ref{theosollovaykit} also guarantees that finding the relevant gate sequence is a computationally feasible task, thus bolstering the operational significance of the result.
Finally, in the Supplemental Material~\cite{Note2} we show that sets of the form $\mathcal{G}=\mathcal{K}\cup \{S\}$, where $\mathcal{K}$ generates a dense subgroup of the passive Gaussian unitary group and $S$ is an arbitrary non-passive Gaussian unitary, satisfy the denseness assumption of Theorem~\ref{theosollovaykit}.

\textbf{\em Conclusions.}--- We investigated the EC diamond norm distance between channels, which has a compelling operational interpretation in the context of EC channel discrimination. For the case of two unitary channels, we showed that optimal discrimination can be carried out without using any entanglement, and with zero error upon invoking finitely many parallel queries. An open question here concerns the possibility of obtaining the same result by means of adaptive rather than parallel strategies. This is known to be possible in the finite-dimensional, energy-unconstrained scenario~\cite{Duan2007}. 

We then studied some problems where the EC diamond norm can be employed to quantify in an operationally meaningful way the distance between quantum operations. We provided quantum speed limits that apply to the conceptually innovative setting where one compares two different time evolution (semi-)groups, instead of looking at a single one, as previously done.

Finally, we established a Gaussian version of the Solovay--Kitaev theorem, proving that any set of Gaussian unitary gates that is sufficiently powerful to be capable of approximating any desired Gaussian unitary can do so also efficiently, i.e.\ by means of a relatively small number of gates. Our result bears a potential impact on the study of all those quantum computing architectures that rely on optical platforms.

\textbf{\em Acknowledgements.}--- All authors contributed equally to this paper. LL acknowledges financial support from the Alexander von Humboldt Foundation. SB gratefully acknowledges support by the EPSRC grant EP/L016516/1 for the University of Cambridge CDT, the CCA.

\bibliographystyle{unsrt}
\bibliography{Unified-Biblio}


\clearpage

\onecolumngrid
\begin{center}
\vspace*{\baselineskip}
{\textbf{\large Supplemental Material: Energy-constrained discrimination of unitaries, quantum speed limits and a Gaussian Solovay--Kitaev theorem}}
\end{center}

\renewcommand{\theequation}{S\arabic{equation}}
\renewcommand{\thethm}{S\arabic{thm}}
\setcounter{equation}{0}
\setcounter{thm}{0}
\setcounter{figure}{0}
\setcounter{table}{0}
\setcounter{section}{0}
\setcounter{page}{1}
\makeatletter

\setcounter{secnumdepth}{2}

\section{Notations and definitions}

\subsection{Operators and norms}

Given a separable Hilbert space $\cH$, we denote by $\cB(\cH)$ the space of bounded linear operators on $\cH$, and by $\cT_p(\cH)$, the \emp{Schatten $p$-class}, which is the Banach subspace of $\cB(\cH)$ formed by all bounded linear operators whose Schatten $p$-norm, defined as $\|X\|_{p}=\left(\tr|X|^p\right)^{1/p}$,  is finite. Henceforth, we refer to $\cT_1(\cH)$ as the set of \emp{trace class} operators. The set of quantum states (or density matrices), that is positive semi-definite operators $\rho \in \cT_1(\cH)$ of unit trace, is denoted by $\cD(\cH)$. The Schatten $1$-norm, $\|\cdot\|_1$, is the {trace norm}, and the corresponding induced distance (e.g.\ between quantum states) is the {trace distance}. Note that the Schatten $2$-norm, $\|\cdot\|_2$, coincides with the \emp{Hilbert--Schmidt norm}.

We denote by $\mathbb{M}_{2m}(\RR)$ the set of $2m \times 2m$ real matrices, and by $\symp_{2m}(\RR)$, the set of symplectic matrices in $\mathbb{M}_{2m}(\RR)$, i.e.~matrices $S \in\mathbb{M}_{2m}(\RR)$ satisfying the condition $S\Omega_{m}S^\intercal = \Omega_{m}$, where $\Omega_{m}$ denotes the $2m\times 2m$ commutation matrix:
\begin{equation} \label{comm}
\Omega_{m} \coloneqq \begin{pmatrix} 0 & 1 \\ -1 & 0 \end{pmatrix}^{\oplus m}\, ,
\end{equation}
Any symplectic matrix $S$ has determinant equal to one and is invertible with $S^{-1} \in \operatorname{Sp}_{2m}(\RR)$. Hence, $\symp_{2m}(\RR)$ is a subgroup of the special linear group $\operatorname{SL}_{2m}(\mathbb{R})$.

For a pair of positive semi-definite operators, $A,B$ with domains $\dom(A),\dom(B) \subseteq \cH$, $A\geq B$ if and only if $\dom\left(A^{1/2}\right)\subseteq \dom\left(B^{1/2}\right)$ and $\left\|A^{1/2}\ket{\psi}\right\|^2\geq \left\|B^{1/2}\ket{\psi}\right\|^2$ for all $\ket{\psi}\in \dom\left(A^{1/2}\right)$. 
If $\rho$ is a quantum state with spectral decomposition $\rho=\sum_i p_i \ketbra{\phi_i}$, and $A$ is a positive semi-definite operator, the \emp{expected value} of $A$ on $\rho$ is defined as
\bb
\tr[\rho A]\coloneqq \sum_{i:\, p_i>0} p_i \left\|A^{1/2}\ket{\phi_i}\right\|^2 \in \RR_+\cup \{+\infty\}\, ;
\label{expected positive}
\ee
here we use the convention that $\tr[\rho A]=+\infty$ if the above series diverges or if there exists an index $i$ for which $p_i>0$ and $\ket{\phi_i}\notin \dom\left(A^{1/2}\right)$. This definition can be extended to a generic densely defined self-adjoint operator $A$ on $\cH$, by considering its decomposition $A=A_+-A_-$ into positive and negative parts, with $A_\pm$ being positive semi-definite operators with mutually orthogonal supports. The operator $A$ is said to have a \emp{finite expected value on $\rho$} if $(i)$ $\ket{\phi_i}\in \dom\big(A_+^{1/2}\big)\cap \dom\big(A_-^{1/2}\big)$ for all $i$ for which $p_i>0$, and $(ii)$ the two series $\sum_i p_i \big\|A_\pm^{1/2} \ket{\phi_i}\big\|^2$ both converge. In this case, the following quantity is called the \emp{expected value} of $A$ on $\rho$:
\bb
\tr[\rho A]\coloneqq \sum_{i:\, p_i>0} p_i \left\|A_+^{1/2} \ket{\phi_i}\right\|^2 - \sum_{i:\, p_i>0} p_i \left\|A_-^{1/2} \ket{\phi_i}\right\|^2
\label{expected}
\ee
Obviously, for a pair of operators $A,B$ satisfying $A\geq B$, we have that $\tr[\rho A]\geq \tr[\rho B]$.

Let $A$ be an (unbounded) operator $A$ on some Banach space ${{X}}$, with domain $\dom(A)$. Such an operator is called closed if its {\em{graph}}, that is $\left\{(\ket{x},A\ket{x}); \ket{x} \in \dom(A) \right\} \subset X \times X,$ is closed. The spectrum of a closed operator $A$ is defined as the set~\cite[Definition~9.16]{HALL}
\[\spec(A)\coloneqq \left\{ \lambda \in \mathbb C:\, \lambda I-A \text{ is not bijective} \right\}.\]
Henceforth, we often suppress the identity operator $I$ in the expression $(\lambda I - A)$ for notational simplicity.
Here, a closed operator $B$ is said to be not bijective if there exists no bounded operator $C$ with the property that: (i)~for all $\ket{\psi}\in \cH$, one has that $K\ket{\psi}\in \dom(B)$, and moreover $B K\ket{\psi}=\ket{\psi}$; and (ii)~for all $\ket{\psi}\in \dom(B)$, it holds that $KB\ket{\psi}=\ket{\psi}$. We remind the reader that the spectrum of a self-adjoint positive operator is a closed subset of the positive real half-line~\cite[Proposition~9.20]{HALL}.

A quantum channel with input system $A$ and output system $B$ is any completely positive, trace-preserving (CPTP) linear map $\cN:\cT_1(\cH_A)\to\cT_1(\cH_B)$, where $\cH_A, \cH_B$ are the Hilbert spaces corresponding to $A,B$, respectively. Our input Hilbert spaces $\cH_A$ are often equipped with Hamiltonians, which we define as follows.

\begin{Def} \label{def:Hamiltonian}
A \emp{Hamiltonian} on a Hilbert space $\cH$ is a self-adjoint positive operator $H\geq 0$ on $\cH$ with dense domain $\dom(H)\subseteq \cH$. A Hamiltonian $H$ is said to be \emp{grounded} if its ground state energy is zero, in formula $\min \spec(H) = 0$.
\end{Def}


Next, given a superoperator $\cL:\cT_1(\cH_A)\to\cT_1(\cH_B)$ that preserves self-adjointness, we introduce the family of \emp{energy-constrained diamond norms} (or simply \emp{EC diamond norms}) \cite{PLOB, Shirokov2018, VV-diamond}
\begin{equation}
\left\| \cL_A \right\|_{\diamond}^{H,E} \coloneqq \sup_{\rho_{AR}\in \cD(\cH_{AR}):\, \tr [\rho_A H_{A}] \leq E} \left\| \left( \cL_A \otimes I_{R}\right)(\rho_{AR}) \right\|_1 = \sup_{\ket{\Psi}_{AR}:\, \tr [\Psi_A H_{A}] \leq E} \left\| \left( \cL_A \otimes I_{R}\right)(\Psi_{AR}) \right\|_1\, ,
\label{en_con_diamond_SM}
\end{equation}
where $E\ge 0$, $R$ is an arbitrary auxiliary quantum system (ancilla), and the Hilbert space associated with the composite $AR$ is simply $\cH_{AR} \coloneqq \cH_A \otimes \cH_R$. The rightmost equality in \eqref{en_con_diamond_SM} follows by restricting the supremum to pure states $\rho_{AR} = \ketbra{\Psi}_{AR}$, which is possible due to purification and the data processing inequality.

\subsection{Phase-space formalism}
In this paper, given $m\in\NN$, we are concerned with the Hilbert space $\cH_m\coloneqq L^2(\RR^m)$ of a so-called $m$-mode oscillator, which is the space of square-integrable functions on $\RR^m$. We denote by $x_j$ and $p_j$ the canonical position and momentum operators on the $j^{\text{th}}$ mode. The $j^{\text{th}}$ creation and annihilation operators $a_j=(x_j-ip_j)/\sqrt{2}$ and $a_j^\dag=(x_j+ip_j)/\sqrt{2}$ satisfy the well-known \emp{canonical commutation relations} (CCR):
\begin{align}
\label{CCRlie}
    [a_j,a_k]=0\,,\qquad [a_j,a_k^\dagger]=\delta_{jk}I\,,
\end{align}
where $I$ denotes the identity operator on $\cH_m$. In terms of the vector of canonical operators $R\coloneqq (x_1,p_1,\dots,x_m,p_m)$, the above relations take the compact form $[R_j,R_{k}]=i (\Omega_{m})_{jk}$, where $\Omega_m$ denotes the $2m\times 2m$ standard symplectic form defined in \eqref{comm}. We will often omit the subscript $m$ if the number of modes is fixed. The \emp{total photon number} is defined by
\bb
N\coloneqq \sum_{j=1}^m a_j^\dag a_j = \sum_{j=1}^m \frac{x_j^2 + p_j^2}{2} - \frac{m}{2}\, .
\label{total_photon_number}
\ee

The following formulae involving displacement operators and characteristic functions follow the conventions of the monograph by Serafini~\cite[Chapter~3]{BUCCO}. Given a real vector $z\in \RR^{2m}$ we define the displacement operator $\D (z)$ as
\bb
\D(z) \coloneqq \exp\left[i z^\intercal \Omega R \right] = \exp\left[ - i \sum_j (\Omega z)_j R_j \right] = \D(-z)^\dag \,.
\label{D}
\ee
Due to \eqref{CCRlie}, the following `Weyl commutation relation' is valid for any $z,w\in\RR^{2m}$:
\bb
\D (z) \D (w) = e^{-\frac{i}{2} z^\intercal \Omega w} \D (z+w)\,.
\label{Weyl}
\ee
A quantum state on $\cH_m$ is fully determined by its \emp{characteristic function} $\chi_\rho:\RR^{2m}\to\CC$, given by
\bb
\chi_\rho(z)\coloneqq \tr[\rho\,\D (- z)]\,.
\label{chi}
\ee
A density operator is said to represent a \emp{Gaussian state} if its characteristic function is that of a multivariate Gaussian distribution, in formula~\cite[Eq.~(4.48)]{BUCCO}
\begin{align*}
    \chi_\rho(z) = \exp\left[-\frac{1}{4}\,z^\intercal \Omega^\intercal \gamma \Omega z+i\mu^\intercal \Omega z \right]\,,
\end{align*}
where $\mu\in\RR^{2m}$ is its \emp{mean vector}, i.e.\ a real vector of mean values $\mu_j\coloneqq \tr[\rho R_j]$, and $\gamma$ is the \emp{covariance matrix} of $\rho$, that is, the $2m\times 2m$ real symmetric matrix whose entries are defined by
\begin{align*}
    \gamma_{jk}\coloneqq \tr\left[ \rho\,\{R_j-\mu_j,\, R_{k}-\mu_{k}\}  \right] ,
\end{align*}
with $\{\cdot, \cdot\}$ being the anti-commutator. In the general case of a (not necessarily Gaussian) state $\rho$, its covariance matrix needs to satisfy the so-called uncertainty inequality
\begin{align*}
    \gamma\geq i\Omega\,.
\end{align*}

A bosonic Gaussian channel $\Phi:\cT_1(\cH_m)\to\cT_1(\cH_m)$ is defined as a linear map which, for all $z\in\RR^{2m}$, operates on $\D (z)$ according to
\begin{align}\label{gaussianchannel}
\Phi^\dagger(\D(z))=\D(\Omega X \Omega^\intercal z)\,\exp\left[ -\frac{1}{4}z^\intercal \Omega^\intercal Y \Omega z - iv^\intercal \Omega z\right]\,,
\end{align}
where $\Phi^\dagger$ denotes the dual map of $\Phi$ with respect to the Hilbert--Schmidt inner product, for a given fixed real vector $v\in\RR^{2m}$, and $Y,X\in\mathbb{M}_{2m}(\RR)$, with $Y$ being a symmetric matrix, such that the following complete positivity condition is satisfied:
\begin{align}\label{cond}
Y\geq i(\Omega -X^\intercal \Omega X )\,.
\end{align}

A bosonic Gaussian channel $\Phi$ is hence uniquely characterized by the triple $(X, Y, v)$ for which \eqref{cond} holds. It maps Gaussian states to Gaussian states, transforming the mean vector $\mu$ and the covariance matrix $\gamma$ of the input Gaussian state as follows:
\begin{align*}
    \Phi: \mu \mapsto X \mu + v\, ; \qquad \gamma \mapsto X\gamma X^\intercal + Y\, .
\end{align*}

An important subset of bosonic Gaussian channels is the set of \emp{Gaussian unitary channels}. The latter are characterized by triples of the form $(X, 0, v)$, with $X\in \operatorname{Sp}_{2m}(\RR)$, and $v\in\RR^{2m}$ arbitrary. In the important case in which $X=I_{2m}$, the channel acts as follows: $\Phi(\cdot)\coloneqq \D (v)(\cdot)\D (v)^\dagger$. In the case in which $v=0$ and $X\in \operatorname{Sp}_{2m}(\RR)$, the channel is characterized by its induced action on the phase space $\mathbb{R}^{2m}$: for all $z\in\RR^{2m}$,
\begin{align*}
    \Phi^\dagger(\D (z)) = \D (Xz)\,.
\end{align*}

\section{Energy-constrained discrimination of unitaries}

\subsection{Entanglement is not needed for optimal energy-constrained discrimination of two unitaries}

As we have seen in the main text (Theorem~\ref{unentangled_thm}), the expression for the EC diamond norm distance between two unitary channels can be considerably simplified, eliminating in particular the need for local ancillary systems. This generalises the seminal result of Aharonov et al.~\cite{Aharonov1998} (for an explicit proof, see Watrous~\cite[Theorem~3.55]{WATROUS}). Such extensions are made possible by the many improvements over the Toeplitz--Hausdorff theorem that have been investigated in the dedicated literature~\cite{AuYeung1979, AuYeung1983, Binding1985}.

\begin{manualthm}{\ref{unentangled_thm}}
Let $U,V$ be two unitary operators on a Hilbert space $\cH$ of dimension $\dim \cH\geq 3$, and call $\mathcal{U}(\cdot) \coloneqq U(\cdot)U^\dag$, $\mathcal{V}(\cdot)\coloneqq V(\cdot)V^\dag$ the associated unitary channels. Let $H\geq 0$ be a grounded Hamiltonian on $\cH$, and fix $E > 0$. Then the EC diamond norm distance between $\mathcal{U}$ and $\mathcal{V}$ satisfies that
\bb
\left\| \mathcal{U} - \mathcal{V} \right\|_{\diamond}^{H,E} = \sup_{\braket{\psi|H|\psi}\leq E} \left\| U\ketbra{\psi} U^\dag - V\ketbra{\psi} V^\dag \right\|_1 = 2 \sqrt{1 - \inf_{\braket{\psi|H|\psi}\leq E} \left| \braket{\psi|U^\dag V |\psi}\right|^2} \, .
\tag{\ref{unentangled}}
\ee
In other words, in this case the supremum in the definition of EC diamond norm can be restricted to unentangled pure states. 
\end{manualthm}

Before we delve into the proof of Theorem~\ref{unentangled_thm}, we need to recall some basic results in matrix analysis. For an $n\times n$ complex matrix $Z$, the \emp{field of values} of $Z$ is defined by~\cite[Definition~1.1.1]{HJ2}
\bb
F(Z) \coloneqq \left\{ \braket{\psi|Z|\psi}:\, \ket{\psi}\in \CC^n\right\} \subset \CC\, .
\label{field}
\ee
For a thorough introduction to the subject, we refer the reader to the excellent book by Horn and Johnson~\cite[Chapter~1]{HJ2}. The fundamental result here is the following.

\begin{thm}[{Toeplitz--Hausdorff~\cite{Toeplitz1918, Hausdorff1919}}] \label{Toeplitz_Hausdorff_thm}
For every complex matrix $Z$, the field of values $F(Z)\subset \CC$ defined by \eqref{field} is convex.
\end{thm}

The above result is proved in many textbooks~\cite[Section~1.3]{HJ2}. Here we will rather be interested in more recent improvements. A very intuitive generalised notion is that of \emp{$k$-dimensional field of $k$ matrices}. For a set of $k$ complex matrices $Z_1,\ldots, Z_k$ of size $n\times n$, set
\bb
F_k\left( Z_1,\ldots, Z_k\right) \coloneqq \left\{ \left(\braket{\psi|Z_1|\psi}, \ldots, \braket{\psi|Z_k|\psi}\right)^\intercal:\, \ket{\psi}\in \CC^n\right\} .
\label{k_field}
\ee
In general $F_k\left( Z_1,\ldots, Z_k\right)\subset \CC^k$; however, if the matrices $Z_j$ are Hermitian, we will rather have that $F_k\left( Z_1,\ldots, Z_k\right)\subset \RR^k$. In this language, the Toeplitz--Haussdorf theorem can be also cast in the following alternative form: \emph{for every two Hermitian matrices $X,Y$ of the same size, the $2$-dimensional field of values $F_2(X,Y)\subset \RR^2$ is convex.} Naturally, this is the same as Theorem \ref{Toeplitz_Hausdorff_thm} up to the identifications $Z=X+iY$ and $\RR^2\simeq \CC$. It turns out that something stronger holds.

\begin{thm}[{Au-Yeung--Poon~\cite{AuYeung1979, AuYeung1983, Binding1985}}] \label{Au-Yeung_thm}
For any three Hermitian matrices $X,Y,Z$ of the same size $n\geq 3$, the associated $3$-dimensional field of values $F_3(X,Y,Z)\subset \RR^3$ defined by \eqref{k_field} is convex.
\end{thm}

We are now ready to prove our main result.

\begin{proof}[Proof of Theorem \ref{unentangled_thm}]
We start by following the general approach put forth by Watrous~\cite[Theorem~3.55]{WATROUS}. Let $\ket{\Psi} \in \cH\otimes \cH'$ be an arbitrary bipartite pure state with Schmidt decomposition
\bbb
\ket{\Psi} = \sum_i \sqrt{p_i} \ket{e_i} \ket{f_i}\, .
\eee
The energy constraints on $\Psi$ reads
\bb
\braket{\Psi| H\otimes I |\Psi} = \sum_i p_i \braket{e_i | H| e_i}\leq E\, .
\ee

Furthermore, remembering that $\left\|\ketbra{\alpha} - \ketbra{\beta}\right\|_1 = 2\sqrt{1-\left|\braket{\alpha|\beta}\right|^2}$ for every pair of pure states $\ket{\alpha}, \ket{\beta}$, we have that
\begin{align*}
\frac14 \left\| \left(\left( \mathcal{U} - \mathcal{V} \right) \otimes \Id\right)\left( \Psi \right) \right\|_1^2 &= \frac14 \left\| (U\otimes I) \ketbra{\Psi} (U\otimes I)^\dag - (V\otimes I) \ketbra{\Psi} (V\otimes I)^\dag \right\|_1^2 \\
&= 1 - \left| \braket{\Psi | U^\dag V \otimes I |\Psi} \right|^2 \\
&= 1 - \left| \sum_i p_i \braket{e_i | U^\dag V |e_i} \right|^2 .
\end{align*}
Therefore, denoting by $\rho \coloneqq \sum_i p_i \ketbra{e_i}$ the reduced state of $\Psi$ on the first subsystem, we immediately see that
\begin{align}
    \left\|\mathcal{U}-\mathcal{V}\right\|_\diamond^{H,E} =&\ 2\sqrt{1 - \nu_E\left(U^\dag V\right)^2}\, , \label{unitaries_diamond_intermediate_1} \\
    \nu_E (W) \coloneqq&\ \inf_{\rho:\, \tr[\rho H]\leq E} \left| \tr[\rho W] \right| . \label{unitaries_diamond_intermediate_2}
\end{align}
Note that the function $\nu_E$ in \eqref{unitaries_diamond_intermediate_2} is well defined on all bounded operators.

The next step is to recast the above function in terms of an optimisation over states with finite rank. Define the modified function
\bb
\widebar{\nu}_E(W) \coloneqq \inf_{\substack{\rho:\, \rk(\rho)<\infty,\\ \tr[\rho H]\leq E}} \left| \tr[\rho W] \right|
\label{nubar}
\ee
We claim that in fact it holds that
\bb
\nu_E(W) = \widebar{\nu}_E(W)\, .
\label{nu=nubar}
\ee
To see why this is the case, start by noting that $\nu_E(W)\leq \widebar{\nu}_E(W)$ holds by definition. The other direction can be proved as follows. Consider a state $\rho$ such that $\tr[\rho H]\leq E$, and let $\rho = \sum_{k=0}^\infty \lambda_k \ketbra{\psi_k}$ be its spectral decomposition, with $\lambda_k>0$ for all $k$. Define $p_n \coloneqq \sum_{k=0}^{n-1} \lambda_k$, as well as the two auxiliary states $\rho_n \coloneqq \frac{1}{p_n}\sum_{k=0}^{n-1} \lambda_k \ketbra{\psi_k}$ and $\sigma_n\coloneqq \frac{1}{1-p_n}\sum_{k=n}^\infty \lambda_k \ketbra{\psi_k}$, so that $\rho = p_n \rho_n + (1-p_n)\sigma_n$. Note that $\lim_{n\to\infty} p_n =1$ and therefore also $\lim_{n\to \infty}\tr[\rho_n H] = \tr[\rho H] \leq E$. Since $E>0=\min \spec(H)$, we can pick a vector $\ket{\phi}\in \dom(H)$ such that $0\leq \braket{\phi|H|\phi}<E$~\cite[Proposition~9.18]{HALL}. For all sufficiently large $n\in \NN$, set
\begin{align*}
\omega_n &\coloneqq q_n \rho_n + (1-q_n) \ketbra{\phi}\, , \\
q_n &\coloneqq \max\left\{1,\, \frac{E-\braket{\phi|H|\phi}}{\tr[\rho_n H] - \braket{\phi|H|\phi}}\right\} .
\end{align*}
Note that $q_n$ is well-defined for all sufficiently large $n$, and that $\lim_{n\to\infty} q_n =1$. Clearly, $\rk(\omega_n)\leq \rk(\rho_n)+1\leq n+1$, and moreover
\bbb
\tr[\omega_n H] = q_n \left( \tr[\rho_n H] - \braket{\phi|H|\phi}\right) + \braket{\phi|H|\phi} \leq E\, .  
\eee
Also, since
\bbb
\left\|\rho - \omega_n \right\|_1 = \left\| (p_n -q_n) \rho_n + (1-p_n)\sigma_n - (1-q_n) \ketbra{\phi}\right\|_1 \leq |p_n-q_n| + 2-p_n-q_n\, ,
\eee
we deduce that
\bbb
\lim_{n\to\infty} \left\|\rho - \omega_n \right\|_1 = 0\, .
\eee
Since $W$ is bounded, this implies that
\bbb
\lim_{n\to \infty} \tr[\omega_n W] = \tr[\rho W]\, .
\eee
Therefore,
\bbb
\left|\tr[\rho W]\right| = \lim_{n\to \infty} \left| \tr[\omega_n W] \right| \geq \widebar{\nu}_E(W)\, .
\eee
Since this holds for all $\rho$ appearing in the infimum in \eqref{unitaries_diamond_intermediate_2}, we obtain \eqref{nu=nubar}.

We now show that one can further simplify \eqref{nubar} by restricting the infimum to pure states only. That is, we claim that
\bb
\widebar{\nu}_E(W) = \inf_{\substack{\ket{\psi}\in \dom(H^{1/2}):\\ \braket{\psi|H|\psi}\leq E}} \left|\braket{\psi|W|\psi}\right| .
\label{nupure}
\ee
Clearly, plugging \eqref{nupure} into \eqref{unitaries_diamond_intermediate_1}--\eqref{unitaries_diamond_intermediate_2} would conclude the proof. Therefore, it remains only to prove \eqref{nupure}. We write that
\begin{align*}
    \widebar{\nu}_E(W) &\texteq{1} \inf_{\substack{\cS \subseteq \dom(H^{1/2}), \\ 3\leq \dim \cS<\infty}} \inf_{\substack{\rho\in \cD(
    \cS), \\ \tr[\rho H]\leq E}} \left|\tr[\rho W]\right| \\
    &\texteq{2} \inf_{\substack{\cS \subseteq \dom(H^{1/2}),\\ 3\leq \dim \cS<\infty}} \inf_{\substack{(x,y,z)\in \co(\mathcal{R}_\cS):\\z\leq E}} \sqrt{x^2+y^2} \\
    &\texteq{3} \inf_{\substack{\cS \subseteq \dom(H^{1/2}),\\ 3\leq \dim \cS<\infty}} \inf_{\substack{(x,y,z)\in \mathcal{R}_\cS:\\ z\leq E}} \sqrt{x^2+y^2} \\
    &= \inf_{\substack{\cS \subseteq \dom(H^{1/2}), \\ 3\leq \dim \cS<\infty}} \inf_{\substack{\ket{\psi}\in \cS, \\ \braket{\psi|H|\psi}\leq E}} \left|\braket{\psi|W|\psi}\right| \\
    &= \inf_{\substack{\ket{\psi}\in \dom(H^{1/2}):\\ \braket{\psi|H|\psi}\leq E}} \left|\braket{\psi|W|\psi}\right| .
\end{align*}
The above identities can be justified as follows. As a start, 1~is just a rephrasing of \eqref{nu=nubar}; since enlarging $\cS$ cannot decrease the infimum, the added constraint that $\dim \cS\geq 3$ causes no loss of generality, and is there just for future convenience. In~2 we defined the regions
\bbb
\mathcal{R}_\cS \coloneqq F_3 \left( \Pi_\cS W_R \Pi_\cS^\dag,\, \Pi_\cS W_I \Pi_\cS^\dag,\, \Pi_\cS H\Pi_\cS^\dag \right) ,
\eee
where $\Pi_\cS:\cH\to \cS$ is the orthogonal projection onto the finite-dimensional subspace $\cS$, and $W = W_R+iW_I$, with $W_R, W_I$ bounded and self-adjoint. Note that the convex hull of $\mathcal{R}_\cS$ appears because all density operators $\rho\in \cD(\cH)$ are convex mixtures of pure states. In step~3 we applied the Au-Yeung--Poon Theorem~\ref{Au-Yeung_thm}, which guarantees that $\mathcal{R}_\cS$ is already a convex region of $\RR^3$. This is made possible by the fact that $\Pi_\cS W_R \Pi_\cS^\dag$, $\Pi_\cS W_I \Pi_\cS^\dag$, and $\Pi_\cS H\Pi_\cS^\dag$ are all finite-dimensional linear operators, i.e.\ matrices. 
\end{proof}

\begin{rem}
In the finite-dimensional case, for a normal~\footnote{A matrix $Z$ is normal if it commutes with its Hermitian conjugate, in formula $\left[ Z,Z^\dag \right]=0$.} matrix $Z$ the field of values coincides with the convex hull of the spectrum, in formula $F(Z) = \co\left( \spec (Z) \right)$. In particular, in that case $F(Z)$ will be a polygon. In general, this seems to be no longer the case when one imposes an energy constraint. In other words, $\left\{ \braket{\psi|W|\psi}:\ket{\psi}\in \cH,\, \braket{\psi|H|\psi}\leq E \right\}\subset \CC$ will not be a polygon even when $W$ is normal.
\end{rem}

\subsection{Perfect discrimination with energy constraint in the multi-copy setting}

The following generalises a celebrated result by Ac\'in~\cite{Acin2001}, subsequently improved by Duan et al.~\cite{Duan2007, Duan2008,Duan2009}.

\begin{manualthm}{\ref{eventual_discrimination_thm}}
Let $U,V$ be two distinct unitary operators on a Hilbert space $\cH$ of dimension $\dim \cH\geq 3$, and denote by $\mathcal{U}(\cdot) \coloneqq U(\cdot)U^\dag$, $\mathcal{V}(\cdot)\coloneqq V(\cdot)V^\dag$ the associated unitary channels. Let $H\geq 0$ be a grounded Hamiltonian on $\cH$. Then there exists a positive integer $n$ such that $n$ parallel uses of $\mathcal{U}$ and $\mathcal{V}$ can be discriminated perfectly with some finite total energy $E<\infty$, i.e.
\bb
\left\| \mathcal{U}^{\otimes n} - \mathcal{V}^{\otimes n} \right\|_\diamond^{H_{(n)},\,E}=2\, ,
\ee
where $H_{(n)}\coloneqq \sum_{j=1}^n H_j$ is the $n$-copy Hamiltonian, and $H_j\coloneqq  I \otimes \cdots I \otimes H \otimes I \cdots \otimes I$, with the $H$ in the $j^{\text{th}}$ location.
\end{manualthm}

\begin{proof}
For a bounded operator $W$ and $E>0$, define the region of the complex plane 
\bbb
\mathcal{S}_E(W)\coloneqq \left\{\braket{\psi|W|\psi}:\, \braket{\psi|H|\psi}\leq E; \Vert |\psi\rangle \Vert=1 \right\} .
\eee
By Theorem \ref{unentangled_thm}, we have that $\left\|\mathcal{U} -\mathcal{V}\right\|_\diamond^{H,E}=2$ if and only if $0\in \mathcal{S}_E(U^\dag V)$. It is not too difficult to see that when $W$ is a normal operator it holds that 
\bb
\inter\big( \co (\spec(W))\big) \subseteq \bigcup_{E>0} \mathcal{S}_E(W) \subseteq \co (\spec(W))\, ,
\label{int_conv}
\ee
where $\spec(W)$ is the spectrum of $W$. The upper bound follows trivially from the spectral theorem, while the lower bound can be proved as follows. Let $z\in \inter\big( \co (\spec(W))\big)$ be a complex number. Clearly, there exists $w_i\in \spec(W)$ ($i=1,2,3$) such that $z\in \inter\left(\co\{w_i\}_i\right)$. Pick $\epsilon>0$ small enough so that in fact $z\in \co\{w'_i\}_i$ whenever $|w'_i-w_i|\leq \epsilon$ for all $i=1,2,3$. By the well-known existence of approximate eigenvectors of bounded operators \cite[Proposition~7.7]{HALL}, we can find normalised vectors $\ket{\psi_i}$ such that $\left\|W\ket{\psi_i} - w_i \ket{\psi_i}\right\|\leq \epsilon/2$. This in particular implies that $\left|\braket{\psi_i| W |\psi_i} - w_i \right|\leq \epsilon/2$. 

Since $H$ is densely defined, we can approximate each $\ket{\psi_i}$ with a vector $\ket{\phi_i}\in \dom (H)$ to any desired degree of accuracy. In particular, we can safely assume that $\left|\braket{\phi_i| W |\phi_i} - w_i \right|\leq \epsilon$ for all $i=1,2,3$. Clearly, we will have that $\braket{\phi_i|H|\phi_i} \leq E$ for some $E<\infty$. Moreover, our assumptions imply that $z\in \co \left\{ \braket{\phi_i|W|\phi_i} \right\}_i$. Therefore, by the standard Toeplitz--Hausdorff Theorem \ref{Toeplitz_Hausdorff_thm} (see the original works~\cite{Toeplitz1918, Hausdorff1919}), a linear combination $\ket{\phi}$ of the three vectors $\ket{\phi_i}$ will be such that $\braket{\phi|W|\phi}=z$. By multiple applications of the Cauchy--Schwarz inequality, it is easy to show that $\braket{\phi|H|\phi}\leq 3E$. Incidentally, the Au-Yeung--Poon Theorem \ref{Au-Yeung_thm} would allow us to get the better estimate $\braket{\phi|H|\phi}\leq E$, which is however not needed. Therefore, $z\in \mathcal{S}_{3E}(W)$, which completes the proof of \eqref{int_conv}. 

From now on we can follow the blueprint of Ac\'{i}n's proof \cite{Acin2001}. Set $W=U^\dag V$. Since $U\neq V$ and thus $W\neq I$, we have that $\Theta \coloneqq \sup_{w,w'\in \spec(W)} \left\{\arg(w) - \arg(w')\right\} > 0$, where $\arg(w)\in (-\pi,\pi]$. Picking $n\coloneqq \floor{\frac{\pi}{\Theta}}+1$, we have that $0\in \inter\big( \co \left(\spec\left(W^{\otimes n}\right) \right)\big)$. By \eqref{int_conv}, this implies that $0\in \mathcal{S}_E\left(W^{\otimes n}\right)$ for some $E<\infty$, in turn ensuring that $\left\|\mathcal{U}^{\otimes n} -\mathcal{V}^{\otimes n} \right\|_\diamond^{H,E}=2$.
\end{proof}

Without further information on the interplay between the unitaries $U,V$ and the Hamiltonian $H$, it is in general not possible to say anything more specific about the value of the threshold energy $E$ that allows to eventually achieve perfect discrimination as per Theorem \ref{eventual_discrimination_thm}. To see this, consider the following example.

\begin{ex}
Let $\cH=\CC^2$ be the Hilbert space of a single qubit. Set $H = \ketbra{1}$, $U=I$, and $V= \ketbra{0} + e^{i\theta} \ketbra{1} = e^{i\theta H}$, where $0<\theta<\pi$. Let $E>0$ be such that $\left\|\mathcal{U}^{\otimes n} - \mathcal{V}^{\otimes n} \right\|_\diamond^{H_{(n)},E} = 2$ for some positive integer $n$, where as usual $H_{(n)} = \sum_{j=1}^n H_j$. Then it holds that
\bb
E\geq \frac{1}{12} + \frac{\sqrt6}{9\theta}\, .
\label{ex:estimate_E}
\ee
Before delving into the proof of~\eqref{ex:estimate_E}, let us discuss its consequences. Since $\theta$ is arbitrary, it implies that no substantial improvement over Theorem \ref{eventual_discrimination_thm} is possible unless we give some additional information on the relationship between $H$ on the one hand and $U,V$ on the other.

We now prove \eqref{ex:estimate_E}. We use Theorem \ref{unentangled_thm} to deduce from the hypotheses the existence of some vector $\ket{\Psi}\in \left( \CC^2 \right)^{\otimes n} = \CC^{2^n}$ such that $\braket{\Psi|H_{(n)}|\Psi} \leq E$ and $\left|\braket{\Psi|V^{\otimes n}|\Psi}\right| = 0$. Defining the random variable $K\in\{0,\ldots, n\}$ with probability distribution $p_k \coloneqq \sum_{j_1,\ldots, j_n\in \{0,1\}:\, \sum_\alpha j_\alpha = k} \left|\braket{\Psi|j_1\ldots j_n}\right|^2$, we see that $\E\, K = \sum_{k=0}^n k p_k \leq E$ and $\left|\sum_{k=0}^n p_k e^{i k \theta}\right|=0$. By Markov's inequality, for every fixed $k\in \{0,\ldots, n\}$ we have that $\pr\{ K\geq k\}\leq \frac{E}{k}$, while
\bbb
0 = \E\, e^{i\theta K} \geq \pr\{ K \leq k-1\} \cos((k-1)\theta) - \pr\{ K \geq k\} \geq \cos((k-1)\theta) - \left(1+\cos((k-1)\theta)\right) \frac{E}{k}\, ,
\eee
from which we deduce that
\begin{align*}
E &\geq \max_{k\in \{0,\ldots, n\}} \frac{k\cos((k-1)\theta)}{1+\cos((k-1)\theta)} \\
&\geq \max_{k\in \{0,\ldots, n\}} \frac{k}{2} \cos((k-1)\theta) \\
&\geq \max_{k\in \{0,\ldots, n\}} \frac{k}{2} \left( 1 - \frac12 (k-1)^2 \theta^2\right) \\
&\textgeq{1} \max_{x\in \RR} \frac{x}{2} \left( 1-\frac12 (x-1)^2 \theta^2\right) - \frac14 \\
&\texteq{2} \frac{\left(\sqrt{\theta^2+6}+2 \theta \right) \left(\theta  \left(\sqrt{\theta^2+6}-\theta \right)+6\right)}{54 \theta } - \frac14 \\
&\geq \frac{1}{12} + \frac{\sqrt6}{9\theta}\, .
\end{align*}
Here, 1~holds by applying Lagrange's theorem, because the derivative of the function $x\mapsto \frac{x}{2} \left( 1-\frac12 (x-1)^2 \theta^2\right)$ is at most $1/2$, and there is always an integer at a distance at most $1/2$ from every real number. Finally, 2~follows by an elementary maximisation whose details we leave to the reader.
\end{ex}

\section{Time evolution and quantum speed limits}

We now turn to the study of the estimates on quantum evolutions for two different channels under energy constraints.

\subsection{Proof of Theorem~\ref{thm:favard} and Corollary~\ref{propgaussexample}}

We start by considering the dynamics of a closed quantum systems evolving under the action of one of two distinct Hamiltonians in the presence of an energy constraint. These bounds immediately lead to a \emph{quantum speed limit}, namely to a bound on the minimum time required for the two different dynamics to evolve a given state to a pair of states which are a pre-specified distance $d$ apart. The theorem is stated in a slightly more general form than in the main text, where $H$ is taken to be equal to $H_0$.

\begin{manualthm}{\ref{thm:favard}}
Let $H,H'$ be self-adjoint operators, and let $H_0$ be positive semi-definite. Without loss of generality, assume that $0$ is in the spectrum of $H$. Let $H-H'$ be relatively bounded with respect to $H$, and let $|H|^{1/2}$ be relatively bounded with respect to $H_0^{1/2}$. In other words, let there be constants $\alpha,\beta,\gamma, \delta >0$ such that
\begin{align}
\Vert (H-H')\ket{\psi} \Vert &\le \alpha\, \Vert H\ket{\psi} \Vert+ \beta\, \Vert\ket{\psi} \Vert  \quad \text{ for all } \ket{\psi}\in \dom(H), \label{relatbound_H-H'} \\
\Vert \vert H \vert^{1/2} \ket{\psi} \Vert^2 &\le \gamma\, \Vert H_0^{1/2} \ket{\psi} \Vert^2 + \delta\, \Vert\ket{\psi} \Vert^2 \quad \text{ for all } \ket{\psi} \in \dom\left(H_0^{1/2}\right) . \label{relatbound_|H|}
\end{align}
Then, for all $t\geq 0$ and for all normalised states $\ket{\psi}$, the unitary operators $U_t \coloneqq e^{-iHt}$, $V_t \coloneqq e^{-iH't}$ satisfy that
\begin{equation}
\Vert (U_t-V_t)\ket{\psi}\Vert \le 2 \sqrt{\gamma \braket{ \psi|H_0|\psi} + \delta} \sqrt{\alpha t} + \beta t.
\label{drift_operators_SM}
\end{equation}
The associated unitary channels $\cU_t(\cdot) \coloneqq U_t(\cdot)U_t^\dag$, $\cV_t(\cdot)\coloneqq V_t(\cdot)V_t^\dag$ instead satisfy that
\begin{equation}
\left\| \cU_{t} - \cV_{t} \right\|_{\diamond}^{H_0,E} \le 2 \sqrt{2}\sqrt{\gamma E + \delta} \sqrt{\alpha t} + \sqrt{2}\beta t \label{drift_SM} 
\end{equation}
for all $E\geq 0$. In particular, the minimum time $t$ needed for $(U_t)_{t\in \RR}$ and $(V_t)_{t\in \RR}$ to evolve a state $\ket{\psi}$ to a pair of states which are at a pre-specified distance $d$ apart, i.e.~such that $\Vert (U_t-V_t)\ket{\psi}\Vert = d$, is bounded from below by
\begin{equation}
t \geq \frac{1}{\beta^2}\left(\sqrt{d\beta+\nu(\psi)^2}-\nu(\psi)\right)^2,
\end{equation}
where $\nu(\psi) \coloneqq \sqrt{\alpha (\gamma\braket{\psi|H_0|\psi} + \delta)}$.
\end{manualthm}

\begin{rem}
Note that $\Vert (U_t-V_t)\ket{\psi}\Vert = \Vert (V_t^\dag U_t-I)\ket{\psi}\Vert$, and that analogously $\left\| \cU_{t} - \cV_{t} \right\|_{\diamond}^{H_0,E} = \left\| \cV_t^\dag \cU_{t} - \Id \right\|_{\diamond}^{H_0,E}$. The operator $M_t\coloneqq V_t^\dag U_t$ is called the `Loschmidt echo operator'. For a review of its known properties, we refer the reader to~\cite{Gorin2006}.
\end{rem}

\begin{proof}[Proof of Theorem~\ref{thm:favard}]
For all $\lambda>0$, we decompose the normalised state $\ket{\psi}$ as
\bb
\begin{aligned}
\ket{\psi} =&\ \ket{\psi_\lambda} + \ket{\varphi_\lambda} , \\
\ket{\psi_\lambda} \coloneqq&\ \lambda (\lambda+iH)^{-1}\ket{\psi} , \\
\ket{\varphi_\lambda} \coloneqq&\ iH(\lambda + iH)^{-1} \ket{\psi} .
\end{aligned}
\label{decomposition_psi}
\ee
Note that neither $\ket{\psi_\lambda}$ nor $\ket{\varphi_\lambda}$ are in general normalised. In fact, we can estimate their norms as follows. First,
\bb
\Vert \ket{\psi_\lambda}\Vert = \lambda\, \sqrt{\braket{\psi| \left(\lambda^2 + H^2\right)^{-1} |\psi}} \le 1
\label{estimate_norm_psi_lambda}
\ee
for all $\lambda>0$ by operator monotonicity of the inverse. Second, denote with $H = \int_{\spec(H)} z\, dE^H(z)$ the spectral decomposition of $H$, and define $E_\psi^H$ as the measure on $\spec(H)$ such that $E_\psi^H(X)\coloneqq \braket{\psi|E^H(X)|\psi}$ for all measurable $X\subseteq \spec(H)$. Then,
\bb
\begin{aligned}
\Vert \ket{\varphi_\lambda} \Vert^2 &= \Vert H(\lambda+iH)^{-1} \ket{\psi} \Vert^2 \\
&= \braket{\psi| \frac{H^2}{\lambda^2+H^2} |\psi} \\
&= \int_{\spec(H)} \frac{z^2}{\lambda^2 + z^2}\, dE_\psi^{H}(z) \\
&\textleq{1} \frac{1}{\lambda} \int_{\spec(H)} |z|\, dE_\psi^{H}(z) \\
&= \frac{1}{2\lambda}\, \braket{\psi|\, |H|\, |\psi} \\
&\textleq{2} \frac{1}{2\lambda}\left( \gamma \braket{\psi|H_0|\psi} + \delta \right) ,
\end{aligned}
\label{estimate_norm_varphi_lambda}
\ee
where 1~descends from the elementary inequality $\frac{z^2}{\lambda^2 + z^2} \leq \frac{|z|}{\lambda}$, while 2~is an application of~\eqref{relatbound_|H|}. Then,
\bb
\begin{aligned}
\Vert (U_t-V_t)\ket{\psi} \Vert &\textleq{3} \inf_{\lambda>0} \left\{ \Vert (U_t-V_t)\ket{\psi_\lambda} \Vert + \Vert (U_t-V_t)\ket{\varphi_\lambda} \Vert \right\} \\
&\textleq{4} \inf_{\lambda>0} \left\{ \Vert (U_t-V_t)\ket{\psi_\lambda} \Vert + 2\Vert \ket{\varphi_\lambda} \Vert \right\} \\
&= \inf_{\lambda>0} \left\{ \left\Vert \int_0^t \frac{d}{ds}\left(U_{t-s}V_{s}\right)\ket{\psi_\lambda} \ ds \right\Vert + 2\Vert \ket{\varphi_\lambda} \Vert \right\} \\
&= \inf_{\lambda>0} \left\{ \left\Vert \int_0^t U_{t-s}(H'-H) V_{s} \ket{\psi_\lambda} \ ds \right\Vert + 2\Vert \ket{\varphi_\lambda} \Vert \right\}\\
&\textleq{5} \inf_{\lambda>0} \left\{ \int_0^t \left\Vert U_{t-s}(H'-H) V_{s} \ket{\psi_\lambda} \right\Vert ds + 2\Vert \ket{\varphi_\lambda} \Vert \right\} \\
&\leq \inf_{\lambda>0} \left\{ t \sup_{0\leq s\leq t} \left\Vert (H'-H) V_{s} \ket{\psi_\lambda} \right\Vert + 2\Vert \ket{\varphi_\lambda} \Vert \right\} \\
&\textleq{6} \inf_{\lambda>0} \left\{ t \sup_{0\leq s\leq t} \left\{ \alpha \Vert H V_s \ket{\psi_\lambda} \Vert + \beta \Vert V_s \ket{\psi_\lambda}\Vert \right\} + 2\Vert \ket{\varphi_\lambda} \Vert \right\} \\
&\texteq{7} \inf_{\lambda>0} \left\{ t \left( \alpha \Vert H \ket{\psi_\lambda} \Vert + \beta \Vert \ket{\psi_\lambda}\Vert \right) + 2\Vert \ket{\varphi_\lambda} \Vert \right\}
\end{aligned}
\ee
\bb
\begin{aligned}
\hspace{6.4ex} &\texteq{8} \inf_{\lambda>0} \left\{ \beta t \Vert \ket{\psi_\lambda}\Vert + (\lambda \alpha t + 2) \Vert \ket{\varphi_\lambda} \Vert \right\} \\
&\textleq{9} \inf_{\lambda>0} \left\{ \beta t + (\lambda \alpha t + 2) \sqrt{\frac{1}{2\lambda}\left( \gamma \braket{\psi|H_0|\psi} + \delta \right)} \right\} \\
&\texteq{10} \beta t + 2\, \sqrt{\gamma \braket{\psi|H_0|\psi} + \delta}\, \sqrt{\alpha t}\, .
\end{aligned}
\ee
Here: 3,~4,~and~5 are just the triangle inequality; in~6 we used \eqref{relatbound_H-H'}; in~7 we observed that $V_s$ commutes with $H$; in~8 we noted that $H\ket{\psi_\lambda} = -i \lambda \ket{\varphi_\lambda}$; in~9 we employed \eqref{estimate_norm_psi_lambda} and \eqref{estimate_norm_varphi_lambda}; and finally in~10 we solved the elementary minimisation in $\lambda$. This proves \eqref{drift_operators_SM}.

Thanks to Theorem \ref{unentangled_thm}, to deduce \eqref{drift_SM} it suffices to substitute the preceding bound into
\begin{equation}
\begin{aligned}
\left\| \cU_t - \cV_t \right\|_{\diamond}^{H_0,E} &= \sup_{\braket{\psi|H_0|\psi}\leq E} \left\| U_t\ketbra{\psi} U_t^\dag - V_t\ketbra{\psi} V_t^\dag \right\|_1 \\
&= 2\sup_{\braket{\psi|H_0|\psi}\leq E} \sqrt{1 - \left| \braket{\psi|U_t^\dag V_t|\psi}\right|^2} \\
&\le \sqrt{2}\sup_{\braket{\psi|H_0|\psi}\le E}\Vert U_t\ket{\psi}-V_t\ket{\psi} \Vert .
\end{aligned}
\end{equation}
Here, the inequality follows because $1-\vert \braket{ \psi'|\varphi'} \vert^2 \le 1-\Re \braket{ \psi'|\varphi'} = \frac12 \Vert \ket{\psi'}-\ket{\varphi'} \Vert^2$ holds for any two states $\ket{\psi'},\ket{\varphi'}$.
\end{proof}

We illustrate the applicability of Theorem \ref{thm:favard} by deducing the following corollary. 

\begin{manualcor}{\ref{propgaussexample}}
On a system of $m$ modes, consider the two Hamiltonians $H = H_0 = \sum_{j=1}^m d_j a^\dag_j a_j$ and $H' = \sum_{j,k=1}^m \left( X_{jk} a_j^\dag a_k+Y_{jk} a_j a_k+Y_{jk}^{*} a_j^\dag a_k^\dag \right)$, where $d_j>0$ for all $j$, and $X,Y$ are two $m\times m$ matrices, with $X$ Hermitian. Then the corresponding unitary operators $U_t,V_t$ and the corresponding unitary channels $\cU_t,\cV_t$ satisfy~\eqref{drift_operators_SM} and~\eqref{drift_SM}, respectively, for all $t\geq 0$ and $E> 0$, with
\begin{equation}
\begin{aligned}
\alpha &= \Vert D^{-1} \Vert_\infty \left(\sqrt{\tfrac{3}{2}} \Vert X-D \Vert_2 + \left(1+\sqrt{\tfrac{3}{2}}\right) \Vert Y \Vert_2\right) , \\
\beta &= \tfrac{m-1}{\sqrt{2}}  \Vert X-D \Vert_2 +\sqrt{\tfrac{(2m+1)^2}{2}+2m^2} \Vert Y \Vert_2\, , \\
\gamma &= 1 , \\
\delta &= 0 .
\end{aligned}
\label{choices_parameters}
\end{equation}
where $D_{jk} \coloneqq d_j \delta_{jk}$, and $\Vert\cdot\Vert_\infty, \Vert \cdot \Vert_2$ denote the operator norm and the Hilbert--Schmidt norm, respectively.
\end{manualcor}

\begin{proof}
We first record the following simple estimate, using that all the terms in the double sum are positive,
\begin{equation*}
\begin{aligned}
\Vert H\ket{\psi} \Vert^2 &= \sum_{j,k=1}^m d_kd_j \braket{\psi| a_j^\dag a_j a_k^\dag a_k |\psi} \ge \sum_{k=1}^m \vert d_k\vert^2 \Vert a_k^{\dagger}a_k \ket{\psi} \Vert^2 \ge \sum_{k=1}^m \frac{\Vert a_k^{\dagger}a_k \ket{\psi} \Vert^2}{\Vert D^{-1} \Vert_\infty^2}.  
\end{aligned}
\end{equation*}
The relative boundedness is due to the following simple estimate: 
\begin{equation*}
\begin{aligned}
\Vert (H-H')\ket{\psi} \Vert &\le \sum_{j,k=1}^m \left(\Vert(X_{jk}-d_k\delta_{jk}) a_j^{\dagger} a_k \ket{\psi} \Vert + \Vert Y_{jk} a_j a_k \ket{\psi} \Vert+ \Vert Y_{jk}^{\dagger} a_j^{\dagger} a_k^{\dagger} \ket{\psi} \Vert \right) \\
&\le \sqrt{\sum_{j,k=1}^m \vert (X_{jk}-d_k\delta_{jk}) \vert^2} \sqrt{\sum_{j,k=1}^m \Vert  a_j^{\dagger} a_k \ket{\psi} \Vert^2} + \Vert Y \Vert_2 \left(\sqrt{\sum_{j,k=1}^m \Vert  a_j a_k \ket{\psi} \Vert^2} + \sqrt{\sum_{j,k=1}^m \Vert  a_j^{\dagger} a_k^{\dagger} \ket{\psi} \Vert^2}\right) \\
&=\Vert X-D \Vert_2 \underbrace{\sqrt{\sum_{j,k=1}^m \Vert  a_j^{\dagger} a_k \ket{\psi} \Vert^2}}_{\eqqcolon~(I)} + \Vert Y \Vert_2 \left(\underbrace{\sqrt{\sum_{j,k=1}^m \Vert  a_j a_k \ket{\psi} \Vert^2}}_{\eqqcolon~(II)} +\underbrace{ \sqrt{\sum_{j,k=1}^m \Vert  a_j^{\dagger} a_k^{\dagger} \ket{\psi} \Vert^2}}_{\eqqcolon~(III)}\right).
\end{aligned}
\end{equation*}
We then start by estimating~(I):
\begin{equation*}
\begin{aligned}
\sum_{j,k=1}^m \Vert a_j^{\dagger} a_k \ket{\psi} \Vert^2 &= \sum_{k=1}^m \Vert a_k^{\dagger} a_k \ket{\psi}\Vert^2 + \sum_{k \neq j} \braket{\psi| a_k^\dag a_j a_j^\dag a_k|\psi} \\
&= \sum_{k=1}^m \Vert a_k^{\dagger} a_k \ket{\psi}\Vert^2 + \sum_{k \neq j} \braket{\psi| a_k^\dag a_k \left(a_j^\dag a_j + 1\right)|\psi} \\
&= \sum_{k, j=1}^m \braket{ \psi| a_j^{\dagger} a_j a_k^{\dagger} a_k| \psi} + \sum_{k \neq j} \braket{\psi| a_k^{\dagger} a_k| \psi} \\
&\le \Vert D^{-1} \Vert_\infty^2 \Vert H \ket{\psi} \Vert^2 + (m-1) \sum_{k=1}^m \braket{\psi| a_k^{\dagger} a_k| \psi} \\
&\le \Vert D^{-1} \Vert_\infty^2 \Vert H \ket{\psi} \Vert^2 + (m-1)\Vert D^{-1} \Vert_\infty \braket{\psi|H|\psi} \\
&\le \frac{3}{2} \Vert D^{-1} \Vert_\infty^2 \Vert H \ket{\psi} \Vert^2+ \frac{(m-1)^2}{2} \Vert \ket{\psi} \Vert^2.
\end{aligned}
\end{equation*}
Here, in the last line we noticed that
\begin{equation*}
(m-1)\Vert D^{-1} \Vert_\infty \braket{\psi|H|\psi} \leq (m-1) \Vert D^{-1} \Vert_\infty \Vert H\ket{\psi}\Vert \leq \frac12 \left( (m-1)^2 + \Vert D^{-1} \Vert_\infty^2 \Vert H\ket{\psi}\Vert^2 \right) .
\end{equation*}
Continuing with the estimate on~(II), and using that 
\begin{equation*}
\Vert a_k^2 \ket{\psi} \Vert^2 = \braket{\psi|a_k^\dag a_k a_k^{\dagger} a_k| \psi}- \braket{\psi | a_k^\dag a_k |\psi } \le \Vert a_k^{\dagger} a_k \ket{\psi} \Vert^2,
\end{equation*}
we find
\begin{equation*}
\begin{split}
\sum_{j,k=1}^m \Vert a_j a_k \ket{\psi} \Vert^2 
&= \sum_{k=1}^m \Vert a_k^2 \ket{\psi} \Vert^2 + \sum_{j\neq k} \braket{\psi|a_j^\dag a_j a_k^\dag a_k |\psi} \\
&\le \sum_{j,k=1}^m \braket{\psi|a_j^\dag a_j a_k^\dag a_k |\psi} \\
&\le \Vert D^{-1} \Vert_\infty^2 \Vert H \ket{\psi} \Vert^2.
\end{split}
\end{equation*}
Turning now to~(III), we use that
\begin{equation*}
\begin{split}
\sum_{k=1}^m \Vert (a_k^{\dagger})^2 \ket{\psi} \Vert^2
&=\sum_{k=1}^m \braket{\psi |a_ka_ka_k^{\dagger} a_k^\dag |\psi } \\
&= \sum_{k=1}^m \braket{\psi | a_k a_k^{\dagger} a_k a_k^\dagger| \psi } + \sum_{k=1}^m \braket{ \psi|a_k a_k^{\dagger}|\psi } \\
&=\sum_{k=1}^m \braket{\psi | a_k^{\dagger} a_k a_k a_k^\dagger| \psi } + 2\sum_{k=1}^m \braket{ \psi|a_k a_k^{\dagger}|\psi } \\
&= \sum_{k=1}^m \braket{\psi | a_k^{\dagger} a_k a_k^\dagger a_k| \psi } + 2\sum_{k=1}^m \braket{ \psi|a_k a_k^{\dagger}|\psi} + \sum_{k=1}^m \braket{\psi | a_k^{\dagger} a_k | \psi} \\
&= \sum_{k=1}^m \braket{\psi | (a_k^{\dagger} a_k)^2 | \psi } + 3\sum_{k=1}^m \braket{\psi | a_k^{\dagger} a_k | \psi} +2 m \Vert \ket{\psi} \Vert^2
\end{split}
\end{equation*}
and combine this with
\begin{equation*}
\begin{split}
\sum_{j \neq k} \Vert  a_j^{\dagger} a_k^{\dagger} \ket{\psi} \Vert^2 &= \sum_{j \neq k} \braket{\psi|a_ja_j^\dag a_k a_k^\dag|\psi} \\
&= \sum_{j \neq k} \braket{\psi|(a_j^\dag a_j + 1) (a_k^\dag a_k + 1) |\psi} \\
&= \sum_{j \neq k} \braket{\psi|a_j^\dag a_j a_k^\dag a_k |\psi} + 2(m-1) \sum_{k=1}^m \braket{\psi|a_k^\dag a_k|\psi} + 2m(m-1)
\end{split}
\end{equation*}
to find that
\begin{equation*}
\begin{split}
\sum_{j,k=1}^m \Vert a_j^{\dagger} a_k^{\dagger} \ket{\psi} \Vert^2 &= \sum_{j,k=1}^m \braket{\psi|a_j^\dag a_j a_k^\dag a_k|\psi} + (2m+1) \sum_{k=1}^m \braket{\psi | a_k^{\dagger} a_k | \psi} + 2 m^2 \Vert \ket{\psi} \Vert^2 \\
&\leq \Vert D^{-1} \Vert_\infty^2 \Vert H \ket{\psi}\Vert^2 + (2m+1) \Vert D^{-1} \Vert_\infty \braket{\psi|H|\psi} + 2 m^2 \Vert \ket{\psi} \Vert^2 \\
&\le \frac{3}{2} \Vert D^{-1} \Vert_\infty^2 \Vert H \ket{\psi} \Vert^2 + \left(\tfrac{(2m+1)^2}{2}+2m^2\right)\Vert \ket{\psi} \Vert^2,
\end{split}
\end{equation*}
where in the last line we observed that
\begin{equation*}
(2m+1)\Vert D^{-1} \Vert_\infty \braket{\psi|H|\psi} \leq (2m+1) \Vert D^{-1} \Vert_\infty \Vert H\ket{\psi}\Vert \leq \frac12 \left( (2m+1)^2 + \Vert D^{-1} \Vert_\infty^2 \Vert H\ket{\psi}\Vert^2 \right) .
\end{equation*}
Putting all estimates together yields 
\begin{equation}
\begin{aligned}
\Vert (H-H')\ket{\psi} \Vert &\le  \Vert D^{-1} \Vert_\infty \left(\sqrt{\tfrac{3}{2}} \Vert X-D \Vert_2 + \left(1+\sqrt{\tfrac{3}{2}}\right) \Vert Y \Vert_2\right) \Vert H \ket{\psi} \Vert \\
&\quad + \left(\tfrac{m-1}{\sqrt{2}}  \Vert X-D \Vert_2 +\sqrt{\tfrac{(2m+1)^2}{2}+2m^2} \Vert Y \Vert_2 \right) \Vert \ket{\psi} \Vert
\end{aligned}
\end{equation}
This proves that \eqref{relatbound_H-H'} and \eqref{relatbound_|H|} hold with the choices in \eqref{choices_parameters} for the parameters $\alpha,\beta,\gamma,\delta$. We can therefore apply Theorem~\ref{thm:favard} and conclude.
\end{proof}

\begin{rem}
In general, the dependence of the estimates in \eqref{choices_parameters} on the Hilbert--Schmidt norm difference is not optimal. To see this, we consider $Y=0$ and $X,X'> 0$ diagonal. In this case, 
\begin{equation*}
\Vert (H-H') \ket{\psi} \Vert = \left\Vert \sum_{k=1}^m (X_k-X_k') a_k^{\dagger} a_k \ket{\psi} \right\Vert.
\end{equation*}
By the spectral theorem, there exists a unitary map $U_k$ that transform $a_k^{\dagger} a_k$ into a positive multiplication operator $M_k$. This way, writing $U=\operatorname{diag}(U_1,\ldots ,U_m)$ we find that
\bb
\begin{aligned}
\Vert (H-H') \ket{\psi} \Vert &= \left\Vert \sum_{k=1}^m (X_k-X_k') U^{-1} M_k U \ket{\psi} \right\Vert \\
&\le \max_k \vert X_k-X_k' \vert \left\Vert \sum_{k=1}^m M_k U \ket{\psi} \right\Vert = \Vert X-X'\Vert_\infty \left\Vert \sum_{k=1}^m a_k^{\dagger} a_k \ket{\psi} \right\Vert \\
&\le \Vert X^{-1} \Vert_\infty \Vert X-X'\Vert_\infty \Vert H \ket{\psi}\Vert.
\end{aligned}
\ee
\end{rem}

\subsection{Tightness of Theorem~\ref{thm:favard}} \label{tightness_subsec}

We now dwell on the problem of whether the estimates provided in Theorem~\ref{thm:favard} are tight. One could in fact expect the trace distance between the states corresponding to different unitary evolution operators to grow linearly in the time $t$ for very small $t$. Indeed, the angle $\theta(t)$ between the evolved states should be proportional to $t$ up to higher-order corrections, and the trace distance is just given by $\sin \theta(t) \approx \theta(t)$. Instead, Theorem~\ref{thm:favard} seems to suggest a faster than linear growth for small times.

We will now argue that the dependence on $\sqrt{t}$ of the estimates in Theorem~\ref{thm:favard} is actually tight, and that the above intuition does not hold up upon a closer inspection. This fact manifests itself with especial clarity in infinite-dimensional systems. We focus on a special case that already contains all the conceptual subtleties we wish to investigate, namely, that corresponding to the choices $H'=0$ and $H=H_0\geq 0$. In this especially simple setting Theorem~\ref{thm:favard} can be applied with $\alpha=\gamma=1$ and $\beta=\delta=0$. It yields the estimate
\bb
\left\|\cU_t - \Id \right\|_\diamond^{H,E} \leq 2\sqrt{2Et}\, .
\label{simple_drift}
\ee
This case was already studied in \cite[Proposition~3.2]{Simon-Nila}. Their estimate~\cite[Eq.~(3.6)]{Simon-Nila} is basically the same as \eqref{simple_drift}, although it features a slightly worse constant. Incidentally, this small improvement is made possible by Theorem~\ref{unentangled_thm}.

To build our intuition on solid grounds, let us clarify what is supposed to be meant by `small $t$'. In this problem there are essentially two time scales. The first is determined by the input energy $E$ and takes the value $T_E \coloneqq 1/E$. The second is instead linked to the absolute maximum value of the energy in the system, and we will denote it by $T_H \coloneqq \|H\|_\infty^{-1}$. Note that $T_H\leq T_E$. Of course, if the system is infinite-dimensional it could well happen --- and it typically \emph{will} happen --- that $T_H=0$, leaving only $T_E$ as a meaningful time scale.

In light of these considerations, in general the expression `small $t$' could either mean $t\ll T_E$ or the much stronger inequality $t\ll T_H$. If the latter is the case, it is not difficult to see that indeed
\bb
\left\|\cU_t - \Id \right\|_\diamond^{H,E} \leq \left\|\cU_t - I \right\|_\diamond \leq \sqrt2 \left\|U_t - I\right\|_\infty \leq \sqrt2\, t \|H\|_\infty\, .
\label{very_small_t}
\ee
However, no such estimate can be given when only $t\ll T_E$ is assumed. For a special case, this has already been verified in \cite{V-2003b}. To see why, let us fix a (small) value of $s\coloneqq Et$, and let us study the universal function
\bb
\varphi(s) \coloneqq \sup_{H\geq 0,\ \min \spec(H) =0} \frac12 \left\| \cU_H - I \right\|_\diamond^{H,s}\, ,
\label{universal f}
\ee
where the supremum is over all self-adjoint positive operators $H\geq 0$ with $0$ in the spectrum, in either finite or infinite dimension, and we set $\cU_H(\cdot)\coloneqq e^{-iH}(\cdot) e^{iH}$.


\begin{lemma}
The universal function $\varphi$ defined by \eqref{universal f} satisfies that
\bb
\varphi(s) \geq 2 \sqrt{\frac{s (\pi + 2s)}{\pi^2 + 4\pi s + 8s^2}} = \frac{2\sqrt{s}}{\sqrt\pi} + O\left(s^{3/2}\right)\, ,
\label{universal_f_lower_1}
\ee
where the expansion on the rightmost side is for $s\to 0^+$. For $s\in [0,\pi/2]$ we have also the better bound
\bb
\varphi(s) \geq 2\sqrt{\frac{s}{\pi}\left(1-\frac{s}{\pi}\right)}\, .
\label{universal_f_lower_2}
\ee
\end{lemma}

\begin{proof}
We start by proving \eqref{universal_f_lower_1}. Consider a single harmonic oscillator with creation and annihilation operators $a^\dag$ and $a$, respectively, so that $a^\dag a$ is the photon number operator. For $\mu\in [0,1)$ and $s>0$, construct $H_{s,\mu} \coloneqq \frac{(1-\mu)s}{\mu}\, a^\dag a$, and define the state
\bb
\ket{\psi_\mu} \coloneqq \sqrt{1-\mu} \sum_{n=0}^\infty \mu^{n/2} \ket{n}\, ,
\ee
where $\ket{n}=(n!)^{-1/2} (a^\dag)^n \ket{0}$ is the $n^\text{th}$ Fock state. Note that $H_{s,\mu}\geq 0$ and $\min \spec(H_{s,\mu})=0$. Moreover, 
\bbb
\braket{\psi_\mu| H_{s,\mu} |\psi_\mu} = \frac{(1-\mu)^2 s}{\mu} \sum_{n=0}^\infty n \mu^n = s\, ,
\eee
where the last equality is deduced by summing the arithmetic--geometric series. A similar computation yields
\begin{align*}
\left|\braket{\psi_\mu| e^{-iH_{s,\mu}} |\psi_\mu}\right| &= (1-\mu) \left| \sum_{n=0}^\infty \mu^n e^{-i\, \frac{(1-\mu)s}{\mu}\, n} \right| \\
&= \frac{1-\mu}{\left| 1 - \mu\, e^{-i\, \frac{(1-\mu)s}{\mu}} \right|} \\
&= \frac{1-\mu}{\sqrt{1+\mu^2 -2\mu \cos\left( \frac{(1-\mu)s}{\mu} \right)}} \, .
\end{align*}
Therefore,
\begin{align*}
\varphi(s) &\geq \frac12 \left\| \cU_{H_{s,\mu}} - \Id \right\|_\diamond^{H_{s,\mu},\,s} \\
&\geq \frac12 \left\| \left(\cU_{H_{s,\mu}} - \Id\right)(\ket{\psi_\mu}) \right\|_1 \\
&= \sqrt{1-\left| \braket{\psi_\mu| e^{-iH_{s,\mu}} |\psi_\mu} \right|^2} \\
&= \sqrt{\frac{2\mu\left(1-\cos\left( \frac{(1-\mu)s}{\mu} \right) \right)}{1+\mu^2 -2\mu \cos\left( \frac{(1-\mu)s}{\mu} \right)}}\, .
\end{align*}
An analytical maximisation over $\mu$ turns out not to be possible. However, we can get a sufficiently good estimate of the bound by making the ansatz $\mu = \frac{2s}{2s+\pi}$, which yields precisely \eqref{universal_f_lower_1}.

To prove \eqref{universal_f_lower_2} for $0\leq s\leq \pi/2$, consider a single qubit with Hilbert space $\cH=\CC^2$, equipped with the Hamiltonian $H_{s,p}\coloneqq \frac{s}{1-p}\ketbra{1}$, where $p\in [0,1)$ is an auxiliary parameter. Define the state $\ket{\psi_p}\coloneqq \sqrt{p}\ket{0}+\sqrt{1-p}\ket{1}$, and note that $\braket{\psi_p|H_{s,p}|\psi_p}\leq s$. Hence,
\begin{align*}
\varphi(s) &\geq \frac12 \left\| \cU_{H_{s,p}} - \Id \right\|_\diamond^{H_{s,p},\,s} \\
&\geq \frac12 \left\| \left(\cU_{H_{s,p}} - \Id\right)(\psi_p) \right\|_1 \\
&= \sqrt{1-\left| \braket{\psi_p| e^{-iH_{s,p}} |\psi_p} \right|^2} \\
&= 2\sqrt{p(1-p)} \left|\sin\left( \frac{s}{2(1-p)} \right)\right| .
\end{align*}
Again, maximising in $p$ analytically is not feasible. The ansatz $p=1-\frac{s}{\pi}$ however yields \eqref{universal_f_lower_2}. This completes the proof.
\end{proof}

The above lemma shows that if only the value of $Et$ is specified, the best possible bound on the diamond norm distance $\left\|\cU_t - \Id\right\|_\diamond^{H,E}$ will necessarily contain the $\sqrt{Et}$ factor that is predicted by Theorem~\ref{thm:favard}. However, we know from \eqref{very_small_t} that for values of $t$ so small that $t \|H\|_\infty \ll 1$ the diamond norm will grow linearly in $t$. The problem with this latter statement is that it becomes empty when we consider unbounded Hamiltonians on infinite-dimensional systems. And in fact, for such Hamiltonians it can happen that the range of values of $t$ for which the growth is linear vanishes altogether! To see this, consider the following example.

\begin{lemma} \label{nasty_lemma}
Consider a single harmonic oscillator, and let the Hamiltonian be the number operator $a^\dag a$. Consider the unitary group of phase space rotations $\cU_t(\cdot)\coloneqq e^{-it\, a^\dag\! a}(\cdot)\, e^{it\, a^\dag\! a}$. For all $\delta>0$ and all $E>0$ one has that
\bb
\lim_{t\to 0^+} \frac{-\log \left\| \cU_t - \Id \right\|_\diamond^{a^\dag\! a, E}}{-\log t} \leq \frac{1+\delta}{2} .
\label{growth_exponent}
\ee
\end{lemma}

\begin{rem}
The left-hand side of \eqref{growth_exponent} is the growth exponent of the quantity $\left\| \cU_t - \Id \right\|_\diamond^{a^\dag\! a, E}$ for very small times.
\end{rem}

\begin{proof}[Proof of Lemma~\ref{nasty_lemma}]
Let $\delta>0$ and $E>0$ be given. Without loss of generality, we can assume that $E\leq \frac{\zeta(1+\delta)}{\zeta(2+\delta)}$, with $\zeta$ being the Riemann zeta function. Construct the states
\begin{align}
\ket{\psi_\delta} &\coloneqq \frac{1}{\sqrt{\zeta(2+\delta)}} \sum_{n=1}^\infty \frac{1}{n^{1+\delta/2}}\ket{n} , \\
\ket{\phi_{\delta, E}} &\coloneqq \sqrt{1-\frac{E\, \zeta(2+\delta)}{\zeta(1+\delta)}}\ket{0} + \sqrt{\frac{E\, \zeta(2+\delta)}{\zeta(1+\delta)}} \ket{\psi_\delta} .
\end{align}
Clearly,
\begin{align*}
\braket{\psi_\delta|a^\dag a|\psi_\delta} &= \frac{\zeta(1+\delta)}{\zeta(2+\delta)}<\infty , \\
\braket{\phi_{\delta, E} | a^\dag a|\phi_{\delta,E}} &= E .
\end{align*}
Now, let us compute
\bbb
\braket{\psi_\delta| e^{-it\, a^\dag a} | \psi_\delta} = \frac{1}{\zeta(2+\delta)} \sum_{n=1}^\infty \frac{e^{-itn}}{n^{2+\delta}} = \frac{1}{\zeta(2+\delta)}\, \Li_{2+\delta}(e^{-it})\, ,
\eee
where we introduced the polylogarithm $\Li_s(z) \coloneqq \sum_{n=1}^\infty \frac{z^n}{n^s}$, where the series representation is valid provided that $|z|<1$ or $|z|=1$ but $\Re s > 1$. We now use the series expansion
\bbb
\Li_s(e^\mu) = \Gamma(1-s) (-\mu)^{s-1} + \sum_{k=0}^\infty \frac{\zeta(s-k)}{k!}\, \mu^k\, ,
\eee
that can be reportedly~\cite{wiki-polylog} derived by analytical continuation from known identities~\cite[\S~9.553]{Gradshteyn2007}. Upon straightforward but tedious algebra, this yields
\bbb
\braket{\psi_\delta| e^{-it\, a^\dag a} | \psi_\delta} = 1 - \frac{\Gamma(-1-\delta)}{\zeta(2+\delta)}\, \sin\left( \frac{\pi\delta}{2}\right) t^{1+\delta} - \frac{i}{\zeta(2+\delta)} \left( \zeta(1+\delta) t - \cos\left(\frac{\pi\delta}{2}\right) \Gamma(-1-\delta)\, t^{1+\delta}\right) + O(t^2)\, ,
\eee
in turn implying that
\begin{align*}
\left(\frac12\left\|\cU_t - \Id\right\|_\diamond^{a^\dag\! a, E} \right)^2 &= \sup_{\braket{\psi|a^\dag a|\psi}\leq E} \left\{ 1 -  \left| \braket{\psi|e^{-it a^\dag a}|\psi}\right|^2 \right\} \\
&\geq 1 - \left| \braket{\phi_{\delta, E}|e^{-it a^\dag a}|\phi_{\delta, E}}\right|^2 \\
&= 1 - \left| 1-\frac{E\, \zeta(2+\delta)}{\zeta(1+\delta)} + \frac{E\, \zeta(2+\delta)}{\zeta(1+\delta)}\, \braket{\psi_\delta| e^{-it\, a^\dag a} | \psi_\delta} \right|^2 \\
&= \frac{2E\, \zeta(2+\delta)\,\Gamma(-1-\delta)}{\zeta(2+\delta)\zeta(1+\delta)}\, \sin\left( \frac{\pi\delta}{2}\right) t^{1+\delta} + O \left( t^2\right)
\end{align*}
thanks to Theorem~\ref{unentangled_thm}. The lower bound in \eqref{growth_exponent} follows immediately.
\end{proof}

What Lemma~\ref{nasty_lemma} teaches us is that in an infinite-dimensional system equipped with an unbounded Hamiltonian the growth of the quantity $\left\| \cU_t - \Id \right\|_\diamond^{a^\dag\! a, E}$ is in general never linear, not even for very small times or very small energies. Moreover, the best universal lower bound on the growth exponent is in fact $1/2$, which matches the upper bound given by Theorem~\ref{thm:favard}.

It is worth noting that the state $\ket{\phi_{\delta, E}}$ used in the proof of Lemma~\ref{nasty_lemma} has finite energy but infinite energy variance. This feature is in fact responsible for the fast growth of the norm $\left\|(\cU_t-\Id)(\phi_{\delta,E})\right\|_1$ with respect to $t$. If only states with finite energy variance are considered, one can show that a linear growth is restored~\cite{Mandelstam1945, Mandelstam1991, Bhattacharyya1983, Pfeifer1993}.


\subsection{On the possibility of `eventually perfect' discrimination}

In light of Theorem~\ref{eventual_discrimination_thm}, we could wonder whether EC perfect discrimination between any two distinct unitary groups $U_t=e^{-itK}$ and $U'_t = e^{-itK'}$ could always be achieved by simply waiting for a long enough time $t$. The following example shows that, depending on the choice of a Hamiltonian $H$ that measures the energy, this may not be the case.

\begin{ex}
Let $f_1,f_2\in L^2(\RR)$ be the two functions given by
\bbb
f_1(x) \coloneqq \frac{\kappa_1 \indic_{\text{even}}(x)}{\sqrt{1+x^2}},\qquad f_2(x) = \frac{\kappa_2 \indic_{\text{odd}}(x)}{\sqrt{1+x^2}}\, ,
\label{f1f2}
\eee
where
\bbb
\indic_{\text{even}}(x) \coloneqq \left\{ \begin{array}{ll} 1 & \text{ for }x \in [2n,2n+1], n \in \mathbb Z\, , \\[1ex]
0 & \text{ otherwise,} \end{array}\right.
\eee
$\indic_{\text{odd}}(x)=1-\indic_{\text{even}}(x)$, and the constants $\kappa_i$ are such that $\int_{-\infty}^{+\infty} dx\, |f_i(x)|^2=1$ for $i=1,2$. We denote with $\ket{f_1}$ and $\ket{f_2}$ the state vectors whose wave functions are given by \eqref{f1f2}. Also, set $F \coloneqq \Span \left\{ \ket{f_1},\, \ket{f_2} \right\}$, so that $\cH_1=L^2(\RR) = F \oplus F^\perp$. Note that the multiplication operator $|x|$ satisfies that $\dom\left(|x|^{1/2}\right)\cap F = \{0\}$. That is, all non-zero elements in $F$ have infinite energy with respect to the Hamiltonian $|x|$.

Now, define self-adjoint operators on $\cH_1$ by
\bbb
K \coloneqq \operatorname{diag}(1,0) \oplus 0\,,\qquad K'\coloneqq \frac{1}{2} \begin{pmatrix} 1 & -1 \\ -1 & 1 \end{pmatrix} \oplus 0\, ,
\eee
where the splitting is with respect to the above orthogonal decomposition of $\cH_1$. Note that $K$ and $K'$ only act non-trivially on the $F$ subspace. The EC diamond norm of the difference $\cU_t-\cU'_t$, where $\cU_t(\cdot) \coloneqq e^{-itK}(\cdot) e^{itK}$ and $\cU'_t(\cdot) \coloneqq e^{-itK'}(\cdot)e^{itK'}$, is given by
\bbb
\left\|\cU_t - \cU'_t\right\|_{\diamond}^{|x|,E} = \sup_{\braket{g||x||g}\leq E} 2\sqrt{1-\left|\braket{g|e^{itK} e^{-itK'}|g}\right|^2}\, .
\eee
This can be equal to $2$ for some appropriate choice of $t$ only if there exists a sequence of states $\left(\ket{g_n}\right)_{n\in \NN}$ such that $\braket{g_n||x||g_n}\leq E$ for all $n$ and moreover $\braket{g_n|e^{itK} e^{-itK'}|g_n} \tendsn{} 0$, which in turn implies that $\left\| (e^{-itK} - e^{-itK'})\ket{g_n} \right\| \tendsn{} \sqrt2$. Since it is not difficult to verify by an explicit computation that
\bbb
\left\| (e^{-itK} - e^{-itK'})\ket{g_n} \right\| \leq \sqrt{1-\cos(t)} \left\|\Pi_F \ket{g_n}\right\| ,
\eee
where $\Pi_F$ is the orthogonal projector onto $F$, we see that we must have $t=\pi$ and $\left\|\Pi_F \ket{g_n}\right\|\tendsn{} 1$. Since $F$ is two-dimensional, we can assume -- up to considering subsequences -- that $\Pi_F \ket{g_n} \tendsn{} \ket{g}$, where the wave function $g\in L^2(\RR)$ is normalised, i.e.\ $\int_{-\infty}^{+\infty} dx\, |g(x)|^2=1$. Standard measure-theoretic results imply that up to taking a subsequence we can further assume that $g_n \tendsn{a.e.} g$, where `a.e.' stands for `almost everywhere', i.e.\ $\lim_{n\to\infty} g_n(x) = g(x)$ for all $x\in \RR\setminus S$, with $S$ of zero Lebesgue measure. Applying Fatou's lemma now shows that
\bbb
\braket{g||x||g} = \int_{-\infty}^{+\infty}dx\, |x||g(x)|^2 \leq \liminf_{n\to\infty} \int_{-\infty}^{+\infty}dx\, |x||g_n(x)|^2 = \liminf_{n\to\infty} \braket{g_n||x||g_n} \leq E\, ,
\eee
which is a contradiction since we assumed that $\ket{g}\in F$ and $\|\ket{g}\|=1$, and non-zero elements of $F$ cannot be in the domain of $|x|^{1/2}$.

\end{ex}

\subsection{Open quantum systems}

In this section we establish estimates on the dynamics of a quantum system with dissipation governed by an unbounded Lindbladian of GKLS-type. In particular, we address the question by how much the dynamics of a closed quantum systems can possibly differ from the dynamics of an open quantum system with the same Hamiltonian part as the closed quantum system when an energy constraint is imposed.

Let $(\Lambda_t)_{t\ge 0}$ be a strongly continuous quantum dynamical semigroup (QDS) on $\cT_1(\cH)$, that is a semigroup of quantum channels $\Lambda_t$ indexed on some (time) parameter $t\ge 0$. By strong convergence, we mean that for all $\rho\in\cD(\cH)$, $\Lambda_t(\rho)\to \rho$ in trace norm, as $t\to 0$. This condition assures the existence of a (possibly unbounded) generator, call it $\cL$ and of dense domain $\dom(\cL)\subset \cT_1(\cH)$, so that for all $\rho\in\dom(\cL)$:
\begin{align*}
\|t^{-1}(\Lambda_t(\rho)-\rho)-\cL(\rho)\|_1\to 0\,~~\text{ as }t\to 0\,,
\end{align*}
In the following, we will sometimes assume that the generator $\cL$ has the standard GKLS form:
Let $G:\dom(G) \subset \cH \rightarrow \cH$ be the generator of a contraction semigroup $(P_t)_{t\ge 0}$ (i.e. $\|P_t\|\le 1$ for all $t\ge 0$) and consider (possibly unbounded) Lindblad operators $(L_l)_{l \in \mathbb N}$ with $\dom(G) \subset \dom(L_l)$ such that for all $x,y \in \dom(G):$
\[\braket{ Gx \vert y }+ \braket{ x| Gy }+ \sum_{l \in \mathbb N}
\braket{ L_l x| L_l y }=0. \]
There exists then a weak$^*$ continuous semigroup $(\Lambda_t^{\dagger})$ on the space of bounded linear operator $ \cB(\cH)$ with a generator $\mathcal L^{\dagger}$ such that for all $S \in \cB(\cH)$ and $x,y \in \dom(G)$
\[\mathcal L^{\dagger}(S)(x,y) = \braket{ Gx \vert Sy } + \sum_{l \in \mathbb N}\braket{ L_lx | SL_l y }+ \braket{ x | SG y}. \]

In order to describe an open quantum system of Lindblad-type, we take $G$ to be the formal operator $G=-\frac{1}{2} \sum_{l=1}^{\infty} L_l^{\dagger}L_l-iH.$ For $G$ to be a generator of a contraction semigroup, it suffices to assume that $-\frac{1}{2} \sum_{l=1}^{\infty} L_l^{\dagger}L_l$ is relatively $H$ bounded with $H$-bound $<1.$ Our main result of this section is the following:

\begin{manualthm}{\ref{thm:open}}
Let $\mathcal U_t(\rho)\coloneqq e^{-iHt} \rho e^{iHt}$ be the dynamics of the closed quantum system and assume that for all $\psi \in \dom(H)$ the following relative boundedness condition holds
\bb
\Vert ((-iH)-G )\ket{\psi} \Vert = \frac{1}{2} \left\Vert \sum_{l \in \mathbb N} (L_l^{\dagger}L_l) \ket{\psi}\right\Vert \le \alpha \Vert H \ket{\psi}\Vert + \beta \Vert \ket{\psi} \Vert \, ,
\label{relatbound_Lindblad_SM}
\ee
with $\alpha<1$ and $\beta \in (0,\infty)$, then it follows that the difference of the dynamics of the closed quantum system, governed by $(\mathcal U_t)_{t \ge 0}$, and the open quantum system, governed by $(\Lambda_t)_{t \ge 0},$ satisfies
\[ \|\mathcal{U}_t-\Lambda_t\|_\diamond^{\vert H\vert,E}\le 4 \left(2^{1/4} \sqrt{\alpha E t}  +  \beta t \right). \]
Moreover, this implies that for $(\mathcal U_t)_{t\in\RR}$ and $(\Lambda_t)_{t\ge 0}$ to evolve a state $\rho$ by a distance $d\coloneqq \Vert (\mathcal U_t-\Lambda_t)\rho\Vert_1$, we find the \emph{quantum speed limit}
\[t \ge \left(\frac{\sqrt{ \sqrt{2} \alpha E+ d \beta} - 2^{1/4} \sqrt{\alpha E} }{ \beta}\right)^2.\]

Finally, let $H$ be a Hamiltonian, then the QMSs $(\Lambda_t)_{t\ge 0}$ and $(\Lambda_t')_{t\ge 0}$ for two different pairs of families of bounded Lindblad operators $(L_l) $ and $(L_l')$ and generators $G=-iH-\frac{1}{2}\sum_{l} L_l^{\dagger}L_l$ and $G'=-iH-\frac{1}{2}\sum_{l} (L_l')^{\dagger}L_l'$, respectively, satisfies 
\begin{equation}
\label{eq:bounded}
\Vert  \widetilde{\Lambda'}_t-\widetilde{\Lambda}_t\Vert_\diamond \le t \sum_{l \in \mathbb N} \Bigg(\Vert L_l^{\dagger}L_l - (L_l')^{\dagger}L_l' \Vert +  \Vert L_l-L_l' \Vert (\Vert L_l \Vert+\Vert L_l' \Vert) \Bigg).
\end{equation}
\end{manualthm}

\begin{proof}
In the sequel, we write $\widetilde{X}\coloneqq X \otimes \operatorname{id}_{\mathbb{C}^n}$ for operators $X$ on $\cH$ and also $\widetilde{X}\coloneqq X \otimes \operatorname{id}_{\cB(\mathbb C^n)}$ for superoperators.

We first establish a propagation estimate where we compare the QDS $(\Lambda_t)_{t \ge 0}$ and the semigroup defined by $\mathcal V_t(S) \coloneqq  P_t^{\dagger} S P_t$. For this purpose, let $x \in \dom(\tilde{G}),$ then 
\begin{equation}
    \begin{split}
    \Vert  (\widetilde{\mathcal V}_t-\widetilde{\Lambda}_t)(\ket{x} \bra{x})\Vert_1=
 &\sup_{\Vert S \Vert=1} \braket{ x | (\widetilde{\mathcal V}_t^{\dagger}-\widetilde{\Lambda}_t^{\dagger})(S) x } 
 =\sup_{\Vert S \Vert=1} \int_0^t  \frac{d}{ds} \braket{ x |\widetilde{\mathcal V}_{t-s}^{\dagger}(\widetilde{\Lambda}_s^{\dagger}(S)) x }  \ ds \\
 &= \sup_{\Vert S \Vert=1} \Bigg(\int_0^t - \braket{ \widetilde{G} \widetilde{P}_{t-s} x | (\widetilde{\Lambda}_s^{\dagger}(S)) \widetilde{P}_{t-s} x } + \braket{ \widetilde{P}_tx | (\widetilde{\Lambda}_s^{\dagger}(S)) \widetilde{P}_{t-s} \widetilde{G} x } \ ds \\
  &\quad + \int_0^t  \braket{ \widetilde{G} \widetilde{P}_{t-s} x |(\widetilde{\Lambda}_s^{\dagger}(S)) \widetilde{P}_{t-s} x } + \braket{ \widetilde{P}_{t-s} x | (\widetilde{\Lambda}_s^{\dagger}(S)) \widetilde{G} \widetilde{P}_{t-s} x } \ ds \\
  &\quad + \sum_{l=1}^{\infty} \int_0^t   \braket{ \widetilde{L}_l \widetilde{P}_{t-s} x | (\widetilde{\Lambda}_s^{\dagger}(S)) \widetilde{L}_l \widetilde{P}_{t-s} x } \ ds\Bigg) \\
  &\le \sum_{l=1}^{\infty} \int_0^t \Vert \widetilde{L}_l \widetilde{P}_{t-s} x \Vert^2 \ ds = -2 \int_0^t \Re(\braket{ \widetilde{P}_{t-s}x |\widetilde{G}\widetilde{P}_{t-s}x }) \ ds  \\
  &= \int_0^t \frac{d}{dt}\Vert \widetilde{P}_{t-s} x \Vert^2 \ ds  = \Vert x \Vert^2 - \Vert \widetilde{P}_t x \Vert^2.
     \end{split}
\end{equation}

Hence, we conclude that for a density operator $\rho$ we have for the semigroup defined by $\widetilde{\mathcal{U}}_t(\rho)=e^{i\widetilde{H}t}\rho e^{-i\widetilde{H}t}$
\begin{equation}
    \begin{split}
\Vert (\widetilde{\mathcal U}_t-\widetilde{\Lambda}_t)(\rho)\Vert_1 &\le\Vert (\widetilde{\mathcal U}_t- \widetilde{\mathcal V}_t)(\rho)\Vert_1+\Vert ( \widetilde{\mathcal V}_t-\widetilde{\Lambda}_t)(\rho)\Vert_1\\
&\le 2\Vert (e^{-i\widetilde{H}t}-\widetilde{P}_t) \rho \Vert_1 + (1-\operatorname{tr}(\widetilde{P}_t \rho \widetilde{P}_t^{\dagger}) ) 
= 2\Vert (e^{-i\widetilde{H}t}-\widetilde{P}_t) \rho \Vert_1 +\operatorname{tr}(e^{-i\widetilde{H}t} \rho e^{i\widetilde{H}t}-\widetilde{P}_t \rho \widetilde{P}_t^{\dagger}) \\
&\le 4 \Vert (e^{-i\widetilde{H}t}-\widetilde{P}_t) \rho \Vert_1  \le 4 \int_0^t \left\Vert (-i\widetilde{H}-\widetilde{G})e^{-i\widetilde{H}s} \rho \right\Vert_1 \ ds = 2 \int_0^t \left\Vert \sum_{l=1}^{\infty}  \widetilde{L}_l^{\dagger}\widetilde{L}_l e^{-i\widetilde{H}s} \rho \right\Vert_1 \ ds.
\end{split}
\end{equation}

Now, let us decompose $\rho = \lambda (\lambda+i\tilde{ H})^{-1}\rho+i\tilde H(\lambda + i\tilde {H})^{-1} \rho\eqqcolon \rho_{\lambda}-\sigma_{\lambda}.$
We notice that by \cite[Theorem~7.1.20]{Si15}, the relative boundedness of $(-iH)-G$ with respect to $H$ implies the relative boundedness of $(-i\widetilde{H})-\widetilde{G}$ with respect to $\widetilde{H}$ with the same coefficient.
Hence, we find, using the spectral decomposition $\rho_{\lambda} = \sum_{i=1}^{\infty} \lambda_i \ketbra{\varphi_i},$
\begin{equation}
    \begin{split}
    \label{eq:firstone2}
\Vert (\widetilde{\mathcal U}_t-\widetilde{\Lambda}_t)\rho_{\lambda}\Vert_1 
&\le 4 \int_0^t \Vert \sum_{l=1}^{\infty}\frac{1}{2}\widetilde{L}_l^{\dagger}\widetilde{L}_l e^{-i\widetilde{H}s} \rho_{\lambda} \Vert_1 \ ds \\
&\le 4 \sup_{s \in [0,t]}\Vert  \lambda^{1/2} \frac{1}{2}  \sum_{l=1}^{\infty}\widetilde{L}_l^{\dagger}\widetilde{L}_l (\lambda+i\widetilde{H})^{-1}e^{-i\widetilde{H}s} \rho \Vert_1  \lambda^{1/2} t \\
&\le 4 \lambda^{1/2} \sup_{s \in [0,t]}\sqrt{\operatorname{tr} \left(  \frac{1}{2}  \sum_{l=1}^{\infty}\widetilde{L}_l^{\dagger}\widetilde{L}_l (\lambda+i\widetilde{H})^{-1}e^{-i\widetilde{H}s} \rho e^{i\widetilde{H}s} (\lambda-i\widetilde{H})^{-1} \frac{1}{2}  \sum_{l=1}^{\infty}\widetilde{L}_l^{\dagger}\widetilde{L}_l  \right)}  \lambda^{1/2} t \\
&\le 4 \lambda^{1/2} \sup_{s \in [0,t]}\sqrt{\sum_{i=1}^{\infty} \lambda_i  \left\Vert \frac{1}{2} \sum_{l=1}^{\infty}\widetilde{L}_l^{\dagger}\widetilde{L}_l (\lambda+i\widetilde{H})^{-1}e^{-i\widetilde{H}s} \ket{\varphi_i} \right\Vert^2  }  \lambda^{1/2} t \\
&\le 4 \lambda^{1/2} \sup_{s \in [0,t]}\sqrt{\sum_{i=1}^{\infty} \lambda_i  \left(\alpha \left\lVert \widetilde{H}(\lambda+i\widetilde{H})^{-1} \ket{\varphi_i} \right\rVert+ \lambda^{-1} \beta \right)^2 }  \lambda^{1/2} t \\
&\le 4 \left( \sqrt{\sum_{i=1}^{\infty} \lambda_i  \alpha ^2 \lambda \Vert \widetilde{H} (\lambda+i\widetilde{H})^{-1} \ket{\varphi_i} \Vert^2}  \lambda^{1/2} t +  \beta t\right) \\
&\le 4\left( \alpha  \sqrt{\operatorname{tr}(\vert \widetilde{H} \vert \rho)} \lambda^{1/2} t + \beta t\right).
    \end{split}
\end{equation}

On the other hand, 
\begin{equation}
    \begin{split}
    \label{eq:secondone2}
\Vert ( \widetilde{\mathcal U}_t-\widetilde{\Lambda}_t)\sigma_{\lambda}\Vert_1 \le 2 t \Vert \sigma_{\lambda}\Vert_1 \le 2 \left(\sqrt{\sum_{i=1}^{\infty}\lambda_i \sup_{\lambda>0} \lambda \Vert  \tilde H(\lambda+i\tilde H)^{-1} \ket{\varphi_i} \Vert^2}\right) \lambda^{-1/2} \le \sqrt{2\operatorname{tr}(\tilde{\vert H \vert}\rho)}\lambda^{-1/2}
    \end{split}
\end{equation}
Choosing $\lambda=\frac{1}{2 \sqrt{2}\alpha t}$, we find from combining \eqref{eq:firstone2} with \eqref{eq:secondone2} 
\[ \Vert (\widetilde{\mathcal U}_t-\widetilde{\Lambda}_t)\rho\Vert \le 4 \left(2^{1/4} \sqrt{\alpha \operatorname{tr}(\vert \widetilde{H} \vert \rho) t}  +  \beta t \right). \]
To see \eqref{eq:bounded}, we use that for the basis expansion $\widetilde{\Lambda}_{t-s}'(\rho) = \sum_{n=1}^{\infty} \lambda_n(s) \ketbra{\varphi_n(s)}$
\begin{equation}
    \begin{split}
\Vert  (\widetilde{\Lambda'}_t-\widetilde{\Lambda}_t)(\rho)\Vert_1 &=\sup_{\Vert S \Vert=1} \tr\left( (\widetilde{\Lambda'}_t^{\dagger}-\widetilde{\Lambda}_t^{\dagger})(S) \rho \right) 
 =\sup_{\Vert S \Vert=1} \int_0^t  \frac{d}{ds} \tr\left( \widetilde{\Lambda'}_{t-s}^{\dagger}(\widetilde{\Lambda}_s^{\dagger}(S)) \rho\right)  ds \\
 &= \sup_{\Vert S \Vert=1} \sum_{n=1}^{\infty}  \Bigg(\int_0^t \lambda_n(s)\braket{ \varphi_n(s) | \widetilde{\Lambda}_s(S) (\widetilde{G}-\widetilde{G}') \varphi_n(s) } \\
  &\qquad + \braket{ (\widetilde{G}-\widetilde{G}') \varphi_n(s) | \widetilde{\Lambda}_s(S) (\widetilde{G}-\widetilde{G}') \varphi_n(s) } \\
  &\qquad + \sum_{l \in \mathbb N}\left(\braket{ (\widetilde{L}_l-\widetilde{L}_l')\varphi_n(s) | \widetilde{\Lambda}_s(S) \widetilde{L}_l \varphi_n(s)}+\braket{ \widetilde{L}_l' \varphi_n(s)| \widetilde{\Lambda}_s(S) (\widetilde{L}_l-\widetilde{L}_l') \varphi_n(s) }\right) \Bigg) ds \\
  &\le t \sum_{l \in \mathbb N} \Bigg(\Vert L_l^{\dagger}L_l - (L_l')^{\dagger}L_l' \Vert +  \Vert L_l-L_l' \Vert (\Vert L_l \Vert+\Vert L_l' \Vert) \Bigg).
    \end{split}
\end{equation}
\end{proof}

\begin{rem}
The relative boundedness condition \eqref{relatbound_Lindblad_SM} ensures that the Lindblad operators $L_\ell$ do not induce very fast transitions from low-energy subspaces of the Hamiltonian. If that were the case, then it would be possible to discriminate the unitary evolution from the open-system dynamics even at very short evolution times by simply preparing the ground state of $H$ and then testing whether the evolved system is still in the same state. 
\end{rem}

We continue by giving some applications of Theorem \ref{thm:open}. Since the following example depends on the precise ratio of masses we include physical examples in the following example:

\begin{ex}[Linear quantum Boltzmann equation~\cite{A02, HV2}]
\label{lqb}
Consider a particle with mass parameter $M$ and a closed quantum system described by the Hamiltonian $H_0 = -\frac{\hbar^2}{2M}\Delta + V.$ 
The linear quantum Boltzmann equation describes the motion of this particle in the presence of an additional ideal gas of particles with mass $m$ distributed as $\mu_{\beta}(p) = \frac{1}{\pi^{3/2} p_{\beta}^3} e^{- \left\lvert p \right\rvert^2 /p_{\beta}^2}$ where $p_{\beta} = \sqrt{2m/\beta}.$ 

In addition, we assume here the Born approximation of scattering theory \cite{HV2}: Let $m_{\text{red}} = mM/(m+M)$ be the reduced mass and $n_{\text{gas}}$ the density of gas particles. We assume that the scattering potential between the gas particles and the single particle is of short-range and smooth such that $V$ is a Schwartz function. The scattering amplitude is then $f(p) = -\frac{m_{\text{red}}}{2\pi \hbar^2} \mathcal F(V)(p/\hbar),$ where $\mathcal F$ denotes the Fourier transform.  

The ideal gas causes both an energy shift $H_{\text{per}} = -2\pi \hbar^2 \tfrac{n_{\text{gas}}}{m_{\text{red}}} \Re(f(0))$ such that the full Hamiltonian reads $H = H_0 + H_{\text{per}}$ and also leads to an additional dissipative part \cite{HV}. Using the standard momentum operator $P = -i \hbar \nabla_x$, we can then introduce operators 
\begin{equation}
L(P,k) = \sqrt{\sqrt{\frac{\beta m}{2\pi}} \frac{n_{\text{gas}}}{m_{\text{red}} \left\lvert k \right\rvert}} f(-k) \exp\left(-\beta \frac{\left((1+\tfrac{m}{M})\left\lvert k \right\rvert^2+2\tfrac{m}{M} \braket{ P| k } \right)^2 }{16m \left\lvert k \right\rvert^2} \right).
\end{equation}

The linear quantum Boltzmann equation describing the evolution of the particle state $\rho$ is then
\[ \frac{d}{dt} \rho(t) = -i [H,\rho(t)] + \int_{\mathbb R^3} \left(e^{i \braket{ k | x }}L(P,k)\rho L(P,k)^{\dagger}e^{-i \braket{ k | x }} - \frac{1}{2}\{ \rho,L(P,k)^{\dagger}L(P,k) \}  \right) \ dk.\]
The perturbation $H_{\text{per}}$ and due to
\[ \int_{\mathbb R^3} \left\lVert L(P,k)^{\dagger}L(P,k) \right\rVert \ dk < \infty, \]
the dissipative part are both bounded such that the dynamics of the linear quantum Boltzmann equation can be easily compared to the asymptotics of the closed system governed by the equation $\frac{d}{dt} \rho(t)= -i [H_0,\rho(t)].$ Since the dissipative part is bounded, Theorem \ref{thm:open} implies that the difference between the open quantum dynamics and the closed quantum dynamics described by the Hamiltonian $H_0$ is $\mathcal O(t).$
\end{ex}

\begin{ex}[Damped and pumped harmonic oscillator~\cite{A02a}]
The closed-system dynamics shall just be described by a  rescaled number operator $H= \zeta a^{\dagger}a$ for some $\zeta>0$. We then consider in addition damping $V(\rho)\coloneqq  \gamma_{\downarrow} a \rho a^{\dagger}$ and pumping $W(\rho)\coloneqq \gamma_{\uparrow} a^{\dagger}\rho a$ operators with transition rates $\gamma_{\downarrow},\gamma_{\uparrow}\ge 0.$ The physical processes of damping and pumping the system can then be described by Lindblad operators $L_{\downarrow}\coloneqq \sqrt{\gamma_{\downarrow}} a $ and $L_{\uparrow}\coloneqq \sqrt{\gamma_{\uparrow}} a^{\dagger}.$
The operator $K=-\frac{1}{2}\left(L_{\downarrow}^{\dagger}L_{\downarrow} + L_{\uparrow}^{\dagger}L_{\uparrow}\right)$ is then dissipative and self-adjoint, such that Theorem \ref{thm:open} applies, if the transition rates are assumed to be sufficiently small. Theorem \ref{thm:open} implies that the difference between the open quantum dynamics and the closed quantum dynamics described by the Hamiltonian $H$ is $\mathcal O(\sqrt{t}+t).$
\end{ex}
Next, we study the evolution of quantum particles under Brownian motion which is obtained as the diffusive limit of the quantum Boltzmann equation, cf. \ref{lqb} \cite[Section~5]{HV2}.

\begin{ex}[Quantum Brownian motion~\cite{Arnold2004, Vacchini2002}]
We take as a Hamiltonian $H=-\frac{d^2}{dx^2}+x^2$ the harmonic oscillator and consider as Lindblad operators modified creation and annihilation operators $L_j=\gamma_j x+ \beta_j \frac{d}{dx}$ for $\gamma_j,\beta_j \in \mathbb{C}.$ The dynamics of a particle undergoing a quantum Brownian motion is then described by the following Lindblad equation
\begin{equation}
\begin{split}
\partial_t \rho 
&= -i[H,\rho]+i(\lambda/2)\left([p,\left\{x,\rho\right\}]-[x,\left\{p,\rho\right\}]\right) - D_{pp}[x,[x,\rho]]-D_{xx} [p,[p,\rho]]\\
&\quad +D_{xp}[p,[x,\rho]]+D_{px}[x,[p,\rho]]
\end{split}
\end{equation}
with diffusion parameters $D_{xx} = \tfrac{\left\lvert \gamma_1 \right\rvert^2+\left\lvert \gamma_2 \right\rvert^2}{2}$, $D_{pp} = \tfrac{\left\lvert \beta_1 \right\rvert^2+\left\lvert \beta_2 \right\rvert^2}{2}$, $D_{xp}=D_{px}=-\Re \tfrac{\gamma_1^{\dagger}\beta_1+\gamma_2^{\dagger}\beta_2}{2}$, and $\lambda= \Im\left(\gamma_1^{\dagger}\beta_1+ \gamma_2^{\dagger}\beta_2 \right).$ The auxiliary operator $K=-\tfrac{1}{2} \sum_{j=1}^2 L_j^{\dagger}L_j$ is then relatively $H$-bounded and, assuming parameters $\gamma_i,\beta_i$ are sufficiently small, the operator $G=iH-K$ is the generator of a contraction semigroup on $\dom(H)$. As in the previous example, Theorem \ref{thm:open} implies that the difference between the open quantum dynamics and the closed quantum dynamics described by the Hamiltonian $H$ is $\mathcal O(\sqrt{t}+t).$
\end{ex}

\begin{ex}[Quantum optics / Jaynes-Cummings model~\cite{CGQ03}]
Quantum systems that couple a harmonic oscillator to another two-level systems are common toy examples in quantum optics and often referred to as \emph{Jaynes-Cummings models}. 
One example of a Jaynes-Cummings model is the coupling of a two-state ion to a harmonic trap with trapping strength $\nu>0$. In addition, for detuning parameter $\Delta$ and \emph{Rabi} frequency $\Omega$, the Lindblad equation with Hamiltonian
\[H =I_{\mathbb C^2} \nu a^{\dagger}a + \frac{\Delta}{2} \sigma_z - \frac{\Omega}{2} \left( \sigma_+ + \sigma_- \right) \sin\left(\eta (a+a^{\dagger}) \right),\]
where $\eta$ is the \emph{Lamb-Dicke} parameter, and with Lindblad operators $L= \sqrt{\Gamma} \sigma_{-}, L^{\dagger}= \sqrt{\Gamma}\sigma_{+}$ has been introduced in \cite{CBPZ} for this model. The parameter $\Gamma$ models the decay rate of the excited state of the ion.
The underlying Hilbert space can therefore be taken as $\ell^2(\mathbb N) \otimes \mathbb C^2$ and as the Lindblad operators are just bounded operators, all conditions of Theorem \ref{thm:open} are trivially satisfied.
The boundedness of the Lindblad operators implies therefore that by Theorem \ref{thm:open} that the difference between the open quantum dynamics and the closed quantum dynamics described by the Hamiltonian $H$ is $\mathcal O(t).$

In greater generality, various models of quantum optics can be cast in the following form \cite{CGQ03}:
As the Hamiltonian part $H$ we take for matrices $h_j \in \mathbb C^{M \times M}$ 
\[ H =  \left(h_j \prod_{k=1}^N (a_k^{\dagger})^{n_k} (a_k)^{m_k} + \operatorname{H.a.} \right)\]
on a Hilbert space $\mathcal{H} = \ell^2(\mathbb N)^{\otimes N} \otimes \mathbb C^M;$ (In the above, $H.a.$ stands for 'Hermitian adjoint'). The Lindblad operators are also rescaled creation and annihilation operators of the form $L_k= \lambda_k a_k$ or $L_k = \lambda_k a_k^{\dagger}$, i.e. $a_k$ is the annihilation operator acting on the $k$-th factor of the tensor product $\ell^2(\mathbb N)^{\otimes N}$ and $\lambda_k \ge 0$ is assumed to be a positive semi-definite matrix on $\mathbb C^M.$

To conclude, operators $-\frac{1}{2} L_k^{\dagger}L_k$ are self-adjoint and dissipative and thus for a large class of Hamiltonians $H$ the asymptotics of Theorem \ref{thm:open} is applicable and yields a $\mathcal O(\sqrt{t}+t)$ estimate on the difference between the open and closed system quantum dynamics described by $H$, only.
\end{ex}

\section{Energy-constrained discrimination of Gaussian unitary channels}
Throughout this section, we will bound the EC diamond norm distance between two Gaussian unitary channels, with the energy being computed with respect to the total photon number Hamiltonian $N$ of \eqref{total_photon_number}.

\subsection{Displacement unitaries}
 Here, the main off-the-shelf result that we will use to deduce our upper bounds has been found by Pfeifer~\cite{Pfeifer1993}. An important consequence can be phrased in our language as follows.

\begin{thm} \label{Pfeifer_thm}
For a self-adjoint operator $K$ on a Hilbert space $\cH$, let $\cU_t(\cdot) \coloneqq e^{-itK}(\cdot)e^{itK}$ be the corresponding channel unitary group. Let $H$ be a Hamiltonian that satisfies $K^2\leq \gamma H+\delta I$ for some constants $\gamma,\delta\in \RR$. Then, for all $t,t'\in \RR$ and for all $E>0$ it holds that
\bb
\frac12 \left\|\cU_t - \cU_{t'}\right\|_\diamond^{H,E} \leq \sin\left( \min\left\{|t-t'|\sqrt{\gamma E+\delta},\, \frac{\pi}{2}\right\}\right) .
\label{Pfeifer}
\ee
\end{thm}

\begin{proof}
Thanks to Theorem~\ref{unentangled_thm}, we need only to compute $\inf_{\braket{\psi|H|\psi}\leq E} \left| \braket{\psi|e^{i(t-t') K}|\psi} \right|$. Thanks to a result by Pfeifer~\cite[Eq.~(3b)]{Pfeifer1993}, for all states $\ket{\psi}$ and $s\in \RR$ it holds that
\bbb
\left| \braket{\psi|e^{isK}|\psi} \right| \geq \cos\left(\min\left\{ |s|\, \Delta_\psi K,\, \frac{\pi}{2}\right\} \right) ,
\eee
where $\Delta_\psi K \coloneqq \sqrt{\braket{\psi|K^2|\psi} - \braket{\psi|K|\psi}^2}$ is the standard deviation of $K$ on $\ket{\psi}$. By hypothesis, however,
\bbb
\Delta_\psi K \leq \sqrt{\braket{\psi|K^2|\psi}} \leq \sqrt{\gamma \braket{\psi|H|\psi} + \delta}\, .
\eee
Therefore,
\bbb
\inf_{\braket{\psi|H|\psi}\leq E} \left| \braket{\psi|e^{i(t-t') K}|\psi} \right| \geq \cos\left( \min\left\{ |t-t'| \sqrt{\gamma E+\delta},\, \frac{\pi}{2} \right\} \right) .
\eee
Plugging this into \eqref{unentangled} proves the claim.
\end{proof}

In what follows, for $z\in \RR^{2m}$ we denote with
\bb
\cD_z(\cdot) \coloneqq \D(z)(\cdot) \D(z)^\dag
\label{displacement_channel}
\ee
the `displacement channel' corresponding to the displacement unitary operator $\D(z) \coloneqq e^{- i \sum_j (\Omega z)_j R_j}$ defined in~\eqref{D}. Our purpose here is that of estimating the EC diamond norm distance between displacement channels with respect to the total photon number Hamiltonian \eqref{total_photon_number}. To this end, let us introduce the notion of \emp{squeezing operator}. For a \emp{squeezing vector} $r\in \RR^m$, set 
\bb
S(r) \coloneqq \exp \left[ \frac{i}{2} \sum_j r_j (x_j p_j + p_j x_j) \right] = S(-r)^\dag\, .
\ee
The action of $S(r)$ on the canonical operators can be written as
\bb
S(r)^\dag x_j S(r) = e^{-r_j} x_j\, ,\qquad S(r)^\dag p_j S(r) = e^{r_j} p_j\, .
\label{action_squeezing}
\ee
The \emp{squeezed vacuum states} are given by
\bb
\ket{\zeta_r}\coloneqq S(r) \ket{0} = \bigotimes_j \left( \frac{1}{\sqrt{\cosh(r_j)}} \sum_{n=0}^\infty \frac{(-1)^n}{2^n} \sqrt{\binom{2n}{n}} \tanh(r_j)^n \ket{2n} \right) .
\ee
They satisfy
\bb
\braket{\zeta_r|x_j^2|\zeta_r} = e^{-2r_j}\, ,\qquad \braket{\zeta_r|p_j^2|\zeta_r} = e^{2r_j}\, .
\label{variances_squeezed}
\ee
Before we establish our bounds on the EC diamond norm distance between displacement channels, we need a couple of lemmata. 

\begin{lemma} \label{optimal_dominance_lemma}
Let $\alpha, \beta\in \RR$ be real coefficients. Then the inequality $\alpha \left(\frac{x^2+p^2}{2}-\frac12 \right) + \beta I \geq p^2$ between operators acting on $\cH_1=L^2(\RR)$ holds if and only if $\alpha\geq 2$ and $2\beta\geq \alpha-\sqrt{\alpha(\alpha-2)}$. The same is true if one exchanges $p$ and $x$.
\end{lemma}

\begin{proof}
First of all, evaluating both sides on a highly squeezed state $\ket{\zeta_r}$, where $\RR \ni r\to \infty$, and using \eqref{variances_squeezed}, yields the necessary condition $\frac{\alpha}{2} \geq 1$. Therefore, from now on we assume that $\alpha\geq 2$. We can transform the inequality into an equivalent form by conjugating by the squeezing operator $S(r_0)$, where $r_0\coloneqq \frac14 \ln \left( \frac{\alpha}{\alpha-2}\right)$. We obtain that
\begin{align*}
    0 &\leq \left(\frac{\alpha}{2}-1\right) S(r_0)^\dag p^2 S(r_0) + \frac{\alpha}{2}\, S(r_0)^\dag x^2 S(r_0) + \left(\beta - \frac{\alpha}{2} \right) I \\
    &= \left(\frac{\alpha}{2}-1\right) \left( S(r_0)^\dag p S(r_0)\right)^2 + \frac{\alpha}{2}\, \left( S(r_0)^\dag x S(r_0)\right)^2 + \left(\beta - \frac{\alpha}{2} \right) I \\
    &= \left(\frac{\alpha}{2}-1\right) e^{2r_0} p^2 + \frac{\alpha}{2}\,e^{-2r_0} x^2 + \left(\beta - \frac{\alpha}{2} \right) I \\
    &= \frac12\sqrt{\alpha(\alpha-2)}\, \left(x^2+p^2\right) + \left(\beta - \frac{\alpha}{2} \right) I\, .
\end{align*}
The positivity of the operator on the last line is thus equivalent to the condition that $\frac12\sqrt{\alpha(\alpha-2)} + \beta - \frac{\alpha}{2}\geq 0$. The proof with $x$ and $p$ exchanged is totally analogous.
\end{proof}

Incidentally, one can immediately deduce the following corollary.

\begin{cor} \label{optimal_variance_cor}
For all $E\geq 0$ it holds that
\bb
\sup_{\braket{\psi|a^\dag a|\psi} \leq E} \braket{\psi|p^2|\psi} = \frac12 \left( \sqrt{E} + \sqrt{E+1}\right)^2\, ,
\label{optimal_variance}
\ee
where the supremum is over all normalised states $\ket{\psi}\in \cH_1$ whose mean photon number does not exceed $E$.
\end{cor}

\begin{proof}
Thanks to Lemma \ref{optimal_dominance_lemma} and remembering that $a^\dag a = \frac{x^2+p^2}{2}-\frac12$, we see that
\bbb
\sup_{\braket{\psi|a^\dag a|\psi} \leq E} \braket{\psi|p^2|\psi} \leq \inf_{\alpha\geq 2} \left\{ \alpha E + \frac{\alpha}{2} - \frac12\sqrt{\alpha(\alpha-2)} \right\} = \frac12 \left( \sqrt{E} + \sqrt{E+1}\right)^2\, .
\eee
To show that this bound is tight, it suffices to verify that it is achieved for $\ket{\psi}=\ket{\zeta_r}$, with $r=\ln \left( \sqrt{E}+\sqrt{E+1}\right)$.
\end{proof}

We are now ready to prove the following.

\begin{prop} \label{displacements_EC_diamond_distance}
For $z,w\in \RR^{2m}$, let $\cD_z, \cD_w$ denote the displacement channels defined by \eqref{displacement_channel}. Then for all $E\geq 0$ we have that
\bb
\sqrt{1-e^{-\frac12 \left( \sqrt{E}+\sqrt{E+1}\right)^2 \|z-w\|^2}} \leq \frac12 \left\|\cD_z - \cD_w\right\|_{\diamond}^{N,E} \leq \sin\left( \min\left\{\frac{1}{\sqrt2}\, \|z-w\|\left( \sqrt{E}+\sqrt{E+1} \right),\, \frac{\pi}{2}\right\}\right) \, ,
\tag{\ref{displacements_EC_diamond_distance}}
\ee
where $N$ is the total photon number Hamiltonian of \eqref{total_photon_number}. In particular, for $\|z-w\|\to 0$ we have that
\bb
\frac12 \left\|\cD_z - \cD_w\right\|_{\diamond}^{N,E} = \frac{1}{\sqrt2}\, \|z-w\|\left( \sqrt{E}+\sqrt{E+1} \right) + O\left( \left(\left( \sqrt{E}+\sqrt{E+1}\right) \|z-w\| \right)^2 \right) .
\ee
\end{prop}

\begin{proof}
Thanks to Theorem~\ref{unentangled_thm}, we have that
\begin{align*}
    \frac12 \left\|\cD_z - \cD_w\right\|_{\diamond}^{N,E} &= \sup_{\braket{\psi|N|\psi}\leq E} \sqrt{1-\left|\braket{\psi|\D(-w)\D(z)|\psi}\right|^2} \\
    &= \sup_{\braket{\psi|N|\psi}\leq E} \sqrt{1-\left|\braket{\psi|\D(z-w)|\psi}\right|^2} \\
    &= \sqrt{1 - \inf_{\braket{\psi|N|\psi}\leq E} \left|\braket{\psi|\D(u)|\psi}\right|^2}\, ,
\end{align*}
where in the last line we set $u\coloneqq z-w$. We can simplify the above expression by performing a passive Gaussian unitary on $\ket{\psi}$. Since passive Gaussian unitaries commute with $N$, doing this does not affect the energy constraint, and amounts to a transformation $u\mapsto Ku$, where $K\in \operatorname{SO}_{2m}(\RR)\bigcap \symp_{2m}(\RR)$ is an orthogonal symplectic matrix. Since this action is well-known to be transitive (see e.g.~\cite[Lemma~13]{assisted-Ryuji}), we will henceforth assume without loss of generality that $u = \left(0,\|u\|,0,\ldots, 0\right)^\intercal$, i.e.\ that $\D(u) = e^{-i\|u\|p_1}$.

To upper bound $\inf_{\braket{\psi|N|\psi}\leq E} \left|\braket{\psi|e^{-i\|u\|p_1}|\psi}\right|^2$, start by taking as ansatz the squeezed state $\ket{\psi}=\ket{\zeta_r}$, where $r=(r_1,0,\ldots, 0)\in \RR^m$, and $\sinh^2(r_1) = \braket{\zeta_r|N|\zeta_r} = E$, that is, $r_1 = \ln\left( \sqrt{E} + \sqrt{E+1} \right)$. We obtain that
\begin{align*}
    \inf_{\braket{\psi|N|\psi}\leq E} \left|\braket{\psi|e^{-i\|u\|p_1}|\psi}\right|^2 &\leq \left|\braket{\zeta_r|e^{-i\|u\|p_1}|\zeta_r}\right|^2 \\
    &= \left|\braket{0| S(r)^\dag e^{-i\|u\|p_1} S(r)|0}\right|^2 \\
    &= \left|\braket{0| e^{-i\|u\| S(r)^\dag p_1 S(r)} |0}\right|^2 \\
    &= \left|\braket{0| e^{-i\|u\| e^{r_1} p_1} |0}\right|^2 \\
    &= e^{-\frac12 e^{2r_1} \|u\|^2} \\
    &= e^{-\frac12 \left( \sqrt{E} + \sqrt{E+1} \right)^2 \|u\|^2}\, .
\end{align*}
This proves the lower bound on $\frac12 \left\|\cD_z - \cD_w\right\|_{\diamond}^{N,E}$ in \eqref{displacements_EC_diamond_distance}.

To prove the upper bound, we use Theorem~\ref{Pfeifer_thm} in conjunction with our Lemma~\ref{optimal_dominance_lemma}. Setting $K=\|u\|p_1$ and $H=N$, Lemma~\ref{optimal_dominance_lemma} guarantees that $K^2\leq \gamma H+\delta$ holds with $\gamma=\|u\|^2 \alpha$ and $\delta=\frac{\|u\|^2}{2} \left( \alpha -\sqrt{\alpha(\alpha-2)}\right)$ for all $\alpha\geq 2$. Hence, by~\eqref{Pfeifer}
\begin{align*}
\frac12 \left\|\cD_z - \cD_w\right\|_\diamond^{H,E} &\leq \inf_{\alpha\geq 2} \sin\left( \min\left\{\|z-w\|\sqrt{\alpha E+\frac{\alpha}{2} - \frac12 \sqrt{\alpha(\alpha-2)}},\, \frac{\pi}{2}\right\}\right) \\
&= \sin\left( \min\left\{\|z-w\|\inf_{\alpha\geq 2} \sqrt{\alpha E+\frac{\alpha}{2} - \frac12 \sqrt{\alpha(\alpha-2)}},\, \frac{\pi}{2}\right\}\right) \\
&= \sin\left( \min\left\{\frac{1}{\sqrt2}\, \|z-w\|\left( \sqrt{E}+\sqrt{E+1} \right),\, \frac{\pi}{2}\right\}\right) .
\end{align*}
This proves the upper bound in \eqref{displacements_EC_diamond_distance}.
\end{proof}

\subsection{Symplectic unitaries}
Now, let us turn our attention to the other important class of Gaussian unitary channels, namely symplectic unitaries. First, we use the well-known polar decomposition of symplectic matrices (see \cite{pramana, heinosaari2009semigroup,HOLEVO-CHANNELS-2}): 
\begin{lemma}\label{lemmapolardecomp}
for any $S\in\operatorname{Sp}_{2m}(\RR)$, there exists a positive symplectic matrix $P\in\operatorname{Sp}_{2m}(\RR)$, as well as an orthogonal symplectic matrix $O$, such that 
\begin{align*}
    S=PO\,.
\end{align*}
Moreover, both $O$ and $P$ belong to $\operatorname{exp}(\mathfrak{sp}_{2m}(\RR))$, where $\mathfrak{sp}_{2m}(\RR)$ denotes the symplectic Lie algebra. 
\end{lemma}
The above lemma will be combined with the following corollary of Theorem \ref{thm:favard} in order to derive a bound on the energy constrained diamond norm difference between two symplectic unitaries:

\begin{cor}\label{coroll}
For any $S\in\operatorname{exp}(\mathfrak{sp}_{2m}(\RR))$, it holds that
\begin{align*}
    \|\mathcal{U}_{S}-\Id\|_\diamond^{N,E}\le 2\,\sqrt{\,({\sqrt{6}+\sqrt{10}+5\sqrt{2}m})\,\|\ln(S)\|_2\,(E+1)} \,.
\end{align*}

\end{cor}

\begin{proof}
Thanks to Theorem \ref{thm:favard} applied with $H_0=N$, we see that for any quadratic Hamiltonian of the form $H=\sum_{jk}\left(X_{jk} a_j^\dagger a_k+Y_{jk}a_ja_k+Y_{jk}^{\dagger}a_j^\dagger a_k^\dagger\right)$ with $X=X^\dagger$, the energy constrained diamond norm between the unitary conjugation $\mathcal{U}_{X,Y}(.)\coloneqq e^{-iH}(.)e^{iH}$ (i.e.\ for $t=1$) and the identity superoperator $\Id$ (i.e.\ for $H'=0$) is upper bounded as follows:
\begin{align}\label{boundUXY}
     \|\mathcal{U}_{X,Y}-\Id\|_\diamond^{N,E}\le 2\sqrt{2}\,\,\sqrt{\alpha\,E+\beta}
\end{align}
where $\alpha$ and $\beta$. Here, we recall that the constants $\alpha$ and $\beta$ have to satisfy 
\begin{align*}
    \langle \psi||H||\psi\rangle\le \alpha\,\langle \psi|N|\psi\rangle+\beta\,\||\psi\rangle\|^2\,.
\end{align*}
Let us first show the stronger relative $N$-boundedness of $H$ by a slight adaptation of the calculations in the proof of Corollary \ref{propgaussexample}: for all $|\psi\rangle\in \dom(N)$,
\bb
\begin{aligned}
  \Vert H\ket{\psi} \Vert &\le \|X\|_2\, \sqrt{\sum_{j,k=1}^m \Vert  a_j^{\dagger} a_k \ket{\psi} \Vert^2} +\|Y\|_2\left(\sqrt{\sum_{j,k=1}^m \Vert  a_j a_k \ket{\psi} \Vert^2} + \sqrt{\sum_{j,k=1}^m \Vert  a_j^{\dagger} a_k^{\dagger} \ket{\psi} \Vert^2}\right) \\
  &\le \|X\|_2\Big(\,\frac{3}{2}\|N|\psi\rangle\|^2+\frac{(m-1)^2}{2}\,\|\psi\|^2\,\Big)^\frac{1}{2}+\|Y\|_2\,\Big(\,\frac{5}{2}\|N|\psi\rangle\|^2+\Big(\frac{(2m+1)^2}{2}+2m^2\Big)\||\psi\rangle\|^2\Big)^{\frac{1}{2}} \\
  &\le \frac{1}{\sqrt{2}}\Big(\sqrt{3}\,\|X\|_2+\sqrt{5}\|Y\|_2\Big)\,\|N|\psi\rangle\|+\frac{1}{\sqrt{2}}\Big((m-1)\,\|X\|_2+(4m+1)\|Y\|_2\Big)\,\||\psi\rangle\| \\
  &\equiv a\|N|\psi\rangle\|+b\,\||\psi\rangle\|\,.
\end{aligned}
\label{eq:Nboundedness}
\ee
Now the $N$-boundedness derived in \eqref{eq:Nboundedness} implies the relative form-boundedness of $H$ with respect to $N$~\cite[Theorem~X.18]{reed1975ii}: for all $\mu>0$ and any $\ket{\psi} \in \dom(N)$,
\begin{align*}
\braket{ \psi  |H|\psi } \le \big(a+\frac{b}{\mu}\big)\,\braket{ \psi |N|\psi}+(\mu a+b)\braket{ \psi | \psi}\,,
\end{align*}
 Choosing $\mu=1$, we can therefore take $\alpha=\beta=a+b$ in \eqref{boundUXY}, so that
\begin{align}
    \|\mathcal{U}_{X,Y}-\Id\|_\diamond^{N,E}\le 2\sqrt{2}\,\sqrt{(a+b)(E+1)}\,.\label{boundXY}
\end{align}
Next, let us call $s$ the element in $\mathfrak{sp}_{2m}(\RR)$ such that $S=\exp(s)$, and let $\mathcal{U}_{X,Y}\equiv\mathcal{U}_S$. We introduce a basis $\{B_{a,b}\}_{a,b\in[2m]}$, of $\mathfrak{sp}_{2m}(\RR)$ \cite{pramana}: for any $i,j\in[m]$:
\begin{align*}
    B_{i,j}\coloneqq -E_{i+m,j}-E_{j+m,i}\,(i\le j),~~~B_{i+m,j+m}\coloneqq E_{i,j+m}+E_{j,i+m}\,(i\le j),~~~B_{i,j+m}\coloneqq -E_{i+m,j+m}+E_{j,i}\,,
\end{align*}
where $E_{a,b}\coloneqq |a\rangle\langle b|$. A simple counting argument shows that the number of such generators is equal to the dimension $m(2m+1)$ of $\mathfrak{sp}_{2m}(\RR)$. Next, normalizing the above matrices, we end up with the orthonormal basis: $\tilde{B}_{i,j}\coloneqq B_{i,j}/\sqrt{2}$ and $\tilde{B}_{i+m,j+m}\coloneqq B_{i+m,j+m}/\sqrt{2}$ for $i<j$, $\tilde{B}_{i,i}\coloneqq X_{i,i}/2$ and $\tilde{B}_{i+m,i+m}\coloneqq X_{i+m,i+m}/2$, and $\tilde{B}_{i,j+m}\coloneqq B_{i,j+m}/\sqrt{2}$ for all $i,j\in[m]$. Therefore, the element $s\in\mathfrak{sp}_{2m}(\RR)$ can be written as 
\begin{align*}
    s\coloneqq \sum_{i\le j }\,s_{i,j}\,\tilde{B}_{i,j}+s_{i+m,j+m}\tilde{B}_{i+m,j+m}+\sum_{i,j}s_{i,j+m}\,\tilde{B}_{i,j+m}\,,~~~\Rightarrow ~~~\|s\|^2_2\coloneqq \sum_{i\le j} s_{i,j}^2+s_{i+m,j+m}^2+\sum_{i,j}s_{i,j+m}^2\,,
\end{align*}
where the coefficients $s_{a,b}$ take real-valued. Now, the following expressions for the representations $\hat{B}_{a,b}$ of the basis elements $B_{a,b}$ in terms of the creation and annihilation operators can be found in \cite{pramana} (here we chose a slightly different normalisation, $B_{a,b}\equiv iX^{(0)}_{a,b}$ in the notations of \cite{pramana}):
\begin{align*}
    &\hat{B}_{i,j}=\frac{i}{2}\,\big(a_i^\dagger a_j+a_j^\dagger a_i+\delta_{ij}I+a_i^\dagger a_j^\dagger+a_ia_j\big)\\
    &\hat{B}_{i+m,j+m}=\frac{i}{2}\big( a_i^\dagger a_j+a_j^\dagger a_i+\delta_{ij}I-a_i^\dagger a_j^\dagger-a_ia_j\big)\\
    &\hat{B}_{i,j+m}=-\frac{1}{2}\big(a_j^\dagger a_i-a_i^\dagger a_j+a_i^\dagger a_j^\dagger-a_ia_j\big)\,.
\end{align*}
Thus, the element $s\in\operatorname{sp}_{2m}(\RR)$ is represented on $L^2(\RR^m)$ by 
\begin{align*}
    \hat{s}&=\frac{1}{\sqrt{2}}\,\sum_{i< j}\, s_{i,j}\,\hat{B}_{i,j}+s_{i+m,j+m}\,\hat{B}_{i+m,j+m}+\frac{1}{2}\,\sum_i\,s_{i,i}\,\hat{B}_{i,i}+s_{i+m,i+m}\hat{B}_{i+m,i+m}+\frac{1}{\sqrt{2}}\sum_{i,j}s_{i,j+m}\hat{B}_{i,j+m}\\
    &\simeq i\sum_{i,j}X_{i,j}\,a_i^\dagger a_j+Y_{i,j}\,a_ia_j+Y_{i,j}^{\dagger}a_i^\dagger a_j^\dagger \equiv iH\,,
\end{align*}
for some complex coefficients $X_{i,j}=X_{j,i}^{\dagger}\equiv $ and $Y_{i,j}$, where the symbol $\simeq$ in the last line means up to irrelevant constant terms. Comparing the two above expressions for $\hat{s}$, we find the correspondence:
\begin{align*}
    &X_{i,j}\coloneqq \frac{1}{2\sqrt{2}}\,(s_{i,j}+s_{i+m,j+m})\delta_{i<j}+\frac{1}{2\sqrt{2}}\,(s_{j,i}+s_{j+m,i+m})\delta_{j<i}+\frac{1}{2}\, (s_{i,i}+s_{i+m,i+m})\delta_{i,j}+\frac{i}{2\sqrt{2}}(s_{j,i+m}-s_{i,j+m})\,,\\
    &Y_{i,j}\coloneqq \frac{1}{2\sqrt{2}}(s_{i,j}-s_{i+m,j+m})\delta_{i<j}+\frac{1}{2\sqrt{2}}(s_{j,i}-s_{j+m,i+m})\delta_{j<i}+\frac{1}{4}(s_{i,i}-s_{i+m,i+m})\delta_{i,j}-\frac{i}{2\sqrt{2}}s_{i,j+m}\,.
\end{align*}
An easy calculation allows us to conclude that $\|X\|_2,\|Y\|_2\le \|s\|_2$. This together with the bound \eqref{boundXY} allows us to conclude that 
\begin{align*}
    \|\mathcal{U}_{X,Y}-I\|_\diamond^{N,E}\le 2\,\sqrt{\,({\sqrt{6}+\sqrt{10}+5\sqrt{2}m})\,\|s\|_2\,(E+1)} \,,
\end{align*}
The result follows since the unitary conjugation $\mathcal{U}_{X,Y}$ is by definition the unitary representation of the symplectic transformation $S$.
\end{proof}

 \begin{thm}\label{symplecticbound}
 Let $m\in\NN$ and $E\ge 0$. Then, for any $S,S'\in\operatorname{Sp}_{2m}(\RR)$, 
 \begin{align*}
     \|\mathcal{U}_S-\mathcal{U}_{S'}\|_\diamond^{N,E}\le      2\sqrt{\,({\sqrt{6}+\sqrt{10}+5\sqrt{2}m})\,(E+1)}\,
    \Big(\sqrt{\frac{\pi}{\|(S')^{-1}S\|_\infty+1}}+
    \sqrt{2\|(S')^{-1}S\|_\infty\,}\Big)\,\sqrt{\|(S')^{-1}S-I\|_2}\,.
 \end{align*}
 \end{thm}
 \begin{proof}
 First, by the unitary invariance of the trace distance, we have that
 \begin{align*}
     \|\mathcal{U}_S-\mathcal{U}_{S'}\|_\diamond^{N,E}=\|\mathcal{U}_{(S')^{-1}}\circ\mathcal{U}_S-\Id\|_\diamond^{N,E}=\|\mathcal{U}_{(S')^{-1}S}-\Id\|_\diamond^{N,E}\,,
 \end{align*}
 where the last identity comes from the group homomorphism property of $S\mathcal{U}_S$. Next, since $(S)'^{-1}S\in\operatorname{Sp}_{2m}(\RR)$, we have from Lemma \ref{lemmapolardecomp} the existence of $O,P\in\operatorname{exp}(\mathfrak{sp}_{2m}(\RR))$, $O$ being an orthogonal matrix and $P$ being a positive matrix, such that $S=PO$. Therefore
 \begin{align*}
     \|\mathcal{U}_S-\mathcal{U}_{S'}\|_\diamond^{N,E}&=\|\mathcal{U}_{PO}-\Id\|_\diamond^{N,E}\\
     &=\|\mathcal{U}_P\circ\mathcal{U}_O-\Id\|_\diamond^{N,E}\\
     &\le \|(\mathcal{U}_P-\Id)\circ\mathcal{U}_O\|_\diamond^{N,E}+\|\mathcal{U}_O-\Id\|_\diamond^{N,E}\\
     &=\|\mathcal{U}_{P}-\Id\|_\diamond^{N,E}+\|\mathcal{U}_O-\Id\|_\diamond^{N,E}\,,
     \end{align*}
     where the last line stands from the fact that $\mathcal{U}_O$ is a passive unitary transformation, so that for any initial finite energy state $\rho\in\cD(\cH)$, $\tr[N\mathcal{U}_O(\rho)]=\tr[N\rho]$. Using Lemma \ref{lemmapolardecomp} together with Corollary \ref{coroll}, we have that
 \begin{align*}
   & \|\mathcal{U}_{P}-\Id\|_\diamond^{N,E}\le 2\,\sqrt{\,({\sqrt{6}+\sqrt{10}+5\sqrt{2}m})\,\|\ln(P)\|_2\,(E+1)} \\
    &\|\mathcal{U}_{O}-\Id\|_\diamond^{N,E}\le 2\,\sqrt{\,({\sqrt{6}+\sqrt{10}+5\sqrt{2}m})\,\|\ln(O)\|_2\,(E+1)} \,.
\end{align*}
Now, since $P$ is a positive symplectic matrix, its eigenvalues come in pairs of positive numbers $(z_j,z_j^{-1})$, say for $z_j\ge 1$. Therefore we have by functional calculus that
 \begin{align*}
     \|\ln(P)\|_2^2=\sum_{j} \left(\ln(z_j)^2+\ln(z_j^{-1})^2\right) = 2\sum_j\ln(z_j)^2\le 2\sum_j(z_j-1)^2\le \sum_{j} \left( (z_j-1)^2+z_j^2(1-z_j^{-1})^2\right) \le \|P\|_\infty^2\|P-I\|_2^2
 \end{align*}
so that $\|\ln(P)\|_2\le \|P\|_\infty\|P-1\|_2$. As for $O$, it is a standard exercise to relate $\|\ln(O)\|_2$ to $\|O-I\|_2$ (see e.g.\ \cite[Exercise~B.5]{aubrun2017alice}):
\begin{align*}
    \|\ln(O)\|_2\le \frac{\pi}{2}\|O-I\|_2\,.
\end{align*}
Putting both bounds together, we have found that
\begin{align*}
    \|\mathcal{U}_S-\mathcal{U}_{S'}\|_\diamond^{N,E}&\le 2\sqrt{\,({\sqrt{6}+\sqrt{10}+5\sqrt{2}m})\,(E+1)}\,\Big(\sqrt{\frac{\pi}{2}\|O-I\|_2}+\sqrt{\|P\|_\infty\,\|P-I\|_2}\Big)\\
    &\le 2\sqrt{\,({\sqrt{6}+\sqrt{10}+5\sqrt{2}m})\,(E+1)}\,
    \Big(\sqrt{\frac{\pi}{\|(S')^{-1}S\|_\infty+1}}+
    \sqrt{2\|(S')^{-1}S\|_\infty\,}\Big)\,\sqrt{\|(S')^{-1}S-I\|_2}\,,
\end{align*}
where we used the continuity bounds for polar decompositions \cite[Corollary~VII.5.6 and Theorem~VII.5.1]{BHATIA-MATRIX}, together with $\|P\|_\infty=\|(S')^{-1}S\|_\infty$ in the last line above.
 \end{proof}

\section{A Solovay--Kitaev theorem for symplectic unitaries}

\subsection{The standard Solovay--Kitaev theorem}

Given a unitary operation $U$, determining how short a concatenation of base gates is required to approximate $U$
is a fundamental problem in quantum computation with practical relevance in the construction of efficient quantum processors. The celebrated Solovay--Kitaev theorem \cite{Kitaev1997, Kitaev1997b} provides an answer to this question by exhibiting an efficient algorithm for quantum compiling (see also the following non-exhaustive list \cite{NC, harrow2001quantum, harrow2002efficient, Dawson2006, aharonov2007polynomial, kuperberg2015hard, bouland2017trading} of more modern treatments, generalizations and refinements):

\begin{thm}[Solovay--Kitaev] \label{thm:Solovay-Kitaev-standard}
For any $U_1,..., U_n\in \operatorname{SU}(d)$ such that the group $\braket{U_1,...,U_n}$ they generate is dense in $\operatorname{SU}(d)$, there exists a constant $C$ and a procedure for approximating any $U\in \operatorname{SU}(d)$ to a precision $\eps>0$ with a string of $U_1,...,U_n$ and their inverses of length no greater than $C\log^c(1/\eps)$, where $c\sim 4$ and $C$ is independent of $U$ and $\eps$. This procedure can be implemented in a time polynomial in $\log(1/\eps)$.
\end{thm}

The Solovay--Kitaev Theorem has the following important corollary~\cite[Corollary~8]{harrow2001quantum}.

\begin{cor} \label{coroQC}
For any family of universal gates, there exists a constant $C$ such that any quantum circuit with $\ell$ arbitrary gates can be constructed from fewer than $C\ell \log^c(\ell)\log(1/\delta)$ universal gates with probability of error at most $\delta$.
\end{cor}

\subsection{A Gaussian Solovay--Kitaev theorem}

The proof of the Solovay--Kitaev theorem consists of an iterative procedure for the construction of $\eps$-nets over the set $\operatorname{SU}(d)$ around the identity. It relies on the approximation of $\operatorname{SU}(d)$ by its Lie algebra and for this reason is generalisable to any compact semi-simple Lie group. An extension to the case of a non-compact Lie group $G$ whose Lie algebra $\mathfrak{g}$ is perfect (that is, $\mathfrak{g}=[\mathfrak{g},\mathfrak{g}]$) was provided by \cite{kuperberg2015hard}. Fortunately, the symplectic group $\operatorname{Sp}_{2m}(\RR)$ belongs to this class. However, the distance used in order to measure the approximation in \cite{kuperberg2015hard} is a Riemannian left-invariant distance on the group $G$ (see also \cite{arnold1966geometrie, mishchenko1978euler, paradan2009symmetric}) whose physical interpretation (for instance in terms of the maximal amount squeezing allowed) is not obvious at first glance. Here, we propose to extend the Solovay--Kitaev theorem to the setting where one wants to approximate a quantum circuit made out of $m$-mode Gaussian unitary gates. Our main theorem is the following: we recall that the number operator is defined as $N\coloneqq \sum_{j=1}^ma_j^\dagger a_j$.

\begin{manualthm}{\ref{theosollovaykit}}
Let $m\in\NN$, $r>0$, $E>0$ and define $\operatorname{Sp}_{2m}^r(\RR)$ to be the set of all symplectic transformations $S$ such that $\|S-I\|_\infty\le r$. Then, there exists a constant $C\equiv C(r)<(2+r)(47r^2+104r+156)$ such that given $0<\eps_0< C(r)^{-2}$, any $\eps_0$-net $\cN_{\eps_0}$ of elements in $\operatorname{Sp}_{2m}^r(\RR)$ of size $|\cN_{\eps_0}|\le (3r/\eps_0)^{4m^2}$ is such that for any symplectic transformation $S\in \operatorname{Sp}_{2m}^r(\RR)$ and every $0<\delta$, there exists a finite sequence $S'$ of $\operatorname{poly}(\log\delta^{-1})$ elements from $\cN_{\eps_0}$ and their inverses, which can be found in time $\operatorname{poly}(\log\delta^{-1})$ such that 
 \begin{align*}
\|\mathcal{U}_S-\mathcal{U}_{S'}\|_{\diamond}^{N,E}\le \, 2\sqrt{\,\sqrt{2m}\left(\sqrt{6}+\sqrt{10}+5\sqrt{2}m\right) (E+1)} \left(\sqrt{\pi}+\sqrt{2}(r+1)\right) \sqrt{(r+1)\delta}\,.
 \end{align*}
\end{manualthm}

\begin{rem}
Notice the slight difference with the formulation in the main article, where $\widetilde{\operatorname{Sp}}_{2m}^r(\RR)$ was taken as the set of symplectic transformations $S$ such that $\|S\|_\infty\le r$. This is due to the fact that, while $\widetilde{\operatorname{Sp}}_{2m}^r(\RR)$ is physically better motivated since it corresponds to the maximum amount of squeezing generated by $S$, the proof relies on a refined estimation of the distance $\|S-I\|_\infty$ from the identity map.
Also, observe that the above version of Theorem~\ref{theosollovaykit} is slightly more general than that in the main text, as it is formulated by means of nets instead of generating sets. To see this, let us assume $\eps_0'\le \eps_0-\gamma$, $\gamma>0$, are fixed. Then a generating set can approximate any $\eps_0'$-net to accuracy $\gamma$ with constant depth (depending on $m$), leading to an $\eps_0$-net $\cM_{\eps_0}$ of size less than $(3r/\eps_0')^{4m^2}$. Since $\eps_0$ is fixed, the efficiency of this operation is neglected as compared to the precision $\delta$ in our analysis. The result of Theorem~\ref{theosollovaykit} follows directly from the proof of the above theorem adapted to the net $\cM_{\eps_0}$.
\end{rem}

Before delving into the proof of Theorem~\ref{theosollovaykit}, we want to dwell on the problem of establishing easily verifiable sufficient conditions in order for a set of symplectic matrices to generate a dense subgroup of $\symp_{2m}(\RR)$. We prove the following.

\begin{lemma} \label{dense_subgroups}
Let $\mathcal{K} = \left\{K_1,\ldots, K_r\right\} \subset \symp_{2m}(\RR)\cap \operatorname{SO}_{2m}(\RR)$ be a finite set of orthogonal symplectic matrices. Assume that the subgroup $\braket{\mathcal{K}}$ they generate is dense in $\symp_{2m}(\RR)\cap \operatorname{SO}_{2m}(\RR)$. Then, for any symplectic matrix $S$ that is not orthogonal, the subgroup $\mathcal{G}\coloneqq \braket{\mathcal{K}\cup \left\{S\right\}}$ is dense in the whole $\symp_{2m}(\RR)$.
\end{lemma}

\begin{rem}
It is well known~\cite{pramana} that orthogonal symplectic matrices of size $2m$ form a group that is isomorphic to the unitary group of size $m$, in formula $\symp_{2m}(\RR)\cap \operatorname{SO}_{2m}(\RR)\simeq \operatorname{U}_m(\CC)$. Therefore, the problem of determining whether $\braket{\mathcal{K}}$ is dense in $\symp_{2m}(\RR)\cap \operatorname{SO}_{2m}(\RR)$ is entirely equivalent to that of deciding whether a finite set of unitary matrices generates a dense subgroup of $\operatorname{U}_m(\CC)$. Curiously, this is the exact same situation encountered in the context of the standard Solovay--Kitaev Theorem~\ref{thm:Solovay-Kitaev-standard}.
\end{rem}

\begin{proof}[Proof of Lemma~\ref{dense_subgroups}]
Throughout the proof, we will make repeated use of the Euler decomposition theorem for the symplectic group~\cite{pramana}. It guarantees that any symplectic $T$ can be decomposed as $T=P\Lambda Q$, where $P,Q\in \symp_{2m}(\RR)\cap \operatorname{SO}_{2m}(\RR)$ and $\Lambda = \bigoplus_{j=1}^m \left(\begin{smallmatrix} \mu_j & 0 \\ 0 & \mu_j^{-1} \end{smallmatrix}\right)$, with $\mu_j\geq 1$ for all $j$. Note that $\left\{ \mu_j \right\}_j = \sv(T) \cap [1,\infty)$, where $\sv$ denotes the set of singular values, and that $\mu_j>1$ for some $j$ if and only if $T$ is not orthogonal.

Now, let $\overline{\mathcal{G}}$ denote the closure of $\mathcal{G}$. Observe that $\overline{\mathcal{G}}$ is closed under products and contains $\operatorname{SO}_{2m}(\RR)$. Thanks to this and to the Euler decomposition theorem, it suffices to show that for every set of numbers $\mu_1,\ldots, \mu_m\geq 1$ there exists $T\in\overline{\mathcal{G}}$ such that $\sv(T)\cap [1,\infty) = \{\mu_j\}_j$. To construct such $T$, start by setting $S'\coloneqq (SC)^{m-1} S$, where $C$ is the symplectic orthogonal matrix that permutes the modes cyclically, i.e.\ $C_{pq} = 1$ if $p=q+2\ (\mathrm{mod}\ 2m)$, and $C_{pq}=0$ otherwise. It is not difficult to verify that $\sv(S')\cap [1,\infty)$ is made of one element only, namely $\lambda \coloneqq \lambda_1\ldots \lambda_m$, with multiplicity $m$. Since $S$ is not orthogonal, there exists $j_0$ such that $\lambda_{j_0}>1$, and hence also $\lambda>1$. Note that $S'\in\overline{\mathcal{G}}$, and hence also $D_\lambda^{\oplus m} \in \overline{\mathcal{G}}$, where $D_\lambda \coloneqq \left( \begin{smallmatrix} \lambda & 0 \\ 0 & \lambda^{-1} \end{smallmatrix} \right)$ . Now, pick an integer $n$ such that $\lambda^{2n}\geq \max_j \mu_j$. Let the rotation matrix of angle $\theta$ be denoted by $R(\theta)\coloneqq \left( \begin{smallmatrix} \cos\theta & -\sin\theta \\ \sin\theta & \cos\theta \end{smallmatrix}\right)$. For angles $\theta_j\in [0,\pi/2]$ to be determined, set
\bbb
T = \left(D_\lambda^{\oplus m} \right)^n \left(\bigoplus\nolimits_j R(\theta_j) \right) \left(D_\lambda^{\oplus m} \right)^n = \bigoplus_j \begin{pmatrix} \lambda^{2n} \cos\theta_j & -\sin\theta_j \\ \sin\theta_j & \lambda^{-2n} \cos\theta_j \end{pmatrix} .
\eee
Note that $T\in \overline{\mathcal{G}}$. The singular values of the above $2\times 2$ blocks can be computed explicitly. For the $j^{\text{th}}$ block, they take the form $\{\eta(\lambda^{2n}, \theta_j),\, \eta(\lambda^{2n},\theta_j)^{-1}\}$, where
\begin{align*}
\eta(\kappa, \theta) &\coloneqq \sqrt{\zeta(\kappa,\theta)+\sqrt{\zeta(\kappa,\theta)^2 - 1}}\, , \\
\zeta(\kappa,\theta) &\coloneqq \frac{\kappa^4+1}{2\kappa^2} \cos^2(\theta) + \sin^2(\theta)\, .
\end{align*}
Clearly, $\eta(\lambda^{2n},\theta)$ is a continuous function of $\theta\in [0,\pi/2]$. Since $\eta(\lambda^{2n}, 0) = \lambda^{2n} \geq \mu_j$ and $\eta(\lambda^{2n},\pi/2) = 1\leq \mu_j$, we can find angles $\theta_j\in [0,\pi/2]$ satisfying $\eta(\lambda^{2n},\theta_j) = \mu_j$ for all $j$. Therefore, we constructed $T\in \overline{\mathcal{G}}$ such that $\sv(T) \cap [1,\infty) = \{ \mu_j\}_j$. Since the numbers $\mu_1,\ldots, \mu_m\geq 1$ were arbitrary, this implies that $\overline{\mathcal{G}}=\symp_{2m}(\RR)$, as claimed.
\end{proof}

\subsection{The proof}

The proof of Theorem \ref{theosollovaykit} is an adaptation of an argument by Aharonov et al.~\cite[Theorem~7.6]{aharonov2007polynomial}, which applies to the case of the group $\operatorname{Sp}_{2m}(\RR)$. We will leverage the estimates derived in Theorem  \ref{symplecticbound}. Thanks to these bounds, we can directly see that Theorem \ref{theosollovaykit} can be reduced to a result on approximations of symplectic matrices.
 
\begin{prop}\label{propfinitedim}
With the notations of Theorem \ref{theosollovaykit}, the sequence $S'$ satisfies 
\begin{align*}
     \|S-S'\|_\infty\le \delta\,.
\end{align*}
\end{prop}
 
By Before proving Proposition \ref{propfinitedim}, we show how it implies Theorem \ref{theosollovaykit}.
 
\begin{proof}[Reduction of Theorem \ref{theosollovaykit} to Proposition \ref{propfinitedim}]
By the bound found in Theorem \ref{symplecticbound}, we have that 
 \begin{align*}
     \|\mathcal{U}_S-\mathcal{U}_{S'}\|_\diamond^{N,E}&\le      2\sqrt{\,({\sqrt{6}+\sqrt{10}+5\sqrt{2}m})\,(E+1)}\,
    \Big(\sqrt{\frac{\pi}{\|(S')^{-1}S\|_\infty+1}}+
    \sqrt{2\|(S')^{-1}S\|_\infty\,}\Big)\,\sqrt{\|(S')^{-1}S-I\|_2}\\
    &\le 2\sqrt{\,\sqrt{2m}({\sqrt{6}+\sqrt{10}+5\sqrt{2}m})\,(E+1)}\,
    \Big(\sqrt{\pi}+
    \sqrt{2}(r+1)\Big)\,\sqrt{r+1}\sqrt{\|S-S'\|_\infty}\,,
 \end{align*}
 which is precisely the bound stated in Theorem \ref{theosollovaykit}.
\end{proof}

Hence, we have reduced the problem to that of proving Proposition \ref{propfinitedim}. Notice that from now on, the problem has become finite dimensional. As in the original proof of the Solovay--Kitaev theorem, our strategy reduces to finding approximations of elements $S \in \operatorname{Sp}^R_{2m}(\RR)$. First of all, we need a rough estimate on the cardinality of an $\eps$-net for this set.

\begin{lemma}\label{originalnet}
Let $r> 0$. Then, for any $0\le \eps\le r $, there exists an $\eps$-net $\cN_\eps$ for $(\operatorname{Sp}^r_{2m}(\RR),d)$, where we recall that $$\operatorname{Sp}^r_{2m}(\RR)\coloneqq \big\{S\in\operatorname{Sp}_{2m}(\RR),\,\|S-I\|_\infty\le r \big\}\,,$$ such that $|\cN_\eps|\le \Big(\frac{3\,r}{\eps} \Big)^{4m^2}$.
\end{lemma}
\begin{proof}
Let $\cN_\eps\subset \operatorname{Sp}^r_{2m}(\RR)$ be a maximal set such that $\|S-S'\|_\infty\ge \eps$ for all $S\ne S'\in \cN_\eps$ (such a set always exists by application of Zorn's lemma). Moreover, $\cN_\eps$ is an $\eps$-net: indeed assume that there exists $S\in \operatorname{Sp}^r_{2m}(\RR)$ such that $\|S-S'\|_\infty\ge \eps $ for all $S'\in \cN_\eps$. Then $\{S\}\bigcup \cN_\eps$ is a set that strictly contains $\cN_\eps$ and satisfies the assumption of an $\eps$-set. But this contradicts the maximality of $\cN_\eps$. Next, for any $S\in \operatorname{Sp}^r_{2m}(\RR)$ and $\delta>0$, we denote the closed ball around $S$ of radius $\delta$ as: 
\begin{align*}
  \overline{ \mathcal{B}}_{\delta}(S)\coloneqq \big\{ T\in \operatorname{Sp}^r_{2m}(\RR):\,\|T-S\|_\infty\le \delta\big\}\,.
\end{align*}
By definition of an $\eps$-net, the elements of $\{\overline{\mathcal{B}}_{\eps/2}(S)\}_{S\in\cN_\eps}$ are pairwise disjoint. Therefore,
\begin{align*}
    \sum_{S\in\cN_\eps}\,\mu\big( \overline{\mathcal{B}}_{\eps/2}(S)\big)=\mu\big(\bigcup_{S\in\cN_\eps}\overline{\mathcal{B}}_{\eps/2}(S)\big)\,,
\end{align*}
where $\mu$ denotes the Lebesgue measure on $\mathbb{M}_{2m}(\RR)\equiv \RR^{(2m)^2}$. Next, for any $T\in \bigcup_{S\in\cN_\eps}\overline{\mathcal{B}}_{\eps/2}(S)$, there exists $S_T\in\cN_\eps$ such that 
\begin{align*}
    \|T-I\|_\infty&\le \|T-S_T\|_\infty+\|S_T-I\|_\infty\le \frac{\eps}{2}+r\,\le \frac{3r}{2}\,,
\end{align*}
 so that $\bigcup_{S\in\cN_\eps}\overline{\mathcal{B}}_{\eps/2}(S)\subset \overline{\mathcal{B}}_{\frac{3r}{2}}(I)$. Therefore, by invariance of the Lebesgue measure under translations:
 \begin{align*}
 |\cN_\eps|\,.\, \mu\big( \overline{\mathcal{B}}_{\eps/2}(I)\big)=\sum_{S\in\cN_\eps}\mu\big(\overline{\mathcal{B}}_{\eps/2}(S)\big)\le \mu \big( \overline{\mathcal{B}}_{\frac{3r}{2}}(S)\big)\,.
\end{align*}
The result follows after using the well-known expression
\begin{align*}
    \mu\big(\overline{\mathcal{B}}_r(I)  \big)=\frac{\pi^{2m^2}\,r^{4m^2}}{(2m^2)!}
\end{align*}
for the volume of a hyperball on $\RR^{4m^2}$.
\end{proof}

As in the proof of the original result by Solovay and Kitaev, this basic first estimate turns out to be sub-optimal around the identity. As expected, a slight adaptation of the treatment of the finite dimensional unitary case (see e.g. \cite{Dawson2006}) carries through almost without any difficulty. In fact, the treatment of the special linear group carried out in \cite{aharonov2007polynomial} extends almost without any change to the present symplectic setting. However, we recall the argument in \cite{aharonov2007polynomial} in order to provide explicit estimates.

We begin by proving a bunch of technical lemmata. The first one is a direct extension of a result by Aharonov et al.~\cite[Lemma~B.1]{aharonov2007polynomial}:

\begin{lemma}\label{polardecomp}
Let $S=OP$ be the polar decomposition of $S\in\operatorname{Sp}_{2m}(\RR)$. Then for all $\eps>0$, 
\begin{align*}
    \|S-I\|_\infty\le \eps\quad \Longrightarrow\quad \|O-I\|_\infty, \|P-I\|_\infty \le 3\eps\,.
\end{align*}
\end{lemma}

\begin{proof}
This follows directly from classical results by Bhatia \cite[Theorem~VII.5.1 and Exercise~VII.5.3]{BHATIA-MATRIX} together with the fact that $\|S^{-1}\|_\infty=\|S\|_\infty$ for $S\in\symp_{2m}(\RR)$.
\end{proof}

We recall that given two elements $S,S'$ in $\symp_{2m}(\RR)$, their group commutator is defined as $\llbracket S,S'\rrbracket\coloneqq SS'S^{-1}S'^{-1}$. The next lemma is adapted from a result by Aharonov et al.~\cite[Section~B.5]{aharonov2007polynomial}.

\begin{lemma}\label{orthogonalpart}
Let $\eps\in [0,1]$, and let $O\in\operatorname{SO}_{2m}(\RR)\bigcap \symp_{2m}(\RR)$ be such that $\|O-I\|_\infty\le \eps$. Then there exist two matrices $O^{(1)},O^{(2)}\in\symp_{2m}(\RR)$ such that $\|O^{(1)}-I\|_\infty,\|O^{(2)}-I\|_\infty\le \frac{3}{2} \sqrt{\eps}$ and $\left\|O-\llbracket O^{(1)},O^{(2)}\rrbracket\right\|_\infty\le \frac{19}{10}\,\eps^{3/2}$.
\end{lemma}

\begin{proof}
Since $O\in\operatorname{SO}_{2m}(\RR)$, there exists an orthogonal transformation $K\in \operatorname{Sp}_{2m}(\RR)\bigcap \operatorname{SO}_{2m}(\RR)$ and parameters $\theta\coloneqq (\theta_1,\ldots ,\theta_m)\in [-\pi,\pi]^m$ such that $O=K D(\theta)\,K^\intercal$ \cite[Appendix~B and Section~5.1.2.1]{BUCCO}, where
\begin{align*}
    D(\theta) = \bigoplus_{j=1}^m \begin{pmatrix} \cos\theta_j & \sin\theta_j \\ -\sin\theta_j & \cos\theta_j \end{pmatrix} \eqqcolon \bigoplus_{j=1}^m D_j\, .
\end{align*}
Therefore, we can reduce the problem to that of approximating $D(\theta)$ by unitary invariance and the fact that $K$ is symplectic. In each block $j$, the matrix $D_j$ can be diagonalized as 
\begin{align*}
D_j = V \begin{pmatrix} e^{i\theta_j} & 0 \\ 0 & e^{-i\theta_j} \end{pmatrix} V^\dag = V e^{iH_j}V^\dag\, ,\qquad H_j \coloneqq \begin{pmatrix} \theta_j & 0 \\ 0 & -\theta_j \end{pmatrix}, \qquad  V\coloneqq \frac{1}{\sqrt{2}} \begin{pmatrix} 1 & i \\ i & 1 \end{pmatrix}.
\end{align*}
Now, since for each $j$ it holds that $2 \left|\sin\left(\theta_j/2\right)\right| = \|D_j-I\|_\infty \leq \|D(\theta)-I\|_\infty\le \eps$, we have that $\|H_j\|_\infty=|\theta_j|\le 2 \arcsin(1/2)\, \eps$. Next, define the Hermitian matrices 
\begin{align*}
F_j\coloneqq i\sqrt{\frac{\theta_j}{2}} \begin{pmatrix} 0 & 1 \\ -1 & 0 \end{pmatrix}, \qquad G_j\coloneqq \sqrt{\frac{\theta_j}{2}} \begin{pmatrix} 0 & 1 \\ 1 & 0 \end{pmatrix} .
\end{align*}

One can easily check that $[F_j, G_j]=iH_j$ and that $\|F_j\|_\infty,\|G_j\|_\infty =\sqrt{\frac{\theta_j}{2}}\le \sqrt{\arcsin(1/2)\, \eps}\, \eqqcolon c_1\,\sqrt{\eps}$. Next, define $\tilde{O}_j\coloneqq e^{F_j}$, $\tilde{O}_j'\coloneqq e^{G_j}$, we also have that 
\bb
\|\tilde{O}_j-I\|_\infty,\|\tilde{O}_j'-I\|_\infty \leq e^{c_1\sqrt\eps} - 1 \leq c_1\sqrt\eps\, e^{c_1\sqrt\eps} \leq c_1 e^{c_1} \sqrt\eps \leq \frac{3}{2} \sqrt{\eps}\, .
\label{distance_Ojtilde}
\ee

Now, denoting the group commutator by $\llbracket S_1,S_2\rrbracket\coloneqq S_1S_2S_1^{-1}S_2^{-1}\in \operatorname{Sp}_{2m}(\RR)$, we have that for any two matrices $A,B$ such that $\|A\|_\infty, \|B\|_\infty\le \delta$, $\left\|e^{[A,B]}-\llbracket e^A,e^B\rrbracket\right\|_\infty\le c_2\delta^3$ for a constant $c_2\le 5$. Therefore, 
\begin{equation}
    \left\|e^{iH_j}-\llbracket \tilde{O}_j,\tilde{O}_j'\rrbracket\right\|_\infty \leq c_1^3 c_2\,\eps^{3/2} \leq \frac{19}{10}\, \eps^{3/2}\,.
    \label{distance_commut}
\end{equation}
We now construct the two matrices
\begin{equation*}
O^{(1)}\coloneqq K \left(\bigoplus_{j=1}^m V^\dag\,\tilde{O}_j \,V\right) K^\intercal\, , \qquad O^{(2)}\coloneqq K \left( \bigoplus_{j=1}^m V^\dag\,\tilde{O}_j'\,V\right) K^\intercal\, .
\end{equation*}
First, observe that thanks to \eqref{distance_Ojtilde} we have that $\left\| O^{(1)} - I \right\|_\infty, \left\| O^{(2)} - I \right\|_\infty \leq \frac{3}{2}\sqrt\eps$. Also, since
\bbb
V^\dag\,\tilde{O}_j \,V = \begin{pmatrix} e^{-\alpha_j} & 0 \\ 0 & e^{\alpha_j} \end{pmatrix} , \qquad V^\dag\,\tilde{O}'_j \,V = \begin{pmatrix} \cosh(\alpha_j) & \sinh(\alpha_j) \\ \sinh(\alpha_j) & \cosh(\alpha_j) \end{pmatrix} ,
\eee
where $\alpha_j \coloneqq \sqrt{\frac{\theta_j}{2}}$, and these clearly belong to $\symp_2(\RR)$, both $O^{(1)}$ and $O^{(2)}$ belong to $\symp_{2m}(\RR)$. Finally,
\begin{align*}
\left\|O-\llbracket O^{(1)},O^{(2)}\rrbracket\right\|_\infty &= \left\|D(\theta) - \llbracket K^\intercal O^{(1)}K,\, K^\intercal O^{(2)} K\rrbracket\right\|_\infty \\
&= \max_{j=1,\ldots, m} \left\|e^{iH_j}-\llbracket \tilde{O}_j,\tilde{O}_j'\rrbracket\right\|_\infty \\
&\le \frac{19}{10}\,\eps^{3/2}\,,
\end{align*}
where the last estimate follows from \eqref{distance_commut}. This completes the proof.
\end{proof}

We proceed similarly on the positive part of $\operatorname{Sp}_{2m}(\RR)$:
\begin{lemma}\label{symmetric}
Let $\eps\in [0,1]$ and $P\in\Pi(m)\coloneqq \{S\in\symp_{2m}(\RR):\,S^\intercal = S,S>0\}$ be such that $\|P-I\|_\infty\leq \eps$. Then there exist two matrices ${P}^{(1)},{P}^{(2)}\in \symp_{2m}(\RR)$ such that $\left\|P^{(1)}-I\right\|_\infty,\left\|P^{(2)}-I\right\|_\infty\le 1.44\sqrt{\eps}$ and $\left\|\llbracket P^{(1)},P^{(2)}\rrbracket \,-P\right\|_\infty\le \frac{9}{5}\,\eps^{3/2}$.
\end{lemma}

\begin{proof}
Thanks to unitary invariance and to the existence of a Williamson decomposition of $P$ \cite[Proposition~2.13]{GOSSON}, we can assume without loss of generality we that $P$ is a diagonal operator of the form $P=\operatorname{diag}(\lambda_1,1/\lambda_1,\ldots,\lambda_m,1/\lambda_m)$ for some parameters $\lambda_1,\ldots, \lambda_m\geq 1$, with $\lambda_j - 1 \leq \eps$ for all $j$.

Now, consider a block of the form $\operatorname{diag}(\lambda_j,1/\lambda_j)$. It can be written as $e^{H_j}$, where $H_j \coloneqq \operatorname{diag}(\theta_j,-\theta_j)$ is such that $e^{\theta_j}=\lambda_j$. We immediately deduce that $0\leq \theta_j\leq \ln(1+\eps)\leq \eps$. Next, we define the matrices
\begin{align*}
F_j\coloneqq \sqrt{\frac{\theta_j}{2}}\begin{pmatrix} 0 & 1 \\ -1 & 0 \end{pmatrix}, \qquad G_j\coloneqq \sqrt{\frac{\theta_j}{2}} \begin{pmatrix} 0 & 1 \\ 1 & 0 \end{pmatrix} ,
\end{align*}
so that $H_j=[F_j,G_j]$. Note that $\left\|F_j \right\|_\infty,\left\|G_j \right\|_\infty = \sqrt{\frac{\theta_j}{2}} \leq \sqrt{\frac{\eps}{2}}$.
Exactly as in the proof of Lemma \ref{orthogonalpart}, one can verify that $P_j\coloneqq e^{F_j}$ and $P_j'\coloneqq e^{G_j}$ are in $\symp_{2}(\RR)$, and that moreover 
\bbb
\|P_j-I\|_\infty,\|P_j'-I\|_\infty \leq e^{\sqrt{\theta_j/2}} - 1 \leq e^{\sqrt{\eps/2}} - 1 \leq 2^{-1/2} e^{2^{-1/2}}  \sqrt\eps \leq 1.44 \sqrt{\eps} ,
\eee
where for the second to last inequality we leveraged the elementary fact that $e^y - 1\leq y e^y$. As above, observing for any two matrices $A,B$ such that $\|A\|_\infty, \|B\|_\infty\le \delta$, it holds that $\left\|e^{[A,B]}-\llbracket e^A,e^B\rrbracket\right\|_\infty\le 5\delta^3$, we finally obtain that
\begin{align*}
    \left\|e^{H_j}-\llbracket P_j,P_j'\rrbracket\right\|_\infty\leq 5\cdot2^{-3/2}\,\eps^{3/2} \leq \frac{9}{5}\, \eps^{3/2}\,.
\end{align*}
The result follows after taking the direct sum all the blocks, indexed by $j=1,\ldots, m$.
\end{proof}

The last Lemma is a quantitative version of Lemma B.2 in \cite{aharonov2007polynomial}:

\begin{lemma}\label{lemmaproductapprox}
Fix $\mu\in [0,1)$. Let $\eps,\delta$ satisfy $\delta+\eps\leq \mu$, and let $V,W,\tilde{V},\tilde{W}$ be four matrices such that $\|V-\tilde{V}\|_\infty, \|W-\tilde{W}\|_\infty\le \eps$ and $\|V-I\|_\infty,\|W-I\|_\infty\le \delta$. Then, 
\bb
\left\|\llbracket V,W\rrbracket -\llbracket \tilde{V},\tilde{W}\rrbracket\right\|_\infty \leq \frac{16 - 12\mu + 4 \mu ^2}{(1-\mu )^3}\, \delta\eps + \frac{7 -9 \mu +13 \mu ^2 -3 \mu ^3}{(1-\mu )^4}\, \eps^2 \,.
\label{commutator_bound}
\ee
For example, for $\mu=1/5$ we obtain that
\bb
\left\|\llbracket V,W\rrbracket -\llbracket \tilde{V},\tilde{W}\rrbracket\right\|_\infty \leq 27 \delta\eps + 14 \eps^2 \,.
\label{commutator_bound_1/5}
\ee
\end{lemma}

\begin{proof}
We first start by denoting $\eps A \coloneqq \tilde{V}-V$ and $\eps B \coloneqq \tilde{W}-W$, so that $\|A\|_\infty,\|B\|_\infty\le 1$ by assumption. We start by recording an elementary observation: if $X$ is any matrix such that $\left\|X-I\right\|_\infty \leq \mu<1$, then $X$ is invertible, and moreover
\begin{align}
\left\|X^{-1}-I\right\|_\infty &\leq \frac{\left\|X-I\right\|_\infty}{1-\mu}\, , \label{estimate_distance_inverse} \\
\left\|X^{-1}\right\|_\infty &\leq \frac{1}{1-\mu}\, . \label{estimate_norm_inverse}
\end{align}
To see why this is the case, first note that the eigenvalues of $X$ are at a distance at most $\mu$ from $1$, hence none of them can vanish. Then,
\bbb
\left\|X^{-1} -I \right\|_\infty = \left\| \left(I - (I-X)\right)^{-1} - I \right\|_\infty = \left\|\sum_{n=1}^\infty (I-X)^n\right\|_\infty \leq \sum_{n=1}^\infty \left\|I-X\right\|_\infty^n \leq \frac{\left\|X-I\right\|_\infty}{1-\mu}\, .
\eee
In our case, this implies that all the operators $V,W,\tilde{V},\tilde{W}$ are invertible. Moreover,
\begin{align}
    \left\|V^{-1}\right\|_\infty,\, \left\|W^{-1}\right\|_\infty,\, \left\|\tilde{V}^{-1}\right\|_\infty,\, \left\|\tilde{W}^{-1}\right\|_\infty &\leq \frac{1}{1-\mu}\, , \label{estimate_1} \\
    \left\|V^{-1} - I\right\|_\infty,\, \left\|W^{-1} - I\right\|_\infty &\leq \frac{\delta}{1-\mu}\, , \label{estimate_2} \\
    \left\|\tilde{V}^{-1} - I\right\|_\infty,\, \left\|\tilde{W}^{-1} - I\right\|_\infty &\leq \frac{\delta+\eps}{1-\mu}\, . \label{estimate_3}
\end{align}
Now, consider that
\begin{align*}
    \tilde{V}^{-1} - V^{-1} + \eps V^{-1}AV^{-1} &= (V+\eps A)^{-1} - V^{-1} + \eps V^{-1}AV^{-1} \\
    &= (V+\eps A)^{-1} \left( V - (V+\eps A) \right) V^{-1} + \eps V^{-1} A V^{-1} \\
    &= - \eps (V+\eps A)^{-1} A V^{-1} + \eps V^{-1} A V^{-1} \\
    &= \eps \left( V^{-1} - (V+\eps A)^{-1} \right) A V^{-1} \\
    &= \eps V^{-1} \left((V+\eps A) - V\right) (V+\eps A)^{-1} A V^{-1} \\
    &= \eps^2 V^{-1} A (V+\eps A)^{-1} A V^{-1} \\
    &= \eps^2 V^{-1} A \tilde{V}^{-1} A V^{-1}\, .
\end{align*}
Then,
\bb
\begin{aligned}
&\llbracket V,W\rrbracket -\llbracket \tilde{V},\tilde{W}\rrbracket \\
&\qquad = \tilde{V}\tilde{W}\tilde{V}^{-1}\tilde{W}^{-1}-VWV^{-1}W^{-1} \\
&\qquad = (\tilde{V}-V)\tilde{W}\tilde{V}^{-1}\tilde{W}^{-1}+V(\tilde{W}-W)\tilde{V}^{-1}\tilde{W}^{-1} + VW\left(\tilde{V}^{-1}-V^{-1}\right)\tilde{W}^{-1}+VWV^{-1}\left(\tilde{W}^{-1}-W^{-1}\right) \\
&\qquad = \eps \left(A \tilde{W} \tilde{V}^{-1} \tilde{W}^{-1} + VB\tilde{V}^{-1}\tilde{W}^{-1} - VWV^{-1}A V^{-1} \tilde{W}^{-1} - VWV^{-1}V^{-1} W^{-1} B W^{-1} \right) \\
&\qquad \quad + \eps^2 \left( VWV^{-1}A \tilde{V}^{-1} A V^{-1} \tilde{W}^{-1} + VWV^{-1} W^{-1} B \tilde{W}^{-1} B W^{-1} \right) \\
&\qquad = \eps Z_1 + \eps^2 Z_2\, ,
\end{aligned}
\label{expansion}
\ee
where
\begin{align}
Z_1 &\coloneqq A \left( \tilde{W} \tilde{V}^{-1} \tilde{W}^{-1} - I\right) - \left( VWV^{-1}A V^{-1} \tilde{W}^{-1} - A \right) + \left( VB\tilde{V}^{-1}\tilde{W}^{-1} - B\right) - \left( VWV^{-1}V^{-1} W^{-1} B W^{-1} - B \right) , \label{Z1} \\
Z_2 &\coloneqq VWV^{-1}A \tilde{V}^{-1} A V^{-1} \tilde{W}^{-1} + VWV^{-1} W^{-1} B \tilde{W}^{-1} B W^{-1} \label{Z2}
\end{align}
We can now proceed to estimate the operator norm of the various terms in \eqref{Z1} and \eqref{Z2}. To this end, we make systematic use of the telescopic bound
\bb
\left\| \prod_{j=1}^r X_j - I\right\|_\infty = \left\| \sum_{k=1}^r \left( \prod_{j=k+1}^r X_j - \prod_{j=k}^r X_j \right) \right\|_\infty = \left\| \sum_{k=1}^r (I-X_k) \prod_{j=k+1}^r X_j \right\|_\infty \leq \sum_{k=1}^r \left\|X_k - I\right\|_\infty \prod_{j=k+1}^r \left\|X_j\right\|_\infty\, .
\label{telescopic}
\ee
Then, using \eqref{estimate_1}--\eqref{estimate_3} in conjunction with \eqref{telescopic},
\bb
\left\| A \left( \tilde{W} \tilde{V}^{-1} \tilde{W}^{-1} - I\right) \right\|_\infty \leq \left\| \tilde{W} \tilde{V}^{-1} \tilde{W}^{-1} - I\right\|_\infty \leq \frac{\delta}{(1-\mu)^2} + \frac{\delta+\eps}{(1-\mu)^2} + \frac{\delta+\eps}{1-\mu} = \frac{3-\mu}{(1-\mu)^2}\, \delta + \frac{2-\mu}{(1-\mu)^2}\, \eps\, .
\label{estimate_4}
\ee
Also,
\bb
\begin{aligned}
\left\| VWV^{-1}A V^{-1} \tilde{W}^{-1} - A \right\|_\infty &\leq \left\| VWV^{-1} - I \right\|_\infty \left\|V^{-1}\right\|_\infty \left\|\tilde{W}^{-1}\right\|_\infty + \left\|V^{-1}\tilde{W}^{-1} - I\right\|_\infty \\
&\leq \left( \frac{1+\mu}{1-\mu}\, \delta + \frac{\delta}{1-\mu} + \frac{\delta}{1-\mu} \right)\frac{1}{(1-\mu)^2} + \frac{\delta}{(1-\mu)^2} + \frac{\delta+\eps}{1-\mu} \\
&= \frac{5-2\mu + \mu^2}{(1-\mu)^3}\, \delta + \frac{\eps}{1-\mu}\, .
\end{aligned}
\label{estimate_5}
\ee
Continuing, we find that
\bb
\begin{aligned}
\left\|VB\tilde{V}^{-1}\tilde{W}^{-1} - B\right\|_\infty &\leq \left\|V-I\right\|_\infty \left\|\tilde{V}^{-1}\right\|_\infty \left\|\tilde{W}^{-1}\right\|_\infty + \left\|\tilde{V}^{-1}\tilde{W}^{-1} - I\right\|_\infty \\
&\leq \frac{\delta}{(1-\mu)^2} + \frac{\delta+\eps}{(1-\mu)^2} + \frac{\delta+\eps}{1-\mu} \\
&= \frac{3-\mu}{(1-\mu)^2}\, \delta + \frac{2-\mu}{(1-\mu)^2}\, \eps\, .
\end{aligned}
\label{estimate_6}
\ee
The norm of the last term in \eqref{Z1} can be bounded as
\bb
\begin{aligned}
\left\|VWV^{-1}V^{-1} W^{-1} B W^{-1} - B\right\|_\infty &\leq \left\|W^{-1}\right\|_\infty \left\|V^{-1}\right\|_\infty \left\|W\right\|_\infty \left\|V\right\|_\infty \left\|W^{-1}-I\right\|_\infty + \left\|VWV^{-1}W^{-1} - I\right\|_\infty \\
&\leq \frac{(1+\mu)^2}{(1-\mu)^2}\, \frac{\delta}{1-\mu} + \frac{\delta(1+\mu)}{(1-\mu)^2} + \frac{\delta}{(1-\mu)^2} + \frac{\delta}{(1-\mu)^2} + \frac{\delta}{1-\mu} \\
&= \frac{5-2\mu + \mu^2}{(1-\mu)^3}\, \delta\, .
\end{aligned}
\label{estimate_7}
\ee
Putting together \eqref{estimate_4}--\eqref{estimate_7}, we obtain that
\bb
\left\|Z_1\right\|_\infty \leq \frac{16-12\mu+4\mu^2}{(1-\mu)^3}\, \delta + \frac{5-3\mu}{(1-\mu)^2}\, \eps\, .
\label{estimate_Z1}
\ee

We now proceed to upper bound the operator norm of the matrix $Z_2$ defined by \eqref{Z2}. This is simply done by combining \eqref{estimate_1} with the triangle inequality. One obtains that
\bb
\begin{aligned}
\left\|Z_2\right\|_\infty &\leq \left\|VWV^{-1}A \tilde{V}^{-1} A V^{-1} \tilde{W}^{-1}\right\|_\infty + \left\|VWV^{-1} W^{-1} B \tilde{W}^{-1} B W^{-1}\right\|_\infty \\
&\leq \left\|V\right\|_\infty \left\| W\right\|_\infty \left\|V^{-1}\right\|_\infty \left\|\tilde{V}^{-1}\right\|_\infty \left\| V^{-1}\right\|_\infty \left\| \tilde{W}^{-1}\right\|_\infty + \left\|V\right\|_\infty \left\|W\right\|_\infty \left\|V^{-1}\right\|_\infty \left\| W^{-1}\right\|_\infty \left\|\tilde{W}^{-1} \right\|_\infty \left\|W^{-1}\right\|_\infty \\
&\leq \frac{2(1+\mu)^2}{(1-\mu)^4}\, .
\end{aligned}
\label{estimate_Z2}
\ee
Plugging \eqref{estimate_Z1} and \eqref{estimate_Z2} into \eqref{expansion} yields the claim \eqref{commutator_bound}.
\end{proof}

We are now ready to prove Proposition \ref{propfinitedim}.

\begin{proof}[Proof of Proposition \ref{propfinitedim}]
The proof follows from a direct extension of the standard Solovay--Kitaev algorithm. The latter is based on a recursive routine SK($S,n$) indexed by an integer $n\in\NN$ that receives the transformation $S\in \operatorname{Sp}_{2m}^r(\RR)$ and returns a product $S_n$ of symplectic matrices such that $\|S_n-S\|_\infty\le \eps_n$ given some approximation constant $\eps_n$ which we will determine. We first briefly describe the routine before analysing its efficiency. For $n=0$, given some $\eps_0>0$ to be determined later, construct an $\eps_0$-net $\cN_{\eps_0}$ as in Lemma \ref{originalnet} and return $S_0$ to be the element in $\cN_{\eps_0}$ such that $\|S-S_0\|_\infty\le \eps_0$. Next for $n\ge 0$ assume given a sequence $S_n\coloneqq \operatorname{SK(}S,n\operatorname{)}$ of symplectic matrices in $\operatorname{Sp}_{2m}(\RR)$ such that $\|S_n-S\|_\infty\le \eps_n$. In what follows, we describe how to get the $(n+1)^{\text{th}}$ round of approximation $S_{n+1}\coloneqq $SK($S,n+1$):
first of all set $\Delta_n\coloneqq SS_n^{-1}$. Next, let $\Delta_n\coloneqq O_nP_n$, with $O_n\in\operatorname{Sp}_{2m}(\RR)\bigcup \operatorname{SO}_{2m}(\RR)$ and $P_n\in \Pi(m)$, be the polar decomposition of $\Delta_n$. By Lemma \ref{polardecomp} the matrices $O_n$ and $P_n$ satisfy $\|O_n-I\|_\infty,\|P_n-I\|_\infty\le 3\eps_n$. Moreover, by Lemma \ref{orthogonalpart}, there exist symplectic operators $O_n^{(1)},O_n^{(2)}$ such that $\|O^{(j)}_n-I\|_\infty\le {\frac{3}{2}\sqrt{\eps_n}}$ and $\|O_n-\llbracket O_n^{(1)},O_n^{(2)}\rrbracket \|_\infty\le {\frac{19}{10}\,\eps_n^{3/2}}$. Similarly, by Lemma \ref{symmetric} there exist symplectic matrices $P_n^{(1)},P_n^{(2)}$ such that $\|P^{(j)}_n-I\|_\infty\le {1.44\sqrt{\eps_n}}$ and $\|P_n-\llbracket P_n^{(1)},P_n^{(2)}\rrbracket \|_\infty\le\, {\frac{9}{5}\eps_n^{3/2}}$. Calling the routines 
SK($P^{(j)}_n,n$) and SK($O^{(j)}_n,n$), we find $\eps_n$ approximations $\tilde{P}^{(j)}_n$, resp. $\tilde{O}_j^{(n)}$, of $P^{(j)}_n$, resp. of $O^{(j)}_n$. Define $\tilde{\Delta}_n\coloneqq \llbracket \tilde{O}^{(1)}_n,\tilde{O}^{(2)}_n\rrbracket\,.\llbracket \tilde{P}^{(1)}_n,\tilde{P}^{(2)}_n\rrbracket$, and $S_{n+1}\coloneqq  \tilde{\Delta}_n S_n$. Let us now analyse the efficiency of the method. First of all, we verify that $S_{n+1}$ is indeed closer from $S$ than $S_n$. This can be verified thanks to Lemma \ref{lemmaproductapprox}: 
\bb
\begin{aligned}
\|\Delta_n-\tilde{\Delta}_n\|_\infty&\le \|O_nP_n-\llbracket {O}^{(1)}_n,{O}^{(2)}_n\rrbracket\,.\llbracket {P}^{(1)}_n,{P}^{(2)}_n\rrbracket\|_\infty+\|\llbracket {O}^{(1)}_n,{O}^{(2)}_n\rrbracket\,.\llbracket {P}^{(1)}_n,{P}^{(2)}_n\rrbracket-\llbracket \tilde{O}^{(1)}_n,\tilde{O}^{(2)}_n\rrbracket\,.\llbracket \tilde{P}^{(1)}_n,\tilde{P}^{(2)}_n\rrbracket\|_\infty \\
&\le\|O_n-\llbracket {O}^{(1)}_n,{O}^{(2)}_n\rrbracket\|_\infty\,\|P_n\|_\infty+\|\llbracket {O}^{(1)}_n,{O}^{(2)}_n\rrbracket\|_\infty \,\|P_n-\llbracket {P}^{(1)}_n,{P}^{(2)}_n\rrbracket\|_\infty \\
&+ \|\llbracket {O}^{(1)}_n,{O}^{(2)}_n\rrbracket\,.\llbracket {P}^{(1)}_n,{P}^{(2)}_n\rrbracket-\llbracket \tilde{O}^{(1)}_n,\tilde{O}^{(2)}_n\rrbracket\,.\llbracket \tilde{P}^{(1)}_n,\tilde{P}^{(2)}_n\rrbracket\|_\infty \\
&\overset{1}{\le} (1+r)\,(1+r+\eps_n)\,{\frac{19}{10}\,\eps_n^{\frac{3}{2}}}+(1+{\frac{3}{2}\,\sqrt{\eps_n}+\frac{19}{10}\,\eps_n^{\frac{3}{2}}}) \,{\frac{9}{5}\,\eps_n^{\frac{3}{2}}} \\
&+ \|\llbracket {O}^{(1)}_n,{O}^{(2)}_n\rrbracket\,-\llbracket \tilde{O}^{(1)}_n,\tilde{O}^{(2)}_n\rrbracket\|_\infty\,\|\llbracket {P}^{(1)}_n,{P}^{(2)}_n\rrbracket\|_\infty+\| \llbracket \tilde{O}^{(1)}_n,\tilde{O}^{(2)}_n\rrbracket \|_\infty\,\|\llbracket \tilde{P}^{(1)}_n,\tilde{P}^{(2)}_n\rrbracket-\llbracket {P}^{(1)}_n,{P}^{(2)}_n\rrbracket\|_\infty \\
&\overset{2}{\le}(1+r)\,(1+r+\eps_n)\,{\frac{19}{10}\,\eps_n^{\frac{3}{2}}}+(1+{\frac{3}{2}\,\sqrt{\eps_n}+\frac{19}{10}\,\eps_n^{\frac{3}{2}}}) \,{\frac{9}{5}\,\eps_n^{\frac{3}{2}}} \\
&+{45\eps_n^{\frac{3}{2}}}\,\big(\| \llbracket \tilde{O}^{(1)}_n,\tilde{O}^{(2)}_n\rrbracket \|_\infty+\| \llbracket {P}^{(1)}_n,{P}^{(2)}_n\rrbracket \|_\infty  \big) \\
&\overset{3}{\le} (1+r)\,(1+r+\eps_n)\,{\frac{19}{10}\,\eps_n^{\frac{3}{2}}}+(1+{\frac{3}{2}\,\sqrt{\eps_n}+\frac{19}{10}\,\eps_n^{\frac{3}{2}}}) \,{\frac{9}{5}\,\eps_n^{\frac{3}{2}}} \\
&+{45\eps_n^{\frac{3}{2}}\,\big(1+(1+r)(1+r+\eps_n)+\frac{3}{2}\sqrt{\eps_n}+94\eps_n^{\frac{3}{2}} \big)} \\
& \overset{4}{\le}{\eps_n^{\frac{3}{2}}\, (47r^2+104r+156)}\eqqcolon \eps_n^{\frac{3}{2}}\,C(r)\,.
\end{aligned}
\label{estimate}
\ee
In~1 we used the estimates recalled above, as well as the fact that 
\begin{align*}
\|P_n\|_\infty=\|\Delta_n\|_\infty=\|SS_n^{-1}\|_\infty&\le (1+\|S-I\|_\infty)\|S_n^{-1}\|_\infty\\
& \overset{(a)}{=}(1+\|S-I\|_\infty)\|S_n\|_\infty\\
&\le (1+\|S-I\|_\infty)(1+\|S-I\|_\infty+\|S_n-S\|_\infty)\\
&\le (1+r)\,(1+r+\eps_n)\,,
\end{align*}
where $(a)$ follows from the fact that $S_n$ is symplectic. In~2 and~3 we used Lemma \ref{lemmaproductapprox} with {$\delta\equiv \frac{3}{2}\sqrt{\eps_n}$ and $\eps \equiv \eps_n$}, which we assume to satisfy the condition {$\frac{3}{2}\sqrt{\eps_n}+\eps_n\le 1/5$}. 
We now prove that $\eps_n$ can be chosen to be of the form $c^{-1}(c\eps_0)^{(3/2)^n}$, for some constants $c\equiv c(m,r)$, and that $S_n$ is a product of $11^n$ generators. The statement holds trivially for $n=0$. Next, assume that it is true for $n$ and consider the $(n+1)^{\text{th}}$ case. Using the estimate \eqref{estimate}, we have that 
\begin{align*}
    \|S_{n+1}-S\|_\infty =\|S_n(\Delta_n-\tilde{\Delta}_n)\|_\infty&\le (\|S_n-S\|_\infty+\|S\|_\infty)\,\|\Delta_n-\tilde{\Delta}_n\|_\infty\\
    &\le (2 +r)\,C(r)\,\eps_n^{\frac{3}{2}}\\
    &\overset{(a)}{\le}(2 +r)\,C(r)\,\Big[ c^{-1}(c\eps_0)^{(3/2)^n}\Big]^{\frac{3}{2}}\\
    &\equiv (2 +r)\,C(r)\, c^{-1/2}\eps_{n+1}\,,
\end{align*}
where (a) above simply follows from the induction hypothesis. Therefore, choosing $c=\big[(2 +r)\,C(r)\big]^2$ and imposing that $\eps_0\,c<1$ gives the convergence result. {Remark that in this case, the condition that $\eps_n+\frac{3}{2}\sqrt{\eps_n}\le 1/5$ is satisfied for all $n$}. Finally, it can be easily checked by induction that for each $n$, $S_n$ is a product of $9^n$ symplectic matrices. The proof follows after choosing $\delta=\eps_n$ for $n$ large enough.
\end{proof}


\end{document}

%% file: def.tex
\usepackage[latin1]{inputenc}
\usepackage{amsthm}
\usepackage{amssymb}
\usepackage{amsmath}
\usepackage{bbold}
\usepackage{bbm}
\usepackage{nonfloat}
\usepackage{hyperref}
\usepackage{braket}
\usepackage{dsfont}
\usepackage{mathdots}
\usepackage{mathtools}
\usepackage{enumerate}
\usepackage{csquotes}
\usepackage{stmaryrd}
\usepackage[cal=boondox]{mathalfa}
\usepackage{graphicx}
\usepackage{stackengine}
\usepackage{scalerel}
\usepackage[table]{xcolor}
\usepackage{array}
\usepackage{makecell}
\newcolumntype{x}[1]{>{\centering\arraybackslash}p{#1}}
\usepackage{tikz}
\usetikzlibrary{shapes.geometric, arrows, decorations.pathreplacing, angles, quotes}
\usepackage{booktabs}
\usepackage{xfrac}
\usepackage{subcaption}
\usepackage{comment}
\usepackage{xr}

\let\savecorresponds\corresponds
\let\corresponds\relax
\let\saveboxplus\boxplus
\let\boxplus\relax
\usepackage{mathabx}
\let\corresponds\savecorresponds
\let\boxplus\saveboxplus

\newtheorem{thm}{Theorem}
\newtheorem*{thm*}{Theorem}
\makeatletter
\newcommand{\setthmtag}[1]{
  \let\oldthethm\thethm
  \newcommand{\thethm}{#1}
  \g@addto@macro\endthm{
    \addtocounter{thm}{-1}
    \global\let\thethm\oldthethm}
  }
\makeatother

\newtheorem{prop}[thm]{Proposition}
\newtheorem*{prop*}{Proposition}
\newtheorem{lemma}[thm]{Lemma}
\newtheorem*{lemma*}{Lemma}
\newtheorem{cor}[thm]{Corollary}
\newtheorem*{cor*}{Corollary}

\newtheorem*{cj*}{Conjecture}
\newtheorem{Def}[thm]{Definition}
\newtheorem*{Def*}{Definition}

\theoremstyle{definition}
\newtheorem{rem}{Remark}
\newtheorem*{rem*}{Remark}

\newtheorem{ex}{Example}

\newcommand{\bb}{\begin{equation}}
\newcommand{\bbb}{\begin{equation*}}
\newcommand{\ee}{\end{equation}}
\newcommand{\eee}{\end{equation*}}
\newcommand*{\coloneqq}{\mathrel{\vcenter{\baselineskip0.5ex \lineskiplimit0pt \hbox{\scriptsize.}\hbox{\scriptsize.}}} =}
\newcommand*{\eqqcolon}{= \mathrel{\vcenter{\baselineskip0.5ex \lineskiplimit0pt \hbox{\scriptsize.}\hbox{\scriptsize.}}}}
\newcommand\floor[1]{\left\lfloor#1\right\rfloor}

\newcommand{\texteq}[1]{\stackrel{\mathclap{\scriptsize \mbox{#1}}}{=}}
\newcommand{\textleq}[1]{\stackrel{\mathclap{\scriptsize \mbox{#1}}}{\leq}}

\newcommand{\textgeq}[1]{\stackrel{\mathclap{\scriptsize \mbox{#1}}}{\geq}}
\newcommand{\ketbra}[1]{\ket{#1}\!\bra{#1}}

\newcommand{\sumno}{\sum\nolimits}

\newcommand{\emp}[1]{{\em{#1}}}

\newcommand{\RR}{{\mathbb{R}}}
\newcommand{\NN}{\mathbb{N}}
\newcommand{\CC}{\mathbb{C}}

\newcommand{\E}{\mathds{E}}

\DeclareMathOperator{\rk}{rk}

\DeclareMathOperator{\co}{conv}
\DeclareMathOperator{\inter}{int}
\DeclareMathOperator{\Span}{span}
\DeclareMathAlphabet{\pazocal}{OMS}{zplm}{m}{n}

\DeclareMathOperator{\pr}{Pr}

\DeclareMathOperator{\Id}{id}

\DeclareMathOperator{\dom}{Dom}
\DeclareMathOperator{\spec}{sp}
\DeclareMathOperator{\sv}{sv}
\DeclareMathOperator{\tr}{Tr}
\DeclareMathOperator{\Li}{Li}
\DeclareMathOperator{\symp}{Sp}

\newcommand{\lsmatrix}{\left(\begin{smallmatrix}}
\newcommand{\rsmatrix}{\end{smallmatrix}\right)}

\stackMath
\newcommand\xxrightarrow[2][]{\mathrel{%
  \setbox2=\hbox{\stackon{\scriptstyle#1}{\scriptstyle#2}}%
  \stackunder[5pt]{%
    \xrightarrow{\makebox[\dimexpr\wd2\relax]{$\scriptstyle#2$}}%
  }{%
   \scriptstyle#1\,%
  }%
}}

\newcommand{\tendsn}[1]{\xxrightarrow[\! n\rightarrow \infty\!]{\mathrm{#1}}}

\def\indic{\operatorname{1\hskip-2.75pt\relax l}}

\stackMath

\makeatletter

\newcommand{\cN}{\mathcal{N}}
\newcommand{\cH}{{\pazocal{H}}}

\newcommand{\cT}{{\pazocal{T}}}
\newcommand{\cB}{{\pazocal{B}}}
\newcommand{\cD}{{\pazocal{D}}}

\newcommand{\cS}{{\pazocal{S}}}

\newcommand{\cM}{{\mathcal{M}}}
\newcommand{\cU}{\mathcal{U}}
\newcommand{\cV}{\mathcal{V}}

\newcommand{\eps}{{\varepsilon}}

\newcommand{\D}{\mathcal{D}}
\newcommand{\cL}{{\mathcal{L}}}

\makeatletter
\newcommand{\raisemath}[1]{\mathpalette{\raisem@th{#1}}}
\newcommand{\raisem@th}[3]{\raisebox{#1}{$#2#3$}}
\makeatother

\newtheorem{manualthminner}{Theorem}
\newenvironment{manualthm}[1]{%
  \renewcommand\themanualthminner{#1}%
  \manualthminner
}{\endmanualthminner}

\newtheorem{manualcorinner}{Corollary}
\newenvironment{manualcor}[1]{%
  \renewcommand\themanualcorinner{#1}%
  \manualcorinner
}{\endmanualcorinner}